\documentclass[12pt,draftclsnofoot,onecolumn]{IEEEtran}
\usepackage[T1]{fontenc}

\usepackage{amsmath,amsfonts,amssymb,amsthm}
\usepackage{mathrsfs}
\usepackage{dsfont}
\usepackage{bm}
\usepackage{empheq}

\usepackage{graphicx}
\usepackage{array}
\usepackage{booktabs}
\usepackage{multirow}
\usepackage{algorithm}
\usepackage{algorithmic}
\usepackage[caption=false,font=footnotesize]{subfig}
\usepackage{tikz}
\usetikzlibrary{arrows.meta,positioning,calc}
\usepackage{forest}

\usepackage{textcomp}
\usepackage{xcolor}
\usepackage{comment}
\usepackage{verbatim}
\usepackage[shortlabels]{enumitem}

\usepackage{cite}
\usepackage{url}
\usepackage[colorlinks=true,
            urlcolor=blue,
            citecolor=blue,
            linkcolor=blue]{hyperref}

\allowdisplaybreaks

\newtheorem{theorem}{Theorem}[section]
\newtheorem{lemma}[theorem]{Lemma}
\newtheorem{proposition}[theorem]{Proposition}

\theoremstyle{definition}
\newtheorem{definition}[theorem]{Definition}

\theoremstyle{remark}
\newtheorem{remark}[theorem]{Remark}

\definecolor{darker}{RGB}{2,100,20}
\definecolor{ink}{HTML}{172033}
\definecolor{accent}{HTML}{315C8C}
\definecolor{soft}{HTML}{EEF3F8}
\definecolor{rulegray}{HTML}{CAD3DE}

\newcommand{\mc}[1]{\mathcal{#1}}

\newcommand{\mb}[1]{\mathbb{#1}}

\newcommand{\D}{\mathcal{D}}

\newcommand{\cS}{\mathcal{S}}

\newcommand{\X}{\mathcal{X}}
\newcommand{\Y}{\mathcal{Y}}
\newcommand{\V}{\mathcal{V}}

\newcommand{\R}{\mathbb{R}}
\newcommand{\T}{\mathbb{T}}

\renewcommand{\H}{\mathbb{H}}
\renewcommand{\d}{\mathrm{d}}

\newcommand{\rleq}{\trianglelefteq}
\newcommand{\rgeq}{\trianglerighteq}

\begin{document}

\title{A complete characterisation of conditional entropies}

\author{Roberto~Rubboli,
        Erkka~Haapasalo,
        and~Marco~Tomamichel,~\IEEEmembership{Senior Member,~IEEE}%
\thanks{R. Rubboli is with the Department of Mathematical Sciences, University of Copenhagen, Copenhagen 2100, Denmark (e-mail: \texttt{ror@math.ku.dk}). E. Haapasalo and M. Tomamichel are with the Centre for Quantum Technologies, National University of Singapore, Singapore 117543, Singapore. M. Tomamichel is also with the Department of Electrical and Computer Engineering, National University of Singapore, Singapore 117583, Singapore.}}

\maketitle

\begin{abstract}
Entropies are fundamental measures of uncertainty with central importance in
information theory and statistics and applications across all the quantitative
sciences. Under a natural set of operational axioms, the most general form of
entropy is captured by the family of R\'enyi entropies, parameterized by a real
number $\alpha$. Conditional entropy extends the notion of entropy by
quantifying uncertainty from the viewpoint of an observer with access to
potentially correlated side information. However, despite their significance
and the emergence of various useful definitions, a complete characterisation
of measures of conditional entropy that satisfy a natural set of operational
axioms has remained elusive. In this work, we provide a complete
characterisation of conditional entropy, defined through a set of axioms that
are essential for any operationally meaningful definition: additivity for
independent random variables, invariance under relabeling, and monotonicity
under conditional mixing channels. We show that the most general conditional entropy is a generalized convex combination of conditional entropies defined as exponential averages of R\'enyi entropies of the conditioned distribution, parameterized by a real number and a probability measure on the positive reals.

\end{abstract}


\section{Introduction}
\label{sec: introduction}
Entropies are measures of uncertainty about the outcome of a random experiment.  Given a discrete random variable
$X$ on $\mc X$ with probability mass function (pmf)
$P(x)$, the Shannon entropy~\cite{shannon1948mathematical} is\footnote{Unless stated otherwise, \(\log\) denotes a logarithm in an arbitrary fixed base (greater than one), and \(\exp\) denotes its inverse. We write \(\ln\) for the natural logarithm. The choice of base fixes only the units of entropy.}
\begin{equation}
 H(X):=\sum_xP(x)\log\frac{1}{P(x)}.
 \label{eq:shannon-introduction}
\end{equation}
While Shannon entropy plays a central role in information processing by characterizing the ultimate performance limits for a range of tasks, most notably data compression~\cite{shannon1948mathematical}, it is evident that a more fine-grained specification of the surprisal is necessary to understand tail events. Notable examples are the min-entropy,
$H_{\min}(X) = -\log\max_{x \in \mc{X}} P(x)$, which plays a pivotal role in cryptography, and the Hartley entropy~\cite{hartley28},
$H_{\max}(X) = \log \sum_{x \in \mc{X}} \mathbf{1}\{P(x) > 0\}$.
A broader family is given by the R\'enyi entropies~\cite{renyi61}. For
$\alpha\in [0,+\infty]$, they are defined as
\begin{equation}
 H_\alpha(X):=\frac{1}{1-\alpha}
 \log\!\left(\sum_x P(x)^\alpha\right).
 \label{eq:renyi-introduction}
\end{equation}
The Hartley, Shannon, and min-entropies are recovered in the limits
$\alpha \to 0$, $\alpha \to 1$, and $\alpha \to \infty$, respectively. Their ubiquity in information theory leads to the question of whether R\'enyi entropies are complete under some set of well-motivated information-theoretic axioms.

The axiomatic study of entropy has a rich history, beginning with Shannon and Khinchin and including the Faddeev and R\'enyi--Dar\'oczy formulations~\cite{shannon1948mathematical,khintchine1957mathematical,faddeev56,renyi61,daroczy1964}; see also the books by Acz\'el and Dar\'oczy~\cite{aczel75} and Ebanks, Sahoo, and Sander~\cite{ebanks98}, and the review by Csisz\'ar~\cite{csiszar08}. These characterizations share a common feature; they prescribe how uncertainty is to be averaged through a mathematical axiom, such as the generalized mean value axiom, rather than deriving the averaging rule from the structure of an information processing task. 

Motivated by this, recent work~\cite{gour21_axiomatic} proposes a set of axioms that 
any operationally meaningful measure of entropy should satisfy: monotonicity under mixing channels, invariance under embedding, additivity for independent variables, and normalization. (see Definition~\ref{def: conditional entropy}). Only very recently Mu \emph{et al.}~\cite{mu21_renyi} showed that any quantity $\H(X)$ satisfying these axioms can be expressed as a convex combination of R\'enyi entropies.\footnote{More precisely, a one-to-one relationship between entropy and relative entropy is established in~\cite{gour21_axiomatic} and \cite{mu21_renyi} gives a complete characterisation of relative entropy. The statement is also a special case of our main theorem.} Namely,
\begin{equation}
 \H(X)=\int_{[0,+\infty]}H_\alpha(X)\,\d\mu(\alpha)
 \label{eq:EntropyBarycentre}
\end{equation}
for some probability measure $\mu$ on the positive reals.

Conditional entropy measures the uncertainty about $X$ from the viewpoint of
an observer who has access to correlated side information $Y$. Consider thus two discrete random variables $X$ on $\mc{X}$ and $Y$ on $\mc{Y}$ with joint pmf $P(x,y)$. The conditional Shannon entropy of $X$ given $Y$ is
\begin{equation}
 H(X|Y)=\sum_yP(y)H(X|Y\!=\!y).
\end{equation}
Several definitions of conditional R\'enyi entropies have been proposed and found operational significance (see, e.g., \cite{teixeira2012conditional} and \cite{fehr14conditional} for reviews that cover some aspects of conditional entropy).  Two
prominent examples are the Hayashi--Skori\'c and Arimoto families
\cite{hayashi2011exponential,vskoric2011sharp,arimoto1977information}, but the landscape of definitions that have found operational meaning is vast. We discuss these in Section~\ref{sec: literature} as special cases of our general definition.

Instead of listing specific definitions, we follow an axiomatic approach. We consider the following five axioms (see also~\cite{gour2018conditional,gour2024inevitability}), which are motivated by our desire to interpret conditional entropy as a measure of uncertainty from the perspective of an observer with side information.

\begin{definition}[Conditional entropy]
\label{def: conditional entropy}
For any finite random variables $X$ and $Y$, we denote the set of probability distributions $P_{XY}$ by $\mc P_{XY}$. A function
\begin{equation}
\H : \bigcup_{X, Y} \mc P_{XY} \rightarrow \mathbb{R}_+
\label{eq:cond-ent-gen-def-1st}
\end{equation}
is called a conditional entropy if it satisfies the following postulates:
\begin{enumerate}
    \item \textbf{Invariance under embedding.}
Adding outcomes of zero probability to either alphabet does not change conditional entropy. Explicitly, if $\X\subseteq\X'$ and $\Y\subseteq\Y'$ and $Q$ agrees with $P$ on $\X\times\Y$ and is zero elsewhere, we require
\begin{equation}
    \H(X|Y)_P=\H(X'|Y')_Q.
\end{equation}
\label{eq:QCE-2nd-postulate-iso}
    \item \textbf{Invariance under conditional permutations.}
\label{item:conditional-permutations}
For every family of permutations $\{\pi_y:\X\to\X\}_{y\in\Y}$, let $Q(x,y)=P(\pi_y^{-1}(x),y)$. We require
\begin{equation}
    \H(X|Y)_P=\H(X|Y)_Q.
\end{equation}
Thus, relabeling $X$ according to the available side information does not affect uncertainty.
\label{eq:QCE-postulate-conditional-permutations}
    \item \textbf{Monotonicity under processing of side information.}
\label{item:side-information}
For every stochastic map from $Y$ to an arbitrary finite register $Y'$, with transition probabilities $D(y'|y)$, let $Q(x,y')=\sum_yP(x,y)D(y'|y)$. We require
\begin{equation}
    \H(X|Y)_P\leq\H(X|Y')_Q.
\end{equation}
Intuitively, processing the side information $Y$ cannot reduce uncertainty about $X$.
\label{eq:QCE-1st-postulate-mono}
    \item \textbf{Additivity}:   For any joint probability distributions $P$ and $Q$, we require that their product distribution, $P \otimes Q$, satisfies $\H(XX'|YY')_{P\otimes Q} = \H(X|Y)_P + \H(X'|Y')_Q$.
  \label{eq:QCE-3rd-postulate-add}
    \item \textbf{Normalization}: We set $H(X)_P=\log 2$ for $P_X=(1/2,1/2)$, fixing the units.
     \label{eq:QCE-4-postulate-normalization}
\end{enumerate}
\end{definition}

These postulates imply invariance under arbitrary injective relabelings of either alphabet: permutations of $X$ are special cases of conditional permutations, permutations of $Y$ are reversible stochastic maps, and embeddings add outcomes of zero probability.

\smallskip

The axiomatic viewpoint is natural in operational settings: the structure of a task often enforces invariance, additivity, and monotonicity requirements, and thus any quantity used to characterize its performance ought to be compatible with these properties.

Consider lossless source coding with common side information. Here we use the binary logarithm and measure rates in bits. For \(n\) independent source--side-information pairs, the encoder observes \((X^n,Y^n)\) and sends one of at most \(2^{nR}\) messages; the decoder uses the message and \(Y^n\) to reconstruct \(X^n\).
Conditional permutations preserve the optimal error probability: the encoder relabels source symbols according to \(Y^n\), and the decoder uses the same side information to reverse the relabeling. Grouping \(k\) independent pairs into one supersymbol multiplies the description length per symbol by \(k\), motivating additivity. 
For monotonicity, consider processing \(Y\) into \(Y'\). Any code for the processed source can be used for the original source by having the encoder and decoder locally simulate the channel on their common \(Y^n\), using a shared random seed to generate the same \(Y'^n\). The error probability averaged over seeds equals that of the processed-source code, so some seed gives no larger error. Fixing this seed in advance yields a deterministic code for the original source with the same rate and no larger error probability. Thus processing side information cannot improve compression performance.

\smallskip

In this work, we characterize the most general conditional entropy satisfying the five axioms. We show that every such conditional entropy can be expressed as a convex combination of the conditional entropies $H_{t,\tau}(X|Y)$ and their pointwise limits (see Theorem~\ref{th:barycenter conditional entropies}). For any probability measure $\tau$ and $t\in\mathbb{R} \setminus \{0\}$, these have the form
\begin{equation}
 H_{t,\tau}(X|Y)
 =\frac1t\log\sum_{y:P(y)>0}P(y)
 \exp\left(t\int_{[0,+\infty]}
 H_\alpha(X|Y\!=\!y)\ \d\tau(\alpha)\right) ,
 \label{eq:General}
\end{equation}
where necessary and sufficient conditions on \((t,\tau)\) for \(H_{t,\tau}\) to satisfy the axioms are collected in Definition~\ref{def: D}.

\smallskip

The remainder is organised as follows. Section~\ref{sec: set of entropies} reviews unconditional entropies, and Section~\ref{sec: semiring of conditional entropies} develops the semiring framework for conditional majorization. Sections~\ref{sec:sufficient-conditions-measure} and~\ref{sec:necessary-conditions-measures} establish sufficiency and necessity of the parameter conditions. Section~\ref{sec: large sample} gives the large-sample comparison theorem, and Section~\ref{sec: set of conditional entropies} derives the general characterization and discusses known conditional entropies.

\section{The set of entropies}
\label{sec: set of entropies}
In this section, we discuss the general form of an entropy.
In particular, we show that any such entropy can be written as a convex combination of R\'enyi entropies.  We refer to this representation as the barycentric decomposition
of an entropy. This result can be understood as a consequence of the characterisation of R\'enyi divergences given in~\cite[Theorem~2]{mu21_renyi}, together with the equivalence between entropies and relative entropies discussed in~\cite{gour21_axiomatic}. Alternatively, it can be viewed as a consequence of~\cite[Proposition~3.7]{Jensen_Kjaerulf_2019}, which establishes sufficient conditions for transforming multiple copies of one probability distribution into another under mixing channels.
Indeed, the proof strategy underlying~\cite[Theorem~2]{mu21_renyi} is quite general, as emphasized in~\cite{haapasalo_inprep}, and shows that the problem of characterizing conditions for transformations between multiple copies of probability distributions is equivalent to an axiomatic characterisation of monotone and additive quantities. Since we will later adopt this approach for conditional entropies, we first discuss it in detail in the simpler setting without side information. This case is already illustrative and provides the foundation for our general result.

We first introduce a few definitions. The entropy without side information is a special case of the conditional entropy in Definition~\ref{def: conditional entropy}. For clarity, we repeat this special case here. 

An entropy is a function $\H:\bigcup_X\mc P_X\to\R_+$ satisfying:
\begin{enumerate}[itemsep=1pt]
 \item 
 \label{eq:E-2nd-postulate-iso}
 \textbf{Invariance under relabeling.} We require that relabeling outcomes or embedding them
into a larger alphabet will not change the entropy.
 
 \item \textbf{Monotonicity under mixing.}
 \label{eq:E-1st-postulate-mono}Mixing channels mapping $P_X$ to $Q_{X'}$ act as $Q=SP$, where $S$ is a doubly stochastic matrix. We require entropy to be monotone under such channels, namely, $H(P_X)\leq H(Q_{X'})$ whenever $Q$ is obtained from $P$ in this way.

 \item
  \label{eq:E-2nd-postulate-add}\textbf{Additivity.} We require that entropies for independent random variables are additive, namely
 $\H(P_X\otimes Q_{X'})=\H(P_X)+\H(Q_{X'})$.
 \item \textbf{Normalization.}
 \label{eq:E-4-postulate-normalization}
 We normalize the entropy of a fair coin toss to $\log 2$, i.e., we set $H(P_X)=\log 2$ for $P=(1/2,1/2)^{\mathsf T}$.
\end{enumerate}

The postulates above share several features with classical axiomatizations of entropy. Relabeling invariance and embedding invariance correspond, respectively, to symmetry and expansibility, while normalization and additivity under independent composition also appear in the axiomatizations of Shannon and Khinchin, Faddeev, and R\'enyi and Dar\'oczy~\cite{khintchine1957mathematical,faddeev56,renyi61,daroczy1964,aczel75,csiszar08}. These approaches are nevertheless not equivalent to ours. The characterizations due to Shannon and Khinchin and to Faddeev typically assume continuity together with a grouping or branching rule, whereas the approach of R\'enyi and Dar\'oczy relies on a generalized mean value axiom. We assume no continuity, averaging rule, or prescribed functional form. Accordingly, a general mean in the sense of Kolmogorov and Nagumo is compatible with our postulates only when the resulting entropy belongs to the barycentric class characterized below. Otherwise, at least one of our operational postulates must fail.

Axiom~\ref{eq:E-1st-postulate-mono} is also closely related to the standard notions of majorization and Schur concavity.

\begin{definition}[Majorization]
Given two probability distributions $P_X$ and $Q_{X'}$, we say that $P_X$ majorizes $Q_{X'}$, written as $ P_X \succeq Q_{X'}$
if there exist probability distributions \(\widetilde{P}_{\widetilde{X}}\) and \(\widetilde{Q}_{\widetilde{X}'}\), obtained from \(P_X\) and \(Q_{X'}\) via embedding and relabeling, and a mixing channel represented by a doubly stochastic matrix \(D\) such that $D \widetilde{P}_{\widetilde{X}} = \widetilde{Q}_{\widetilde{X}'}$.
\end{definition}
By the classical Hardy--Littlewood--P\'olya theorem, the above
operational definition of majorization is equivalent, after
padding the distributions with zeros to a common dimension,
to the initial sum condition stated below
\cite{marshall1979inequalities}.
The related Birkhoff--von Neumann theorem states that every
doubly stochastic matrix is a convex combination of permutation
matrices \cite{birkhoff1946tres}. Thus, mixing by a doubly
stochastic matrix can be interpreted as a probabilistic mixture
of relabelings.
\begin{lemma}[Initial sum condition]
\label{lem: Initial sum condition}
Let $P_X,Q_X\in\R^d_{\geq0}$, and let $P_X^\downarrow,Q_X^\downarrow$ denote
their entries in nonincreasing order.  Then $P_X\succeq Q_X$ if and only if $\sum_{i=1}^kP_X^\downarrow(i)\geq
 \sum_{i=1}^kQ_X^\downarrow(i)$ for all $1\leq k<d$
and the total sums agree.
\end{lemma}
Axiom~\ref{eq:E-1st-postulate-mono} is precisely the requirement of Schur concavity with respect to the majorization preorder. Doubly stochastic mixing makes a distribution more disordered in this preorder, and its entropy therefore cannot decrease~\cite[Problem~1.10]{csiszar11}~\cite{marshall1979inequalities}. In this manuscript, we state the axiom directly in terms of the mixing transformation because this formulation makes both the allowed operation and its operational interpretation explicit.

Next, we state the result concerning the transformation of multiple copies of probability distributions under mixing channels.
\begin{lemma}[{\cite[Proposition 3.7]{Jensen_Kjaerulf_2019}}]\label{lemma:Jensen}
Let $P_X,Q_X$ be two probability distributions and assume that
\begin{equation}
    H_\alpha(P_X) < H_\alpha(Q_X)  \quad \forall \alpha\in[0,+\infty] \,.
\end{equation}
Then, for sufficiently large $n \in \mathbb{N}$ it holds that
$P_X^{\otimes n} \succeq Q_X^{\otimes n}$.
\end{lemma}
We now characterize the most general entropy satisfying the above axioms. We denote by $\mathcal B([0,+\infty])$ the Borel
$\sigma$-algebra of the one-point compactification $[0,+\infty]$
of $[0,+\infty)$.
\begin{proposition}[General form of entropies]\label{proposition:barycenter entropies}
Let $\H$ be an entropy. Then, there exists a probability measure $\mu:\mc B\big([0,+\infty]\big)\to\mb R_+$ such that for all probability distributions $P_X$
\begin{equation}\label{eq:barycenter}
\H(P_X)=\int_{[0,+\infty]}H_\alpha(P_X)\d\mu(\alpha) \,.
\end{equation}
Furthermore, if the normalization axiom~\ref{eq:E-4-postulate-normalization} is not assumed, the conclusion remains valid upon replacing the probability measure with a finite positive measure.
\end{proposition}
\begin{proof}
    The result follows from Lemma \ref{lemma:Jensen}; the proof is essentially the same as that of \cite[Theorem 2]{mu21_renyi} and this result can be seen as a special case of Theorem 7 of \cite{haapasalo_inprep}. We give an outline of this proof for completeness. This same proof also works for Theorem~\ref{th:barycenter conditional entropies} later in this work with slight adaptations.

    The first step is to show that $H_\alpha(P_X)\leq H_\alpha(Q_X)$ for all $\alpha\in[0,+\infty]$ implies $\mb H(P_X)\leq\mb H(Q_X)$. First, let us assume that $H_\alpha(P_X)< H_\alpha(Q_X)$ for all $\alpha\in[0,+\infty]$. According to Lemma \ref{lemma:Jensen}, this means that, for any $n\in\mb N$ sufficiently large, there is a doubly stochastic matrix $S_n$ such that $S_nP_X^{\otimes n}=Q_X^{\otimes n}$. Thus, since $\mb H$ is monotone and additive,
    \begin{equation}
    n\mb H(P_X)=\mb H(P_X^{\otimes n})\leq\mb H(Q_X^{\otimes n})=n\mb H(Q_X),
    \end{equation}
    i.e., $\mb H(P_X)\leq\mb H(Q_X)$. Let us then assume that $H_\alpha(P_X)\leq H_\alpha(Q_X)$ for all $\alpha\in[0,+\infty]$. Denoting by $R$ the two-outcome uniform distribution $(1/2,1/2)$, it follows that
    \begin{equation}
    H_\alpha(P_X^{\otimes m})=mH_\alpha(P_X)\leq mH_\alpha(Q_X)<mH_\alpha(Q_X)+\log 2=H_\alpha(Q_X^{\otimes m}\otimes R)
    \end{equation}
    for all $m\in\mb N$ and $\alpha\in[0,+\infty]$. Thus, according to our initial step, we also have
    \begin{equation}
    m\mb H(P_X)=\mb H(P_X^{\otimes m})\leq\mb H(Q_X^{\otimes m}\otimes R)=m\mb H(Q_X)+\log 2
    \end{equation}
    for all $m\in\mb N$. This means $\mb H(P_X)\leq\mb H(Q_X)+(\log 2)/m$ for all $m\in\mb N$, so that $\mb H(P_X)\leq\mb H(Q_X)$, and the initial result is proven.

    For the details of the rest of the proof, we refer to the proof of Theorem 7 of \cite{haapasalo_inprep} and that of Theorem 2 of \cite{mu21_renyi}. The strategy is to define a positive linear functional on the set $C\big([0,+\infty]\big)$ of continuous real functions on the compact set $[0,+\infty]$ using the initial step proven above and then use the representation result for such functionals.

    Consider the evaluation functions ${\rm ev}_P\in C\big([0,+\infty]\big)$, ${\rm ev}_P(\alpha)=H_\alpha(P)$ at any finite probability distribution $P$ and all $\alpha\in[0,+\infty]$. By `evaluation' we essentially mean that we actually identify $[0,+\infty]$ with the set of R\'{e}nyi entropies $H_\alpha$, $\alpha\in[0,+\infty]$. This kind of identification has to be done in the general case of Theorem 7 of \cite{haapasalo_inprep} where the situation is such that large-sample ordering conditions are given by a compact set of functions which is Hausdorff. This is also the case of Theorem \ref{th:barycenter conditional entropies} later on. It can be shown very easily, using our initial step, that we may set up a well-defined positive functional $F$ on the positive cone generated by the evaluation functions in $C\big([0,+\infty]\big)$ through $F({\rm ev}_P)=\mb H(P)$. For the well-definedness, note that if $P$ and $Q$ are such that ${\rm ev}_P={\rm ev}_Q$, i.e., $H_\alpha(P)=H_\alpha(Q)$ for all $\alpha\in[0,+\infty]$, then, according to our initial result, also $\mb H(P)=\mb H(Q)$. The positivity of $F$ follows similarly. With some extra work, we may show that $F$ extends into a positive functional $G$ on the vector space $V$ spanned by the evaluation functions.

    Suppose that $f\in C\big([0,+\infty]\big)$ and fix any $n\in\mb N$ such that $n\log 2\geq\|f\|_\infty$. Denoting by $1$ the constant function 1 on $[0,+\infty]$ and by $R$ the two-outcome uniform distribution, we have
    \begin{equation}
    f\leq \|f\|_\infty\,1\leq n(\log 2)\,1=n\,{\rm ev}_R={\rm ev}_{R^{\otimes n}},
    \end{equation}
    so that all functions in $C\big([0,+\infty]\big)$ are upper bounded by functions from $V$. Now a theorem due to Kantorovich \cite{Kantorovich} related to the Hahn-Banach theorem implies that $G$ extends into a positive linear functional $I:C\big([0,+\infty]\big)\to\R$. Due to the Riesz-Markov-Kakutani theorem, there is a finite positive measure $\mu:\mc B\big([0,+\infty]\big)\to\R_+$ (the Borel $\sigma$-algebra associated with the one-point compactification of the natural topology of $[0,\infty)$) such that
    \begin{equation}
    I(f)=\int_{[0,+\infty]}f\,\mathrm{d}\mu,\qquad\forall f\in C\big([0,+\infty]\big).
    \end{equation}
    This means
    \begin{equation}
    \mb H(P)=F({\rm ev}_P)=I({\rm ev}_P)=\int_{[0,+\infty]}{\rm ev}_P(\alpha)\,\mathrm{d}\mu(\alpha)=\int_{[0,+\infty]}H_\alpha(P)\,\mathrm{d}\mu(\alpha).
    \end{equation}
    Denoting again by $R$ the two-outcome uniform distribution, we have
    \begin{equation}
    \log 2=\mb H(R)=\int_{[0,+\infty]}\underbrace{H_\alpha(R)}_{=\log 2}\,\mathrm{d}\mu(\alpha)=(\log 2)\mu\big([0,+\infty]\big),
    \end{equation}
    i.e., $\mu$ is a probability measure.
\end{proof}

\section{The semiring of conditional majorization}
\label{sec: semiring of conditional entropies}
In this section, we introduce the main technical tool used to establish sufficient conditions for multiple-copy transformation (large sample result) that leads to the general form of conditional entropies (barycentric decomposition). Our approach relies on the real-algebraic theory of preordered semirings developed in \cite{fritz2023abstract,fritz2023abstractII}. This framework is particularly powerful because it provides a large sample theorem (Theorem \ref{thm:Fritz2022}) whenever one can endow the underlying problem with a suitable semiring structure.
The large sample comparison theorem provides a characterisation of the sufficient conditions under which multiple copies of one element of the semiring can be transformed into multiple copies of another.
In our setting, the semiring is essentially the set of joint probability distributions $P_{XY}$, and the large-sample theorem yields necessary and sufficient conditions under which multiple copies of a joint distribution $P_{XY}$ can be transformed into multiple copies of another joint distribution $Q_{XY}$ under conditionally mixing operations in terms of the so-called monotone homomorphisms (see later for the definition) of the theory. In the case of conditional entropies, the monotone homomorphisms and derivations are exactly the extremal conditional entropies (up to a constant and a logarithm). These are the conditional entropies in terms of which any general conditional entropy can be expressed as a convex combination. From here, a standard argument \cite{mu21_renyi,haapasalo_inprep} allows us to deduce from the large sample theorem that any conditional entropy can be expressed as a convex combination of these extremal conditional entropies.

We begin by introducing the fundamental concepts underlying the theory of semirings and the general results that will be employed in our manuscript in Section~\ref{sec: background semiring}. We then construct the semiring associated with conditional entropies and derive the general form of monotone homomorphisms and derivations for this particular setting in Sections~\ref{sec: all monotone homomorphisms}, and~\ref{sec: derivations}.

\subsection{Background on preordered semirings and a \textit{Vergleichsstellensatz}}
\label{sec: background semiring}
We define a preordered semiring as a tuple $S = (\mathcal{S}, +, \cdot, 0, 1, \rleq)$, where $\mathcal{S}$ is a set equipped with binary operations of addition $+$ and multiplication $\cdot$, a zero element $0 \in \mathcal{S}$, a multiplicative unit $1 \in \mathcal{S}$, and a preorder relation $\rleq$  (a reflexive and transitive binary relation) defined on $\mathcal{S}$  satisfying
\begin{equation}
    x \rleq y \ \Rightarrow \
\begin{cases}
x + a \ \rleq\ y + a, \\[0.3em]
x a \ \rleq\ y a,
\end{cases}
\end{equation}
for all $a \in \mathcal{S}$.
Moreover,
$(\mathcal{S}, +, 0)$ and $(\mathcal{S}, \cdot, 1)$ are commutative semigroups, and the multiplication distributes over the addition.
For compactness, we omit the multiplication dot when writing products in the semiring.
We write $x \sim y$ for the equivalence relation induced by $\rleq$, meaning that $x \sim y$ if and only if there exist $z_1, \ldots, z_n \in \mathcal{S}$ such that
\begin{equation}
    x \rleq z_1 \rgeq z_2 \rleq \cdots \rgeq z_n \rleq y.
\end{equation}
A preordered semiring $S$ has polynomial growth if it admits a power universal element $u \in \mathcal{S}$, that is,
\begin{equation}
    x \rleq y \quad \Rightarrow \quad \exists\, k \in \mathbb{N}:\ y \rleq x u^k.
\end{equation}
We say that a preordered semiring $S$ is a preordered semidomain if
\begin{equation}
    \begin{aligned}
&xy = 0 \ \Rightarrow\ x = 0 \ \text{or} \ y = 0, \\
&0 \rleq x \rleq 0 \ \Rightarrow\ x = 0.
\end{aligned}
\end{equation}
Finally, $S$ is zerosumfree if $x + y = 0$ implies $x = 0 = y$.

Given preordered semirings $S$ and $T$, we say that a map $\Phi:\cS\to \mathcal{T}$ is a {\it monotone homomorphism} if it satisfies the properties
\begin{enumerate}
\item Additivity: $\Phi(x+y)=\Phi(x)+\Phi(y)$ for all $x,y\in \cS$,
\item Multiplicativity: $\Phi(xy)=\Phi(x)\Phi(y)$ for all $x,y\in \cS$,
\item Monotonicity: $x\rleq y$ $\Rightarrow$ $\Phi(x)\rleq\Phi(y)$, and
\item $\Phi(0)=0$ and $\Phi(1)=1$.
\end{enumerate}
We call a monotone homomorphism \textit{degenerate} if $ x \rleq y \quad \Rightarrow \quad \Phi(x) = \Phi(y)$. Otherwise, it is called \textit{nondegenerate}.
We need to consider monotone homomorphisms with values in certain special semirings, namely:
\begin{enumerate}
    \item $\mathbb{R}_+$: the half-line $[0,+\infty)$ with the usual addition, multiplication, and total order.
    \item $\mathbb{R}_+^{\mathrm{op}}$: the same set with reversed order.
    \item $\mathbb{T}\mathbb{R}_+$: the half-line $[0,+\infty)$ with the usual multiplication, order, and the tropical sum $x + y := \max\{x,y\}$.
    \item $\mathbb{T}\mathbb{R}_+^{\mathrm{op}}$: the same as for $\mathbb{T}\mathbb{R}_+$, but with reversed order.
\end{enumerate}
The term \emph{temperate reals} refers to the pair $(\mathbb{R}_+, \mathbb{R}_+^{\mathrm{op}})$, while \emph{tropical reals} denotes the pair $(\mathbb{T}\mathbb{R}_+, \mathbb{T}\mathbb{R}_+^{\mathrm{op}})$.
Given a monotone homomorphism $\Phi: \cS \to \mathbb{R}_+$, an additive map $\Delta: \cS \to \mathbb{R}$ is called a derivation at $\Phi$ (or a $\Phi$-derivation) if it satisfies the Leibniz rule
\begin{equation}
    \Delta(xy) = \Delta(x)\,\Phi(y) + \Phi(x)\,\Delta(y)
\end{equation}
for all $x, y \in \cS$.
In this work, we are interested in $\Phi$-derivations at degenerate homomorphisms that are also monotone, i.e.,
\begin{equation}
    x \rleq y \quad \Rightarrow \quad \Delta(x) \leq \Delta(y).
\end{equation}
The following large-sample theorem belongs to a series of results collectively
known as the ``Vergleichsstellens\"{a}tze'' \cite{fritz2023abstractII}.
\begin{theorem}[Based on Theorem 8.6 in \cite{fritz2023abstractII}]\label{thm:Fritz2022}
Let $S$ be a zerosumfree preordered semidomain with a power universal element $u$. Assume that for some $d\in\mb N$ there is a surjective homomorphism $\|\cdot\|:\cS\to\R_{>0}^d\cup\{(0,\ldots,0)\}$ with trivial kernel and such that
\begin{equation}\label{eq:surjectivehomomorphismproperties}
a\rgeq b\ \Rightarrow\ \|a\|=\|b\|\quad {\rm and} \quad \|a\|=\|b\|\ \Rightarrow\ a\sim b.
\end{equation}
Denote the component homomorphisms of $\|\cdot\|$ by $\|\cdot\|_{(j)}$, $j=1,\ldots,d$. Let $x,y\in S\setminus\{0\}$ with $\|x\|=\|y\|$. If
\begin{itemize}
\item[(i)] for every $\mb K\in\{\R_+,\R_+^{\rm op},\T\R_+,\T\R_+^{\rm op}\}$ and every nondegenerate monotone homomorphism $\Phi : \cS \to \mb K$ with trivial kernel, we have $\Phi(x) > \Phi(y)$, and
\item[(ii)] $\Delta(x) > \Delta(y)$ for every monotone $\|\cdot\|_{(j)}$-derivation $\Delta : \cS \to \R$ with $\Delta(u) = 1$ for all component indices $j = 1,\ldots,d$,
\end{itemize}
and if additionally $x$ is power universal, then
$x^n\rgeq y^n$ for all sufficiently large $n\in\mb N$.
Conversely, if $x^n\rgeq y^n$ for some $n\in\mb N$, then the inequalities in
items (i) and (ii) above hold non-strictly.
\end{theorem}

\subsection{Definition of the conditional majorization semiring}
In this section, we define the semiring of conditional majorization. In particular, we need to specify the elements of the tuple $S=(\mathcal{S}, +, \cdot, 0, 1, \rgeq)$.

We identify a joint distribution $P$ with the matrix whose entries are $(P)_{x,y}=P(x,y)$, with rows indexed by $\X$ and columns indexed by $\Y$.

A \emph{conditionally mixing channel} is any finite composition of stochastic maps acting only on the side information and permutations of $X$ controlled by that side information, allowing arbitrary finite intermediate side-information registers. The following lemma gives an equivalent matrix representation.

\begin{lemma}
\label{lem:elementary-conditional-mixing}
A channel mapping an $m\times n$ joint probability matrix $P$ to an $m\times n'$ matrix $Q$ is conditionally mixing if and only if it can be written as
\begin{equation}
\label{eq:intro/mixing}
    P\mapsto Q=\sum_{i=1}^k S_iPD_i,
\end{equation}
where each $S_i$ is a nonnegative $m\times m$ doubly stochastic matrix (its rows and columns sum to one), each $D_i$ is a nonnegative $n\times n'$ sub-stochastic matrix (its row sums are at most one), and $\sum_iD_i$ is row-stochastic. Consequently, Items~\ref{item:conditional-permutations} and~\ref{item:side-information} of Definition~\ref{def: conditional entropy} are equivalent to monotonicity under every conditionally mixing channel.
\end{lemma}
\begin{proof}
We first show that every channel of the form~\eqref{eq:intro/mixing} is conditionally mixing by expressing it as a composition of the two elementary operations. By the Birkhoff--von Neumann theorem, write $S_i=\sum_jp_{ij}\Pi_{ij}$, where the $\Pi_{ij}$ are permutation matrices, $p_{ij}\geq0$, and $\sum_jp_{ij}=1$. First apply a stochastic map on $Y$ with output $\widetilde Y=(Y',I,J)$ and transition probabilities
\begin{equation}
    T\bigl((y',i,j)\mid y\bigr)=D_i(y,y')p_{ij}.
\end{equation}
These probabilities sum to one for each $y$ because $\sum_iD_i$ is row-stochastic. Apply $\Pi_{ij}$ to $X$, controlled by $(i,j)$, and then discard $I,J$. The resulting joint distribution is
\begin{equation}
    \sum_{i,j}p_{ij}\Pi_{ij}PD_i=\sum_iS_iPD_i=Q.
\end{equation}
Both adjoining the indices and discarding them are stochastic maps on the side information.

We now prove the reverse implication: every conditionally mixing channel has the form~\eqref{eq:intro/mixing}. A stochastic map on $Y$ has this form with $S_1=I$. A conditional permutation has this form with $S_y=\Pi_y$ and $D_y$ the diagonal matrix selecting column $y$, so that $\sum_yD_y=I$. The class~\eqref{eq:intro/mixing} is closed under composition: the left factors are products of doubly stochastic matrices, and the sum of the right factors is the product of their row-stochastic sums. Hence every finite composition of the elementary operations has the required matrix form.

Finally, Items~\ref{item:conditional-permutations} and~\ref{item:side-information} of Definition~\ref{def: conditional entropy} imply monotonicity under composition. Conversely, monotonicity under all conditionally mixing channels gives the side-information postulate and, by applying a conditional permutation and its inverse, the conditional-permutation equality.
\end{proof}

\begin{figure}[ht]
  \centering
  \resizebox{0.5\columnwidth}{!}{%
  \begin{tikzpicture}[
    x=1cm,
    y=1cm,
    vnode/.style={
      draw=accent,
      very thick,
      circle,
      fill=white,
      minimum size=9mm,
      inner sep=1pt
    },
    operation/.style={
      draw=rulegray,
      thick,
      rounded corners=2mm,
      fill=soft,
      align=center,
      minimum height=10mm,
      inner xsep=6pt,
      inner ysep=3pt
    },
    arrow/.style={
      -{Latex[length=2.2mm,width=1.5mm]},
      thick,
      color=accent
    },
    labeltext/.style={
      font=\small,
      color=ink,
      inner sep=1pt
    }
  ]

    \node[vnode] (X) at (0,1.7) {\(X\)};
    \node[vnode] (Y) at (0,0) {\(Y\)};

    \node[operation] (mix) at (5.0,1.7) {
      mixing\\
      \(S_i\)
    };

    \node[operation] (instr) at (3.4,0) {
      side-information processing\\
      \(\{D_i\}_{i=1}^{k}\)
    };

    \node[vnode] (Xout) at (7.7,1.7) {\(X\)};
    \node[vnode] (Yp) at (7.7,0) {\(Y'\)};

    \draw[arrow] (X.east) -- (mix.west);
    \draw[arrow] (Y.east) -- (instr.west);
    \draw[arrow] (instr.east) -- (Yp.west);

    \draw[arrow]
      (instr.north)
      -- (3.4,0.72)
      -- node[labeltext, above=2pt, midway] {\(i\)} (5.0,0.72)
      -- (mix.south);

    \draw[arrow] (mix.east) -- (Xout.west);

  \end{tikzpicture}%
  }

  \caption{A conditional mixing operation. The matrices \(D_i\) jointly process the shared side information \(Y\) into \(Y'\) and select the index \(i\); the corresponding matrix \(S_i\) mixes \(X\). The resulting transformation has the matrix form~\eqref{eq:intro/mixing}.}
  \label{fig:conditional-mixing}
\end{figure}

We first introduce the symbol $\succeq$, which will be used throughout the text to denote the preorder relation of the conditional majorization semiring.
\begin{definition}[Conditional majorization]
\label{def:cond-maj-cond-mix}
Given two probability distributions $P=P_{XY}$ and $Q=Q_{X'Y'}$, we say that $P_{XY}$ conditionally majorizes $Q _{X'Y'}$, and write 
\begin{equation}
    P_{XY}\succeq Q_{X'Y'}
        \label{eq:cond-maj-def-1}
\end{equation}
if there exist probability distributions $\widetilde{P}_{\widetilde{X}\widetilde{Y}}$ and $\widetilde{Q}_{\widetilde{X}'\widetilde{Y}'}$, obtained from \(P_{XY}\) and \(Q_{X'Y'}\) via embedding and relabeling, and a conditionally mixing channel that maps $\widetilde{P}_{\widetilde{X}\widetilde{Y}}$ into $\widetilde{Q}_{\widetilde{X}'\widetilde{Y}'}$.
This means that, up to relabeling and embedding, $P_{XY}$ can be transformed into $Q_{X'Y'}$ using a conditional mixing channel.
\end{definition}

The transformations entering this definition also connect conditional majorization with several established operational orders. Taken together, the family $\{D_i\}$ defines a Markov kernel from $Y$ to the augmented output $(i,Y')$, and the branch $i$ controls the doubly stochastic map applied to $X$. Ordinary stochastic processing of $Y$ is recovered by taking the source map to be the identity, so monotonicity under conditionally mixing channels contains the usual data-processing requirement. If $X$ is regarded as the state and its prior is held fixed, this processing of $Y$ is exactly the garbling operation used to compare statistical experiments~\cite{blackwell1953,torgersen1991comparison}. The full conditional operation is also related to classical comparison and degradation orders for communication channels~\cite{shannon1958partial,nasser2018channelorderings}; see also the historical discussion in~\cite[Problems~6.16--6.18]{csiszar11}. The notions are not equivalent, however. Classical channel and experiment orders compare the corresponding kernels, whereas our preorder compares particular joint source--side-information distributions and restricts the processing of the source to side-information-dependent mixing. Up to conventions, this transformation is the classical-memory conditional-majorization operation studied in~\cite{gour2018conditional}. The new result is therefore neither the Markov-kernel formalism nor the conditional-majorization relation itself, but the semiring analysis of this operational preorder and the complete characterization of its additive monotones.

\begin{enumerate}

\item \textbf{The set of elements ``$\cS$"}. The elements of the semiring are, essentially, all (not normalized) joint probability distributions $P_{XY}$.
However, since distributions that differ only by embedding into a larger outcome space and/or permutation of outcomes are considered equivalent, the set $\cS$ is formally defined as the set of equivalence classes of joint probability distributions under this similarity relation.
We make this precise in the following.

We can represent (unnormalized) joint distributions as matrices, with rows corresponding to the outcomes of the variable $X$ and columns to those of $Y$. For example, if $X$ has $d$ outcomes and $Y$ has $n$ outcomes, then $P_{XY}$ can be represented as a $d \times n$ matrix
\begin{equation}
P_{XY} = \begin{pmatrix}
P(x_1,y_1) & \hdots & P(x_1,y_n) \\
\vdots &  & \vdots \\
P(x_d,y_1) & \hdots & P(x_d,y_n)
\end{pmatrix} \,.
\end{equation}
We also denote the set of outcomes $\X=\{x_1,...,x_d\}$ and $\Y=\{y_1,...,y_n\}$.
We denote with $\mathcal{V}_{d \times n}$ the set of all $d \times n$ matrices with non-negative entries. Furthermore, define
\begin{equation}
    \mathcal{V} := \bigcup_{d,n \in \mathbb{N}} \mathcal{V}_{d \times n}
\end{equation}
as the set of all matrices with non-negative entries of arbitrary finite dimension.

We say that two joint probability distributions $P_{XY}$ and $Q_{X'Y'}$ are similar, denoted
\begin{equation}
    P_{XY} \approx Q_{X'Y'}  \,,
\end{equation}
if one can be obtained from the other by permuting rows or columns, and/or inserting or deleting rows or columns consisting entirely of zeros.
In other words, $P_{XY}$ and $Q_{X'Y'}$ represent the same distribution up to reordering of outcomes and inclusion or exclusion of impossible events. In yet other words, there is $R_{X''Y''}$ such that both $P_{XY}$ and $Q_{X'Y'}$ can be relabeled and embedded into $R_{X''Y''}$ when we lift the notion of relabeling and embedding to non-normalized distributions.
For $P \in \mathcal{V}$, we define
    $[P] := \{ P' \in \mathcal{V} \mid P \approx P' \}$
as the equivalence class of $P$ under the relation $\approx$. The elements of the semiring are then given by the quotient set
\begin{equation}
    \mathcal{S} := \mathcal{V} / \approx,
\end{equation}
which consists of all equivalence classes of $\mathcal{V}$ under the similarity relation $\approx$.

\item \textbf{The addition ``$+$"}.
Addition in $\mathcal{S} = \mathcal{V}/\!\approx$ is defined via the direct sum of matrices. For two joint distributions $P_{XY}$ and $Q_{X'Y'}$, their direct sum is given by
\begin{equation}
    P_{XY} \oplus Q_{X'Y'} := \begin{pmatrix}
P_{XY} & 0 \\
0 & Q_{X'Y'}
\end{pmatrix}.
\end{equation}
Hence, the sum of two equivalence classes $[P_{XY}]$ and $[Q_{X'Y'}]$ in $\mathcal{S}$ is defined as the equivalence class of the direct sum of any representatives of the respective classes
\begin{equation}
    [P_{XY}] + [Q_{X'Y'}] := [P_{XY} \oplus Q_{X'Y'}] \,.
\end{equation}

\item \textbf{The multiplication ``$\cdot$"}.
Multiplication in $\mathcal{S} = \mathcal{V}/\!\approx$ is defined via the tensor product (or Kronecker product) of two matrices. For $P_{XY} \in \V_{d \times n}$, $Q_{X'Y'} \in \V_{d' \times n'}$
\begin{equation}
P_{XY} \otimes Q_{X'Y'} = \begin{pmatrix}
P(x_1,y_1)Q_{X'Y'} & \hdots & P(x_1,y_n) Q_{X'Y'} \\
\vdots &  & \vdots \\
P(x_d,y_1) Q_{X'Y'} & \hdots & P(x_d,y_n) Q_{X'Y'}
\end{pmatrix} \,,
\end{equation}
where the product of a scalar and a matrix is defined as entry-wise (i.e., element-wise) multiplication. For example,
\begin{equation}
P(x_1,y_1)Q_{X'Y'} = \begin{pmatrix}
P(x_1,y_1)Q(x'_1,y'_1) & \hdots & P(x_1,y_1) Q(x'_1,y'_{n'}) \\
\vdots &  & \vdots \\
P(x_1,y_1) Q(x'_{d'},y'_1) & \hdots & P(x_1,y_1) Q(x'_{d'},y'_{n'})
\end{pmatrix} \,.
\end{equation}
Accordingly, the multiplication of two equivalence classes $[P_{XY}]$ and $[Q_{X'Y'}]$ in $\mathcal{S}$ is defined as the equivalence class of the tensor product of any representatives of the respective classes
\begin{equation}
    [P_{XY}] \cdot [Q_{X'Y'}] := [P_{XY} \otimes Q_{X'Y'}] \,.
\end{equation}

\item \textbf{The zero element ``$0$"}. We define a zero element $0$ in $\V/\!\approx$ as the equivalence class of the $1\times 1$ matrix $0$.

\item \textbf{The unit element ``$1$"}.  We define the unit element $1$ as the equivalence class of the $1\times 1$ matrix $1$.

\item \textbf{The preorder ``$\rgeq$"}.
\label{item: preorder}
We write
\begin{equation}
\label{eq: relation preorder}
\quad[Q_{X'Y'}]\rgeq[P_{XY}].
\end{equation}
if $P_{XY}$ conditionally majorizes $Q_{X'Y'}$, that is, $P_{XY} \succeq Q_{X'Y'}$ in the sense of Definition~\ref{def:cond-maj-cond-mix} when we lift this definition also for non-normalized distributions. Note that in the latter relation, the roles of $P$ and $Q$ are reversed compared to those in equation~\eqref{eq: relation preorder}.
\end{enumerate}

\subsection{Invariance of the homomorphisms under shifting}
\label{sec: invariance embedding and shifting}
In this section, we prove that monotone homomorphisms are invariant under column-wise shifts. This property will be crucial in characterizing the general form of such homomorphisms, which in turn determines the general form of conditional entropies.

\bigskip
We begin by introducing some notation. Let $0_{n \times 1}$ denote the column vector of zeros of dimension $n$, $0_{1 \times n}$ the corresponding row vector, and $0_{n \times n}$ the $n \times n$ zero matrix. We denote by $E^{j \times j}_{n \times n}$ the matrix of size $n \times n$ with a 1 in the $(j, j)$-th entry and zeros elsewhere. Finally, $I_{n \times n}$ stands for the $n \times n$ identity matrix.

The next lemma states that the columns can be shifted without affecting the value of any homomorphism, as there exists a channel implementing the transformation in both directions.
\begin{lemma}[Shifting]\label{lemma:shifting}
    Let $P_{XY} \in \V_{d\times n}$. Moreover,  for each $i=1,...,n$, let $P_{X,Y=y_i}$ be the $i$-th column of $P_{XY}$, that is, $P_{XY} = (P_{X,Y=y_1},...,P_{X,Y=y_n})$. Define
    \begin{equation}
\widetilde{P}_{\widetilde{X}Y}=\bigoplus_{i=1}^n P_{X,Y=y_i}= \begin{pmatrix}
P_{X,Y=y_1} & 0_{d \times 1}  & \hdots & 0_{d \times 1} \\
0_{d \times 1} & P_{X,Y=y_2} & \hdots & 0_{d \times 1} \\
\vdots &  & \vdots & \vdots\\
0_{d \times 1} & \hdots & 0_{d \times 1} & P_{X,Y=y_n}\,
\end{pmatrix} \,,
\end{equation}
where $\widetilde{X}=X\sqcup\cdots\sqcup X$ is the $n$-fold separate union of $X$. We have
\begin{equation}
 \begin{pmatrix}
        P_{XY}\\
        0_{d(n-1) \times n}
    \end{pmatrix}\succeq\widetilde{P}_{\widetilde{X}Y}\succeq  \begin{pmatrix}
        P_{XY}\\
        0_{d(n-1) \times n}
    \end{pmatrix}.
\end{equation}
\end{lemma}

\begin{proof}
Let us first prove that
\begin{equation}
    \begin{pmatrix}
        P_{XY}\\
        0_{d(n-1) \times n}
    \end{pmatrix}\succeq\widetilde{P}_{\widetilde{X}Y}\,.
\end{equation}
For this case, the proof follows by first measuring the $Y$ register and, depending on the outcome, performing an appropriate relabeling of the outcomes in the $X$ register.
 Explicitly, let us define $D^{(j)}= E^{j \times j}_{n \times n}$ which selects the outcome $y_j$ of the $Y$ register. For example, for $j=2$, we have
\begin{equation}
    D^{(2)}=\begin{pmatrix}
        0 & 0 & 0_{1 \times (n-2)}  \\
       0 & 1 & 0_{1 \times (n-2)} \\
       0_{(n-2) \times 1} & 0_{(n-2) \times 1} & 0_{(n-2) \times (n-2)}
    \end{pmatrix} \,.
\end{equation}
The matrix $S^{(j)}$ relabels the outcomes on the $X$ register. For example, for $j=2$, we have
\begin{equation}
    S^{(2)}=\begin{pmatrix}
        0_{d \times d} & I_{d \times d} & 0_{d \times d(n-2)} \\
        I_{d \times d} & 0_{d \times d} & 0_{d \times d(n-2)} \\
        0_{d(n-2) \times d} & 0_{d(n-2) \times d} & I_{d(n-2) \times d(n-2)}
    \end{pmatrix} \,.
\end{equation}
To prove the other direction, we use the same matrices as before to select the appropriate outcome of $Y$ and then apply the same permutation to restore the original ordering of the entries.
\end{proof}

\begin{remark}\label{rem:SomOfColumns}
Lemma \ref{lemma:shifting} essentially tells us that, for a matrix $P_{XY}\in\V_{d \times n}$ with columns $P_{X,Y=y_i}$, we may write
\begin{equation}
[P_{XY}]\rgeq[P_{X,Y=y_1}]+\cdots+[P_{X,Y=y_n}]\rgeq [P_{XY}] \,.
\end{equation}
This will subsequently aid us in finding all the monotone homomorphisms of $S$ into $\mb K\in\{\mb R_+,\mb R_+^{\rm op},\mb{TR}_+,\mb{TR}_+^{\rm op}\}$.
\end{remark}

\subsection{Some important properties of the semiring \texorpdfstring{$S$}{S}}
To use the large-sample results in Theorem~\ref{thm:Fritz2022}, we must first show that the semiring is of polynomial growth, i.e., it contains a power-universal element. Since the large sample results apply only when the input element is power-universal, we also need to determine exactly which elements have this property. Lastly, we must construct a surjective degenerate homomorphism $\|\cdot\|:\cS\to\R_{>0}^d\cup\{0\}$ with the properties required in Theorem~\ref{thm:Fritz2022}. We begin with the last point. In the following lemma, we show that such a homomorphism can be chosen as the total weight, given by the sum of the matrix entries associated with elements of the semiring.
Finally, note that the total weight is a scalar quantity rather than a vector, and therefore there exists only a single component homomorphism of \(\|\cdot\|\) in Theorem~\ref{thm:Fritz2022}, namely \(j = 1\).
\begin{lemma}
\label{lem: degenerate and zig-zag}
In the preordered semiring of conditional majorization $S$, the map
\begin{equation}
\|[P_{XY}]\| = \sum_{x\in\mc X}\sum_{y\in\mc Y} P_{XY}(x,y) \,,
\end{equation}
where $P_{XY}$ can be chosen to be any representative of the equivalence class $[P_{XY}]$,
defines a surjective homomorphism $\|\cdot\| : \cS \mapsto \mathbb{R}_+$ with trivial kernel. Moreover, it satisfies
\begin{equation}
[P] \rgeq [Q] \quad \Rightarrow \quad \|[P]\| = \|[Q]\| \quad \Rightarrow\quad  [P] \sim [Q]\,.
\end{equation}
\end{lemma}
\begin{proof}
Clearly, the total weight is a homomorphism since $\|[P]\cdot [Q]\|=\|[P\otimes Q]\|=\|[P]\|\|[Q]\|$ and $\|[P]+[Q]\|=\|[P\oplus Q]\|=\|[P]\|+\|[Q]\|$ for all $[P],[Q]\in\cS$.

The first direction, namely that $[P] \rgeq [Q] \Rightarrow \|[P]\| = \|[Q]\|$, follows immediately, since conditionally mixing operations preserve the total weight $\|P\|$ of any element $P$. This also shows that the total weight is a degenerate homomorphism since it is invariant under the preorder $\rgeq$.

Let us now turn to the implication $\|[P]\|=\|[Q]\|\Rightarrow [P]\sim [Q]$.  We show that if we define $\|P\|=\|Q\| = c$, then $P\preceq c\succeq Q$
and hence $[P]\sim [Q]$. In particular, we show that from the element $[c]$ we can generate both $P$ and $Q$. This is straightforward by considering the matrix with a single non-zero entry $c$ in the top-left corner. Indeed, starting from this matrix, we can construct both $P$ and $Q$. For example, let us focus on $P_{XY} \in \V_{d\times n}$.
Then the steps are the following. Let us denote with $P_{Y}$ the row vectors with components $P_Y(y_i) = \sum_jP_{XY}(x_j,y_i)$.
First, we create for each outcome $Y$ a column whose first entry is $P_Y(y_i)$ and all other entries are zero. This is done by multiplying on the r.h.s by a row stochastic matrix.
\begin{align}
\begin{pmatrix}
     P_{Y} \\
    0_{(d-1) \times n}
\end{pmatrix}=
\begin{pmatrix}
    \|P_{XY}\| & 0_{1\times (n-1)} \\
    0_{(d-1)\times 1} & 0_{(d-1)\times (n-1)}
\end{pmatrix}
\begin{pmatrix}
    \|P_{XY}\|^{-1} P_{Y} \\
    \vdots \\
    \|P_{XY}\|^{-1} P_{Y}
\end{pmatrix}\,.
\end{align}
Next, we measure the register $Y$ and, conditioned on the outcome $y_i$, we transform the deterministic column (with a single non-zero entry) into a more distributed column $P_{X,Y=y_i}$. Let us denote with $P_{X,Y=y_1}$ the column with components $P_{X,Y=y_1}(x_j)=P_{X,Y}(x_j,y_1)$ and with $P_{X|Y=y_1}$ the column with components $P_{X|Y=y_1}(x_j)=P_{X,Y}(x_j,y_1)/P_Y(y_1)$,
For the first outcome $y_1$, we obtain
\begin{align}
\begin{pmatrix}
     P_{X,Y=y_1} & 0_{d \times (n-1)}
\end{pmatrix}=
\begin{pmatrix}
    P_{X|Y=y_1} & \star_{d\times (d-1)}
\end{pmatrix}
\begin{pmatrix}
     P_{Y} \\
    0_{(d-1) \times n}
\end{pmatrix}
\begin{pmatrix}
    1 & 0_{1\times (n-1)}\\
    0_{(n-1) \times 1} & 0_{(n-1)\times (n-1)}
\end{pmatrix}
\end{align}
Here, $\star_{d\times (d-1)}$ is any entry that completes the matrix to a doubly stochastic one. By summing over all the outcomes $y_i$ for all $i=1,..,n$, we obtain the desired probability $P_{XY}$. This shows that overall, we have
\begin{equation}
    P_{XY} = \sum_{i=1}^n S^{(i)} P DD^{(i)} \,,
\end{equation}
where
\begin{align}
P=\begin{pmatrix}
    \|P_{XY}\| & 0_{1\times (n-1)} \\
    0_{(d-1)\times 1} & 0_{(d-1)\times (n-1)}
\end{pmatrix} \,, \quad
    D=\begin{pmatrix}
    \|P_{XY}\|^{-1} P_{Y} \\
    \vdots \\
    \|P_{XY}\|^{-1} P_{Y}
\end{pmatrix} \,,
\end{align}
and for example, for $i=1$, we have
\begin{align}
D^{(1)}=\begin{pmatrix}
    1 & 0_{1\times (n-1)}\\
    0_{(n-1) \times 1} & 0_{(n-1)\times (n-1)}
\end{pmatrix} \,,\quad
S^{(1)} = \begin{pmatrix}
    P_{X|Y=y_1} & \star_{d\times (d-1)}
\end{pmatrix}  \,.
\end{align}

\end{proof}

As mentioned above, to apply the large-sample result of
Theorem~\ref{thm:Fritz2022}, we must identify a power universal element.  This
also establishes polynomial growth.  Within the semiring of conditional
majorization, power universality means that the $X$ register has nonzero
randomness for every outcome of $Y$.

\begin{proposition}
\label{prop: power universals}
The following statements are equivalent for $[R_{XY}] \in \cS$:
\begin{enumerate}
    \item \( [R_{XY}] \in \cS \) is power universal.
    \label{item: power universal}
     \item  \( \|[R_{XY}]\| = 1 \), and every non-zero column \( R_{X,Y=y} \) of any representative \( R_{XY} \) contains at least two non-zero entries.
    \label{item: two non-zero entries}
    \item There exists \( r \in (0,1) \) such that
    \begin{equation}\label{eq:pucond}
    \begin{pmatrix}
        r \\
        1 - r
    \end{pmatrix}
    \succeq R_{XY}.
    \end{equation}
    \label{item: majorization R_r}
\end{enumerate}
Consequently, the semiring $S$ is of polynomial growth.
\end{proposition}

\begin{proof}
We prove separately the implications $\ref{item: power universal} \implies \ref{item: two non-zero entries}$, $\ref{item: two non-zero entries} \implies \ref{item: majorization R_r}$ and $\ref{item: majorization R_r} \implies \ref{item: power universal}$. This shows that the statements are all equivalent.
In the following, for brevity, we simply write $P_{XY}$ to refer to a representative of the equivalence class $[P_{XY}]$, without always stating it explicitly. Throughout this proof, we assume that $R_{XY}\in\mc V_{d\times n}$ (i.e.\ $|\mc X|=d$ and $|\mc Y|=n$) and we number $\mc X$ and $\mc Y$ explicitly to aid our matrix calculations.

\bigskip

Let us start with the implication $\ref{item: power universal} \implies \ref{item: two non-zero entries}$. The fact that \( \|[R_{XY}]\| = 1 \) is immediate, since conditional majorization preserves the total weight of elements. If instead \( \|[R_{XY}]\| \neq 1 \), the transformation would not be possible. Next, we show that if $R_{XY}$ contains a column with only one non-zero entry, then it cannot be a power universal. The key idea is that if \( R_{XY} \) contains a column with only one non-zero entry, then any tensor power \( R_{XY}^{\otimes m} \) will also contain a column of the same form. However, starting from a mixed probability distribution, it is impossible to generate a deterministic distribution conditioned on the event that \( Y \) corresponds to such a column. This follows from the fact that mixing channels cannot produce a deterministic distribution from a non-deterministic one; instead, they only increase the level of mixing.
Explicitly, let us define \( P_{XY} = (1/2, 1/2) \) and take \( Q_{XY} = 1 \) to be the single unit element. Note that we use the subscript $XY$ for both, as the smaller matrix can always be embedded into the space of the larger one by padding with zeros.
Observe that \( Q_{XY} \succeq P_{XY} \).  Assume now that \( R_{XY} \) is power universal. Then, there exists a sufficiently large \( m \in \mathbb{N} \) such that
\begin{equation}
    P_{XY} \succeq Q_{XY} \otimes R_{XY}^{\otimes m} = R_{XY}^{\otimes m}\,.
\end{equation}
This implies that, for a matrix \( R_{XY} \in \V_{d \times n} \), there exist an integer \( k \in \mathbb{N} \), doubly stochastic matrices \( S^{(1)}, \ldots, S^{(k)}\), and (without loss of generality---since the zero entries in the embedding will eliminate the effect of the remaining components) single-row matrices \( D^{(i)} = (d^{(i)}_1, \ldots, d^{(i)}_{n^m}) \),
such that an embedding of the vector \( P_{XY} \) can be transformed into \( R_{XY}^{\otimes m} \)
\begin{align}\label{eq:apu3}
R_{XY}^{\otimes m}&=\sum_{i=1}^k S^{(i)}
\begin{pmatrix}
    \frac{1}{2}\\
    \frac{1}{2} \\
    0_{(d^m-2)\times 1}
\end{pmatrix}
\big(d^{(i)}_{1}\ \cdots\ d^{(i)}_{n^m}\big)=\sum_{i=1}^k S^{(i)}
\begin{pmatrix}
    \frac{1}{2}d^{(i)}_{1}\ \cdots\ \frac{1}{2}d^{(i)}_{n^m} \\
   \frac{1}{2} d^{(i)}_{1}\ \cdots\ \frac{1}{2}d^{(i)}_{n^m} \\
   0_{(d^m-2) \times 1}
\end{pmatrix} \,.
\end{align}
We now distinguish two cases: either \( d_j^{(i)} = 0 \) or \( d_j^{(i)} > 0 \). In both situations, the column of \( R_{XY}^{\otimes m} \) with a single non-zero entry cannot be generated by any column of the matrix \( P_{XY} D^{(i)} \). Indeed, if \( d_j^{(i)} = 0 \), the corresponding column is identically zero; if \( d_j^{(i)} > 0 \), then the resulting column contains at least two non-zero entries of the form \( \frac{1}{2} d_j^{(i)} \), and a doubly stochastic matrix cannot reduce this uncertainty to produce a deterministic outcome. This follows from the fact that both the rows and columns of a doubly stochastic matrix must sum to one, thus increasing entropy rather than eliminating it.

\bigskip
Let us now turn to the implication $\ref{item: two non-zero entries} \implies \ref{item: majorization R_r}$
Let us first assume that $\|R_{XY}\|=1$ and $R_{XY}$ is such that its columns have supports with at least two points. We need to find $r\in(0,1)$ such that the condition in item (3) holds.
The key idea is that each column of \( R_{XY} \) containing at least two non-zero entries can be generated from a vector of the form \( (r, 1 - r) \), for some sufficiently large \( r \in (0,1) \). Such a vector represents an almost deterministic distribution, which---after suitable embedding---can be transformed, via mixing, into any non-deterministic distribution with the same support. Furthermore, the same vector \( (r, 1 - r) \) can be used to generate all columns of \( R_{XY} \), since we are free to enlarge the system \( Y \) to accommodate additional outcomes if needed. Explicitly,
because of the restriction on $R_{XY}$, we have $R_{XY}(x_j,y_i)<R_Y(y_i)$ for all $x_j\in \mathcal{X}$ and $y_i\in \mathcal{Y}$ such that the column $R_{X,Y=y_i}$ is non-zero. Therefore, there exists a $r\in[1/2,1)$ such that
\begin{equation}
R_Y(y_i)r\geq R_{XY}(x_j,y_i)\qquad\forall(x_j,y_i)\in \mathcal{X}\times \mathcal{Y}:\ R_{X,Y=y_i}\neq0.
\end{equation}
By Lemma~\ref{lem: Initial sum condition}, it follows that for each \( y_i \in \mathcal{Y} \) with \( R_{X,Y=y_i} \neq 0 \), the two-dimensional embedded vector $\widetilde{R}_Y(y_i) := \big(r R_Y(y_i), (1 - r) R_Y(y_i)\big)$
majorizes the column \( R_{X,Y=y_i} \).
Hence, for each \( y_i \in \mathcal{Y} \) with \( R_{X,Y=y_i} \neq 0 \), there exists a doubly stochastic matrix \( S^{(i)} \) and an embedding of the vector \( \widetilde{R}_Y(y_i) \), denoted by \( \widehat{R}_Y(y_i) \), such that $S^{(i)} \widehat{R}_Y(y_i) = R_{X,Y=y_i}$.
To isolate the outcome \( y_i \in \mathcal{Y} \), we define the matrices \( D^{(i)} \) with a $1$ in the \( (y_i, y_i) \)-th position and zeros elsewhere. Treating \( R_Y \) as a row-stochastic vector and indexing the outcomes of $Y$ as \( i = 1,..,n \), we obtain
\begin{align}
\left(\begin{array}{c}
r\\1-r
\end{array}\right)&\succeq\left(\begin{array}{c}
r\\1-r
\end{array}\right)R_Y\\
&=\left(\begin{array}{ccc}
rR_Y(y_1)&\cdots&rR_Y(y_n)\\
(1-r)R_Y(y_1)&\cdots&(1-r)R_Y(y_n)
\end{array}\right)\\
&\succeq\sum_{i=1}^n S^{(i)}\left(\begin{array}{ccc}
\widehat{R}_Y(y_1)\cdots \widehat{R}_Y(y_n)
\end{array}\right)D^{(i)}\\
&=\big(S^{(y_1)}\widehat{R}_Y(y_1)\cdots S^{(y_n)}\widehat{R}_Y(y_n)\big)\\
&=\big(R_{X,Y=y_1}\cdots R_{X,Y=y_n}\big)\\
&=R_{XY}.
\end{align}
Hence, we obtain what we wanted to prove.

\bigskip

Finally, let us prove the implication $\ref{item: majorization R_r} \implies \ref{item: power universal}$. We need to show that if there exists $r\in(0,1)$ such that $R_r\succeq R_{XY}$, where $R_r$ is the matrix (column) appearing on the left-hand side of \eqref{eq:pucond},
then $[R_{XY}]$ is power universal. To establish this, we first show that $[R_r]$ is a power universal element; that is, for any elements $[P_{XY}]$ and $[Q_{XY}]$ such that $\|[P_{XY}]\| = \|[Q_{XY}]\|$, there exists a sufficiently large integer $m \in \mathbb{N}$ such that
$P_{XY} \succeq Q_{XY} \otimes R_r^{\otimes m}
$. This, in turn, implies that [$R_{XY}]$ is also power universal, since $R_r\succeq R_{XY}$ implies the chain of inequalities
\begin{equation}
    P_{XY} \succeq Q_{XY} \otimes R_r^{\otimes m}\succeq Q_{XY} \otimes R_{XY}^{\otimes m}\,.
\end{equation}
Note that we use the subscript $XY$ for both $P_{XY}$ and $Q_{XY}$, since the smaller matrix can always be embedded into the larger one by zero-padding.
Without loss of generality, we may assume that $r \in [1/2, 1)$, as the case $r \in (0, 1/2)$ lies in the same equivalence class due to the ability to permute elements.
The fact that $[R_r]$ is power universal follows from the observation that, given a sufficiently large number of copies of $R_r^{\otimes m}$, the non-zero entries of the resulting vector become arbitrarily small. In particular, the largest entry is $r^m$, which approaches zero as $m \to \infty$.
By the initial sum condition for majorization (see Lemma~\ref{lem: Initial sum condition}), this implies that any column of a suitably large embedded distribution $P_{X,Y=y_i}$ can be transformed into $Q_{X,Y=y_i} \otimes R_r^{\otimes m}$, provided the corresponding weights match; that is, when $P_Y(y_i) = Q_Y(y_i)$.
If the column weights differ, we can apply an additional stochastic map to $P_{XY}$ to adjust the marginals accordingly. Explicitly, let $D$ be a row-stochastic matrix such that
\begin{equation}
    P_{Y}D=Q_Y
\end{equation}
and denote $\widetilde{Q}_{XY}:=P_{XY}D$. It is easy to notice that the columns $\widetilde{Q}_{X,Y=y}$ have the same weight as the columns $Q_{X,Y=y}$ for each $y$.
Moreover,  let $m\in\mb N$ be large enough so that
\begin{equation}
r^m\max_{x \in\mc X}Q_{XY}(x,y)\leq \min_{\substack{x \in\mc X \\ \widetilde{Q}(x, y) \neq 0}} \widetilde{Q}(x, y) \qquad\forall y\in\mc Y.
\end{equation}
It then follows, by Lemma~\ref{lem: Initial sum condition}, that for every $y \in\mc Y$, the column $\widetilde{Q}_{X,Y=y}$ majorizes the column $Q_{X,Y=y} \otimes R^{\otimes m}_r$.
Let $S^{(i)}$ be a doubly stochastic matrix such that
\begin{equation}
    S^{(i)} \widehat{Q}_{X,Y=y_i} = Q_{X,Y=y_i} \otimes R_r^{\otimes m} \,,
\end{equation}
where we denoted with $\widehat{Q}_{X,Y}$ a suitable embedding of $\widetilde{Q}_{XY}$.
Furthermore, let $D^{(i)}$ denote the matrix with a 1 at entry $(y_i, y_i)$ and zeros elsewhere.
We now have
\begin{align}
P_{XY}&\succeq P_{XY}D=\widetilde{Q}_{XY}\succeq\sum_{i=1}^kS^{(i)}\widehat{Q}_{XY}D^{(i)} =  Q_{XY}\otimes R_r^{\otimes m}\,.
\end{align}
Finally, the semiring has polynomial growth, as it evidently contains at least one power-universal element.

\end{proof}
\subsection{The homomorphisms}
\label{sec: all monotone homomorphisms}
In this section, we derive the general form of conditional entropy satisfying the axioms stated in Definition~\ref{def: conditional entropy}. Within the theory of preordered semirings, monotone homomorphisms and derivations are functions designed to satisfy analogous properties. In particular, for the semiring of conditional majorization, these functions correspond---up to an additive constant and a logarithmic factor---to the extremal conditional entropies, which we later denote by $H_{t,\tau}$. As we prove later, these entropies generate the entire family of conditional entropies via convex combinations (see Theorem~\ref{th:barycenter conditional entropies}). Hence, for our purposes, it is crucial to determine the form of the homomorphisms and derivations, which is achieved by characterizing them through their mathematical properties.

As we reviewed in Section~\ref{sec: background semiring}, we need to consider monotone homomorphisms \(\Phi: S \to \mathbb{K}\) in the fields $
    \mathbb{K} = \mathbb{R}_+,\mathbb{R}_+^{\mathrm{op}},\mathbb{TR}_+,\mathbb{TR}_+^{\mathrm{op}} \,.$
The cases $\mathbb{R}_+,\mathbb{R}_+^{\mathrm{op}}$ correspond to the temperate case, whereas $\mathbb{TR}_+,\mathbb{TR}_+^{\mathrm{op}} $ correspond to the tropical cases.
As we show later, the temperate cases correspond to the conditional entropy $H_{t,\tau}$ (up to a constant and a logarithmic factor), with the probability measure $\tau$ and $t>0$ on $\mathbb{R}_+$ and $t<0$ on $\mathbb{R}_+^{\mathrm{op}}$ (see Definition~\ref{def: extremal entropies}). The tropical cases yield the limits $t\to+\infty,-\infty$. In particular, $\mathbb{K} = \mathbb{TR}_+^{\rm op}$ gives rise to the entropies $H_{-\infty,\tau}$ (see Definition~\ref{def: extremal entropies}), corresponding to the limiting form of $H_{t,\tau}$ as $t\to-\infty$.
Finally, the case \(\mathbb{TR}_+\) gives rise to \(H_{+\infty,0}\), which corresponds to the limiting form of \(H_{t,\tau}\) as \(t \to +\infty\) and to the choice of \(\tau\) as a point measure at \(\alpha = 0\) (see Definition~\ref{def: extremal entropies}). As discussed later, in this last limiting regime, the convexity properties of the function enforce that only point measures supported at \(\alpha = 0\) are admissible.
In the following proposition, we use the term ``candidate'' to emphasize that we characterize the possible forms of monotone homomorphisms without yet determining which parameter values actually satisfy monotonicity. Thus, every monotone homomorphism must belong to the family below, while the monotonicity conditions derived later further restrict the admissible parameters. For convenience, we state the result here in terms of a measure \(\mu\); later, we write \(\mu=\tau t\), where \(t\) is the total mass of \(\mu\) and \(\tau\) is the corresponding probability measure.

\begin{proposition}
\label{prop: monotone homomorphisms}
Let \(S\) be the semiring of conditional majorization. The candidate nondegenerate
monotone homomorphisms with trivial kernel
$\Phi:\mc S\to\mb K\in\{\R_+,\R_+^{\rm op},\T\R_+,\T\R_+^{\rm op}\}$ are of the form
\begin{equation}
\label{eq:FinalForm}
\Phi(P_{XY}) =
\begin{cases}
\displaystyle \sum_{y \in \mathcal{Y}} P_Y(y) \exp\Bigg(\int_{[0,+\infty]} H_\alpha(P_{X|Y=y}) \, \mathrm{d}\mu(\alpha) \Bigg), & \mb K\in\{\mathbb{R}_+,\mathbb{R}_+^{\mathrm{op}}\}, \\
\displaystyle \max \limits_{\substack{y \in \mathcal{Y} \\ P_Y( y) > 0}} \exp\Bigg(\int_{[0,+\infty]} H_\alpha(P_{X|Y=y}) \, \mathrm{d}\mu(\alpha) \Bigg), &\mb K=\mathbb{TR}_+^{\mathrm{op}}, \\
\displaystyle \max \limits_{\substack{y \in \mathcal{Y} \\ P_Y( y) > 0}} \exp\left(k H_0(P_{X|Y=y}) \, \right), &\mb K=\mathbb{TR}_+ \,,
\end{cases}
\end{equation}
where $\mu$ is a finite positive measure in case $\mb K=\mathbb{R}_+$, and it is a negative measure in cases $\mb K\in\{\mathbb{R}_+^{\mathrm{op}},\T\R_+^{\rm op}\}$. Moreover, in case $\mb K=\mathbb{TR}_+$,  $k$ is a positive constant.
\end{proposition}

\begin{proof}
The proof is organized into several key steps that lead to the general form of the monotone homomorphisms. Specifically, we proceed as follows:
\begin{enumerate}
    \item Step 1: Additivity. We use the shifting lemma~\ref{lemma:shifting} to reduce the problem to determining the monotone homomorphism for a fixed column vector \( P_{X,Y=y} \), corresponding to a specific value \( Y = y \).
    \item Step 2: Multiplicativity. We then use the property that monotone homomorphisms must be multiplicative under tensor products to further simplify the problem to normalized conditional distributions \( P_{X|Y=y} \).
    \item Step 3: Inner barycenter. We invoke Proposition~\ref{proposition:barycenter entropies}, which states that the most general form of an entropy can be expressed as a linear combination of R\'enyi entropies.
     \item Step 4: Simplifying $\mathbb{TR}_+$. We use the fact that, in this case, the monotone homomorphisms are strongly quasi-concave, which allows us to further simplify their form and to single out point measures supported at $\alpha=0$.
\end{enumerate}

\bigskip
\bigskip

\noindent\textbf{Step 1: Additivity}. The first step is to use the additivity of the homomorphism under direct sums to reduce the problem from the matrix $P_{XY}$ to a single column vector $P_{X,Y=y}$ for a fixed value $Y = y$.
Lemma~\ref{lemma:shifting}, together with the fact that the homomorphism is additive in both the temperate and tropical sense under summation (see also Remark~\ref{rem:SomOfColumns}), leads to the following additive form
\begin{equation}\label{eq:additivity}
\Phi(P_{XY}) =
\begin{cases}
\sum_{y \in \Y} \Phi(P_{X,Y=y}), & \mathbb{R}_+,\mathbb{R}_+^{\mathrm{op}}, \\
\max_{y \in \Y} \Phi(P_{X,Y=y}), & \mathbb{TR}_+,\mathbb{TR}_+^{\mathrm{op}} \,.
\end{cases}
\end{equation}
Hence, the problem reduces to finding the homomorphisms for the column vector $P_{X,Y=y}$ with fixed $Y=y$.

\bigskip
\bigskip
 \noindent\textbf{Step 2: Multiplicativity}. The second step applies the multiplicativity of the homomorphism under tensor products, specifically scalar multiplication, to reduce the problem from the unnormalized vector $P_{X,Y=y}$ to the normalized vector probability distribution $P_{X|Y=y} = P_{X,Y=y}/P_Y(y)$, where $P_Y(y) = \sum_x P_{XY}(x,y)$ is the total weight of the column $P_{X,Y=y}$. We then obtain
 \begin{equation}\label{eq:normalize}
     \Phi(P_{X,Y=y}) = \Phi\left(P_Y(y) \otimes \frac{P_{X,Y=y}}{P_Y(y)}\right) =\Phi(P_Y(y)) \Phi\left(\frac{P_{X,Y=y}}{P_Y(y)}\right) = \Phi(P_Y(y)) \Phi\big(P_{X|Y=y}\big) \,.
 \end{equation}

Next, we analyze the form of $\Phi(P_Y(y))$, which corresponds to the homomorphism applied to the single matrix entry $P_Y(y) \in \mb R_+$.
By definition, the homomorphism satisfies $\Phi(0) = 0$ for the zero element and $\Phi(1) = 1$ for the unit element.
Moreover, from the multiplicativity of $\Phi$, it follows that $\Phi(xy)=\Phi(x)\Phi(y)$ for all $x,y\in\mb R_+$. Thus, if the restriction of $\Phi$ to $\mathbb{R}_+$ satisfies a suitable regularity condition---such as being bounded on an interval with non-empty interior---it follows that $\Phi(x) = x^\gamma$ for some $\gamma \in \mathbb{R}$.
 We shall next show that $\Phi(x)\leq 1$ for all $x\in[0,1]$. Let $x\in[0,1]$. We easily see that
\begin{equation}\label{eq:x1}
\left(\begin{array}{cc}
x&0\\
0&1-x
\end{array}\right)\succeq\left(\begin{array}{cc}
1&0\\
0&0
\end{array}\right)\succeq\left(\begin{array}{cc}
x&0\\
0&1-x
\end{array}\right) \,.
\end{equation}
The first implication follows from the fact that we can apply a mixing channel to the second column, conditioned on the second outcome of the register $Y$, and then merge both outcomes into a single $Y$ outcome. The second implication follows by applying a stochastic map on $Y$ that discards the second outcome and creates two new outcomes with probabilities $x$ and $1-x$. Finally, we permute the outcomes of the $X$ register conditioned on the second outcome of the $Y$ register. Since the homomorphism must be monotone under both actions, equality must hold for both elements, and thus, for the temperate case, the sum property implies that
\begin{equation}
1=\Phi(1)=\Phi\left[\left(\begin{array}{cc}
x&0\\
0&1-x
\end{array}\right)\right]=\Phi(x)+\Phi(1-x)\geq\Phi(x) \,,
\end{equation}
since $\Phi(1-x)\geq 0$. The tropical case is similar since $\max\{\Phi(x),\Phi(1-x)\} \geq \Phi(x)$. Hence, the homomorphisms are bounded and, for what was discussed earlier, $\Phi(x)=x^\gamma$ for all $x>0$ for some $\gamma\in\mb R$. Next, we show that $\gamma = 1$ in the temperate cases $\mathbb{R}_+$ and $\mathbb{R}^{\mathrm{op}}_+$, while $\gamma = 0$ in the tropical cases $\mathbb{TR}_+$ and $\mathbb{TR}^{\mathrm{op}}_+$.
By substituting $x = \frac{1}{2}$, and using the equivalence established in \eqref{eq:x1} along with the (temperate or tropical) additivity, we obtain
\begin{equation}
1 = \Phi(1) =
\begin{cases}
2\,\Phi\!\left(\tfrac{1}{2}\right) = 2^{1-\gamma}, & \quad \mathbb{R}_+,\mathbb{R}_+^{\mathrm{op}}, \\
\Phi\!\left(\tfrac{1}{2}\right) = 2^{-\gamma}, & \quad \mathbb{TR}_+,\mathbb{TR}_+^{\mathrm{op}} \,.
\end{cases}
\end{equation}

This establishes that $\gamma = 1$ in the temperate cases $\mathbb{R}_+$ and $\mathbb{R}^{\mathrm{op}}_+$, and $\gamma = 0$ in the tropical cases $\mathbb{TR}_+$ and $\mathbb{TR}^{\mathrm{op}}_+$. In the above arguments, we implicitly assumed that $P_{Y}(y)>0$. In the situation where $P_Y(y)=0$ we have that $\Phi(P_{X,Y=y}) =0$ since $\Phi(0)=0$. This implies that in the cases $\mathbb{TR}_+$ and $\mathbb{TR}^{\mathrm{op}}_+$, only the $Y$ outcomes such that $P_Y(y)>0$ must be included in the maximization. Substituting the values of $\gamma$ into \eqref{eq:additivity}, we obtain
\begin{equation}\label{eq:general1}
\Phi(P_{XY}) =
\begin{cases}
\sum_{y \in \Y} P_Y(y)\,\Phi\!\left(P_{X|Y=y}\right), & \mathbb{R}_+,\mathbb{R}_+^{\mathrm{op}}, \\
\max \limits_{\substack{y \in \mathcal{Y} \\ P_Y( y) > 0}} \Phi\!\left(P_{X|Y=y}\right), & \mathbb{TR}_+,\mathbb{TR}_+^{\mathrm{op}} \,.
\end{cases}
\end{equation}

 \bigskip
 \bigskip

\noindent \textbf{Step 3: Inner barycenter}. The next step is to determine the form of the homomorphism for a normalized column vector $P_{X|Y=y}$. Using multiplicativity, for two column vectors $P_X$ and $Q_{X'}$, we have that
\begin{equation}
    \Phi(P_X \otimes Q_{X'}) = \Phi(P_X) \Phi(Q_{X'}) \,.
\end{equation}
Furthermore, by monotonicity under mixing channels, we have that for any vector $P_X$ and any doubly stochastic matrix $S$,
\begin{equation}
\begin{cases}
 \Phi(P_X)\leq \Phi(SP_X), & \quad \mathbb{R}_+,\mathbb{TR}_+, \\
\Phi(P_X) \geq \Phi(SP_X), & \quad \mathbb{R}_+^{\mathrm{op}},\mathbb{TR}_+^{\mathrm{op}}\,.
\end{cases}
\end{equation}
Since $\Phi$ has trivial kernel, $\Phi(P_X)>0$ for every probability distribution $P_X$, and hence $P_X\to \log{\Phi(P_X)}$ is well defined.
The latter function is additive under tensor product and monotone under mixing channels. Using Proposition~\ref{proposition:barycenter entropies}, we obtain that the function $\log{\Phi(P_X)}$ is a linear combination of R\'enyi entropies. This leads to the form
\begin{equation}
\label{eq:FinalForm proof}
\Phi(P_{XY}) =
\begin{cases}
\displaystyle \sum_{y \in \mathcal{Y}} P_Y(y) \exp\Bigg(\int_{[0,+\infty]} H_\alpha(P_{X|Y=y}) \, \mathrm{d}\mu(\alpha) \Bigg), & \mathbb{R}_+,\mathbb{R}_+^{\mathrm{op}}, \\
\displaystyle \max \limits_{\substack{y \in \mathcal{Y} \\ P_Y( y) > 0}} \exp\Bigg(\int_{[0,+\infty]} H_\alpha(P_{X|Y=y}) \, \mathrm{d}\mu(\alpha) \Bigg), & \mathbb{TR}_+,\mathbb{TR}_+^{\mathrm{op}}.
\end{cases}
\end{equation}
where $\mu$ is a positive measure in cases $\mathbb{R}_+$ and $\mathbb{TR}_+$, and it is a negative measure in the opposite cases $\mathbb{R}^{\mathrm{op}}_+$ and $\mathbb{TR}^{\mathrm{op}}_+$.

\bigskip
\bigskip

\noindent \textbf{Step 4: Simplifying $\mathbb{TR}_+$.}
For two probability vectors with the same support, each can be written as a
nontrivial mixture containing the other. Strong quasi-concavity therefore
forces a tropical monotone to take the same value on both vectors. It depends
only on support, so its inner entropy is $H_0$. Consequently,
\begin{equation}
    \Phi(P_{XY}) = \max \limits_{\substack{y \in \mathcal{Y} \\ P_Y( y) > 0}}
    \exp\!\bigl(kH_0(P_{X|Y=y})\bigr), \qquad k>0.
\end{equation}
\end{proof}

\subsection{The derivations}
\label{sec: derivations}
Besides nondegenerate homomorphisms, the spectral theorem requires monotone
derivations at $\|\cdot\|$. Algebraically, these arise from the limiting family
$H_{0,\alpha}$ as $t\to0$. Again, we use the term ``candidate'' to emphasize that we characterize the possible forms of monotone derivations without yet determining which values of \(\alpha\) actually satisfy monotonicity. Thus, every monotone derivation must belong to the family below, while the monotonicity conditions derived later further restrict the admissible range of \(\alpha\).
\begin{proposition}
\label{prop: monotone derivations}
Let \(S\) be the semiring of conditional majorization. The candidate normalized extremal monotone derivations at \(\|\cdot\|\), as defined in Lemma~\ref{lem: degenerate and zig-zag}, are of the form \begin{equation} \label{eq:ExtDerivation} H_{0,\alpha}(P) = \sum_{y \in \mathcal{Y}} P_Y(y) \, H_\alpha\!\big(P_{X|Y = y}\big)\,, \end{equation} for \(\alpha\in[0,+\infty]\). \end{proposition}
\begin{proof}
Let \(\Delta\colon S\to\mathbb R\) be a monotone derivation at
\(\|\cdot\|\). Thus \(\Delta\) is additive and order-preserving, and $\Delta(P\otimes Q)
=\Delta(P)\|Q\|+\|P\|\Delta(Q)$.

\medskip
\noindent\textbf{Step 1: Reduction to derivations vanishing on scalars.}
For \(x,y\in\mathbb R_+\), the semiring equivalence
\begin{align}
    \begin{pmatrix}x&0\\0&y\end{pmatrix}
\succeq
\begin{pmatrix}x+y&0\\0&0\end{pmatrix}
\succeq
\begin{pmatrix}x&0\\0&y\end{pmatrix},
\end{align}
proved as in Section~\ref{sec: all monotone homomorphisms}, implies by
monotonicity and additivity that $\Delta(x+y)=\Delta(x)+\Delta(y)$.
Define $\Delta_0(P):=\Delta(P)-\Delta(\|P\|)$.
Then \(\Delta_0(x)=0\) for every \(x\in\mathbb R_+\). Using the
additivity and multiplicativity of \(\|\cdot\|\), one easily verifies that
\(\Delta_0\) is additive and satisfies the Leibniz rule. Moreover, if
\([P]\rgeq[Q]\), then \(\|P\|=\|Q\|\), and hence
\begin{align}
    \Delta_0(P)-\Delta_0(Q)=\Delta(P)-\Delta(Q)\geq0.
\end{align}
Thus \(\Delta_0\) is again a monotone derivation and is interchangeable
with \(\Delta\) in the sense of Definition~8.2 of
\cite{fritz2023abstractII}. We may therefore assume that $\Delta(x)=0$.

\medskip
\noindent\textbf{Step 2: General form.}
On probability vectors, \(\Delta\) satisfies the axioms of an entropy
functional. Proposition~\ref{proposition:barycenter entropies} therefore
gives a finite positive measure \(\tau\) on \([0,+\infty]\) such that
\begin{equation}
\label{eq:derivation-on-probabilities}
\Delta(P_X)
=\int_{[0,+\infty]}H_\alpha(P_X)\,\mathrm d\tau(\alpha).
\end{equation}

Decompose \(P_{XY}\) into its nonzero columns. By
Lemma~\ref{lemma:shifting}, additivity, the Leibniz rule, and $\Delta(x)=0$,
\begin{align}
\Delta(P_{XY})
=\sum_{y:P_Y(y)>0}\Delta\!\left(P_Y(y)P_{X|Y=y}\right) =\sum_{y:P_Y(y)>0}P_Y(y)\Delta(P_{X|Y=y}).
\end{align}
Substituting \eqref{eq:derivation-on-probabilities} yields
\begin{equation}
\label{eq:derivFinalForm}
\Delta(P_{XY})
=\sum_{y\in\mathcal Y}P_Y(y)
  \int_{[0,+\infty]}
  H_\alpha(P_{X|Y=y})\,\mathrm d\tau(\alpha),
\end{equation}
where terms with \(P_Y(y)=0\) are understood to vanish. Additivity of
\eqref{eq:derivFinalForm} is immediate, and its Leibniz rule is easily
verified from $H_\alpha(P\otimes Q)=H_\alpha(P)+H_\alpha(Q)$.

Let \(u_2=(1/2,1/2)\). Since \(H_\alpha(u_2)=\log 2\) for every
\(\alpha\), the normalization \(\Delta(u_2)=\log 2\) implies that \(\tau\) is a probability
measure. The normalization \(\Delta(u)=1\) in the abstract
Theorem~\ref{thm:Fritz2022} is recovered by positive rescaling; such rescaling
preserves the ordering inequalities.

\medskip
\noindent\textbf{Step 3: Extremal derivations.}
For \(\alpha\in[0,+\infty]\), define
\begin{equation}
\label{eq:ExtDerivation-proof}
H_{0,\alpha}(P_{XY})
:=\sum_{y\in\mathcal Y}P_Y(y)
H_\alpha(P_{X|Y=y}).
\end{equation}
Then \eqref{eq:derivFinalForm} becomes
\begin{equation}
\Delta(P_{XY})
=\int_{[0,+\infty]}
H_{0,\alpha}(P_{XY})\,\mathrm d\tau(\alpha).
\end{equation}
for some probability measure \(\tau\). Hence, each derivation is a convex combination (or, rather, a barycentre) of the extremal derivations \(H_{0,\alpha}(P_{XY})\). 
Consequently, to characterize the extremal elements of the set---and thereby obtain the general form of conditional entropies through barycentric, i.e., convex, decompositions---it suffices to consider the derivations \(H_{0,\alpha}\) with \(\alpha\in[0,+\infty]\). These derivations are themselves extremal. Indeed, this follows by restricting to probability distributions with trivial side information and invoking the known fact that, for fixed \(\alpha\), the R\'enyi entropy \(H_\alpha\) cannot be expressed as a nontrivial convex combination of R\'enyi entropies of different orders. 
\end{proof}

\newcommand{\ParameterMainForms}{the candidate families derived above}
\newcommand{\ParameterMainTemperateForm}{the temperate representation derived above}
\newcommand{\ParameterMainLift}{the convexity-to-monotonicity and semiring-lift results proved in the appendix}
\newcommand{\ParameterMainDerivationLift}{the corresponding derivation results proved in the appendix}
\newcommand{\ParameterMainDerivationRepresentation}{the representation of derivations obtained above}
\newcommand{\ParameterMainSemiring}{the conditional-majorization semiring defined above}
\newcommand{\ParameterMainNormalization}{the normalization fixed above}
\newcommand{\ParameterMainDerivationDefinition}{defined above}
\newcommand{\ParameterMainFinalFamilies}{the temperate, tropical, and derivation candidate families introduced above}
\newcommand{\ParameterSpanOrigin}{is the representation result established above}

\section{Sufficient conditions for the parameters}
\label{sec:sufficient-conditions-measure}

In this section, we establish sufficient conditions for the parameters under which the
homomorphisms and derivations (and, consequently, the associated conditional entropies)
are monotone under conditional majorization. In the following section, we show that these conditions are also necessary. Although the arguments in these sections are conceptually elementary, reducing largely to the analysis of derivatives of regular functions together with the verification of the relevant convergence properties, the length and technical detail of these calculations make them tedious to verify. To make these arguments as transparent and verifiable as possible, we have therefore provided a Lean~4 formalization of the results of these two sections, accompanied by a document that relates the formal proofs to the corresponding mathematical statements. Detailed proofs and calculations for Sections~IV and~V are provided in a separate companion note published on GitHub~\cite{rubboli2026parametersFull}. The Lean formalization and its dependencies are documented in a companion note published on GitHub~\cite{rubboli2026parametersLean}. Both notes and the formalization source code are preserved together in the archived release on Zenodo~\cite{rubboli2026lean}.

In the previous sections, we derived the general form of monotone homomorphisms and derivations, both of which depend on a choice of positive measure. However, not all such choices are valid, as some lead to violations of monotonicity under conditionally mixing channels. In particular, as we show later, while monotonicity under mixing channels always holds, monotonicity under channels acting on the conditioning system fails for some choices of the parameters. To characterize the admissible cases, we examine the convexity properties of the homomorphism as a function of the joint probability matrix---a standard and often more tractable approach. Indeed, as discussed in Appendix~\ref{app: convexity and monotonicity}, convexity and monotonicity under conditionally mixing channels are essentially equivalent. Moreover, it suffices to establish the convexity properties required to ensure monotonicity under channels acting on the conditioning system, which form a subclass of conditionally mixing channels as characterized in Lemma~\ref{lem: extension to general channels}. Indeed, monotonicity under mixing channels on the first system is readily verified for all values of the parameters.

Let us first identify the function for which we need to establish convexity or concavity properties. We observe that, in the temperate case, the homomorphisms take the form~\eqref{eq:FinalForm}
\begin{align}\label{eq:TemperateForm}
\Phi(P_{XY})&=\Phi_{t,\tau}(P_{XY})=
\sum_{y\in \Y} P_Y(y)\exp{\left(t\int_{[0,+\infty]}H_\alpha(P_{X|Y=y})\d\tau(\alpha)\right)} \,,
\end{align}
where $\tau$ is a probability measure, the normalized version of the measure $\mu$ of \eqref{eq:FinalForm}, and $t=\mu\big([0,+\infty]\big)\in\R\setminus\{0\}$. As detailed in Appendix~\ref{app: convexity and monotonicity}, we need to analyze the convexity properties of the function $P_{XY} \mapsto \Phi(P_{XY})$.
Therefore, in the following analysis, it suffices to study the convexity properties of the function defined for $\vec{x} \in \mathbb{R}^d_+$ by
\begin{equation}\label{eq:conditionfunct}
    \vec{x} \longmapsto \|\vec{x}\|_1 \exp\left( t\int_{[0,+\infty]} H_\alpha(\hat{x}) \,\d\tau(\alpha) \right)\,,
\end{equation}
where $\|\vec{x}\|_1=\sum_i x_i$ and $\hat{x} = \vec{x}/\|\vec{x}\|_1$. We begin with the discrete case, where the measure is supported on a finite set of points. We then extend these results to a general measure by constructing a suitable discrete approximation and taking the limit in which the approximation converges to the original measure. Indeed, convexity and monotonicity properties are preserved under taking limits. Finally, at the end of the section, we establish sufficient conditions for both the tropical case and the derivations by taking limits of the conditions obtained earlier for the temperate case. Indeed, these two cases can be viewed as limiting regimes of the temperate case, where the total weight of the measure $t$ diverges to infinity or vanishes to zero, respectively.

\subsection{Sufficient conditions for the discrete case}

The discrete case exploits two basic facts.  The first is Minkowski's inequality
\cite[Theorem~9]{bullen2013handbook}.

\begin{lemma}
\label{lem:Minkowski}
Let $\alpha>0$.  The function, defined on
$\R^d_{>0}$,
\begin{equation}
\vec{x}\longmapsto\left(\sum_{i=1}^d x_i^\alpha\right)^{1/\alpha},
\end{equation}
\begin{enumerate}[itemsep=1pt]
\item is convex if $\alpha>1$;
\item is concave if $0<\alpha<1$.
\end{enumerate}
\end{lemma}

The second fact is the standard convexity criterion for multivariate monomials
\cite[Lemma~8]{mu21_renyi}; see also Ref.~\cite[Proposition~13]{farooq2024matrix}.

\begin{lemma}
\label{lem:Matrix-majorization}
Let $N\in\mb{N}$ and let
$\vec{\beta}=(\beta_1,\ldots,\beta_N)\in\R^N$ satisfy
$\sum_{\ell=1}^N\beta_\ell=1$.  The function, defined on
$\R^N_{>0}$,
\begin{equation}
\vec{y}\longmapsto\prod_{\ell=1}^N y_\ell^{\beta_\ell},
\end{equation}
\begin{enumerate}[itemsep=1pt]
\item is concave if and only if $\beta_\ell\geq0$ for every $\ell$;
\label{it: concavity matrix majorization}
\item is convex if and only if there is a unique index $k$ such that
$\beta_k>0$ and $\beta_\ell\leq0$ for every $\ell\neq k$.
\label{it: convexity matrix majorization}
\end{enumerate}
\end{lemma}

Combining these facts gives the result for a discrete probability measure supported on finitely many points.
\begin{lemma}
\label{lem:discrete-measure}
Let $\tau$ be a discrete probability measure supported on finitely many points
$\alpha_1<\cdots<\alpha_N$ of $[0,+\infty]$, and let
$t\in\R\setminus\{0\}$.  For every $d\in\mb{N}$, the function
\begin{equation}
\vec{x}\longmapsto\|\vec{x}\|_1
\exp\!\left(t\int_{[0,+\infty]}H_\alpha(\hat{x})\,
\d\tau(\alpha)\right)
\end{equation}
on $\R^d_{\geq0}\setminus\{0\}$ is
\begin{enumerate}
\item concave if $t>0$, $\alpha_1,\ldots,\alpha_N<1$, and $\sum_{\ell=1}^N\frac{\alpha_\ell}{1-\alpha_\ell}
\tau(\{\alpha_\ell\})\leq\frac1t$
;
\label{it: concave discrete}
\item convex if $t<0$ and $\alpha_1,\ldots,\alpha_N\leq1$, or
  \label{item: sufficient discrete convex 1}
  \item  convex if $t<0$ and $\alpha_1,\ldots,\alpha_{N-1}<1$, $\alpha_N>1$, and $\sum_{\ell=1}^N\frac{\alpha_\ell}{1-\alpha_\ell}
  \tau(\{\alpha_\ell\})\leq\frac1t$,
  where $\alpha/(1-\alpha)$ has value $-1$ at $\alpha=+\infty$.
  \label{item: sufficient discrete convex 2}
\end{enumerate}
\end{lemma}

\begin{proof}
    The proof leverages the discreteness of the measure in conjunction with Minkowski's inequality and convexity of multivariate monomials in Lemma~\ref{lem:Minkowski} and~\ref{lem:Matrix-majorization}.
    Using that $\hat{x}=\vec{x}/\|\vec{x}\|_1$, we can rewrite 
    \begin{align}
    \|\vec{x}\|_1 \exp\left( t\sum\nolimits_\ell H_{\alpha_\ell}(\hat{x})\,\tau(\{\alpha_\ell\}) \right)
    = \|\vec{x}\|_1^{1 - t\sum\nolimits_\ell\frac{\alpha_\ell}{1-\alpha_\ell} \,\tau(\{\alpha_\ell\})} \exp\left( t\sum\nolimits_\ell H_{\alpha_\ell}(\vec{x}) \,\tau(\{\alpha_\ell\}) \right)
\end{align}
as long as $\alpha_\ell\neq 1$ for any $\ell\in\{1,\ldots,N\}$. Hence, we need to study the properties of the function
    \begin{equation}
    \vec{x} \longmapsto \left(\sum\nolimits_{i} x_i \right)^{1 - t\sum\nolimits_\ell \frac{\alpha_\ell}{1-\alpha_\ell} \, \tau(\{\alpha_\ell\})} \prod\nolimits_{\ell} \left(\sum\nolimits_{i} x_i^{\alpha_\ell} \right)^{\tau(\{\alpha_\ell\})\frac{1}{1-\alpha_\ell}t}  = \left(\sum\nolimits_{i} x_i \right)^{\bar{\beta}} \prod\nolimits_{\ell} \left(\sum\nolimits_{i} x_i^{\alpha_\ell} \right)^{\frac{\beta_\ell}{\alpha_\ell}}\,,
    \end{equation}
    where we defined the exponents $\beta_\ell = t\frac{\alpha_\ell}{1-\alpha_\ell} \, \tau(\{\alpha_\ell\})$ and $\bar{\beta}= 1 - \sum\nolimits_\ell \beta_\ell $. We discuss the three cases separately.  Let us start with the case in item~\ref{it: concave discrete}. By assumption, it holds that $\beta_\ell \geq 0$ for all $\ell$. By Lemma~\ref{lem:Minkowski}, the function $\vec{x} \mapsto \big(\sum\nolimits_{i=1}^d x_i^{\alpha_\ell}\big)^{1/\alpha_\ell}$ is concave. Hence, for two vectors $\vec{x}$ and $\vec{z}$ and a coefficient $\lambda \in [0,1]$, we have that
\begin{align}
    \left(\sum\nolimits_{i} (\lambda x_i+(1-\lambda)z_i)^{\alpha_\ell}\right)^\frac{1}{\alpha_\ell} \geq \lambda \left(\sum\nolimits_{i} x_i^{\alpha_\ell}\right)^\frac{1}{\alpha_\ell} + (1-\lambda) \left(\sum\nolimits_{i}  z_i^{\alpha_\ell}\right)^\frac{1}{\alpha_\ell} \,.
\end{align}
Moreover, the exponent $\beta_\ell \geq 0$, so the corresponding power is monotone. Therefore, we have that
\begin{align}
    \left(\sum\nolimits_{i} (\lambda x_i+(1-\lambda)z_i)^{\alpha_\ell}\right)^{\frac{\beta_\ell}{\alpha_\ell}} \geq \left(\lambda \left(\sum\nolimits_{i} x_i^{\alpha_\ell}\right)^\frac{1}{\alpha_\ell} + (1-\lambda) \left(\sum\nolimits_{i}  z_i^{\alpha_\ell}\right)^\frac{1}{\alpha_\ell}\right)^{\beta_\ell} \,.
\end{align}
Finally, we use the joint concavity of multivariate monomials from Lemma~\ref{lem:Matrix-majorization}, item~\ref{it: concavity matrix majorization}. Since it holds that $\beta_\ell >0$, we only need to ensure that $0 \leq \bar{\beta}= 1- \sum\nolimits_\ell\beta_\ell $. Under this condition, we have that
\begin{align}
    &\left(\sum\nolimits_{i} \lambda x_i+(1-\lambda )z_i \right)^{\bar{\beta}} \prod\nolimits_{\ell} \left(\sum\nolimits_{i} (\lambda x_i+(1-\lambda ) z_i)^{\alpha_\ell} \right)^{\frac{\beta_\ell}{\alpha_\ell}} \\
    &\qquad \geq \lambda \left(\sum\nolimits_{i} x_i \right)^{\bar{\beta}} \prod\nolimits_{\ell} \left(\sum\nolimits_{i} x_i^{\alpha_\ell} \right)^{\frac{\beta_\ell}{\alpha_\ell}} + (1-\lambda) \left(\sum\nolimits_{i} z_i \right)^{\bar{\beta}} \prod\nolimits_{\ell} \left(\sum\nolimits_{i} z_i^{\alpha_\ell} \right)^{\frac{\beta_\ell}{\alpha_\ell}} \,.
\end{align}
The constraint $ 0 \leq \bar{\beta} =1-\sum\nolimits_\ell\beta_\ell$ yields the condition on the measure appearing on the first item of the lemma. 

 Let us turn to the second case in the statement in item~\ref{item: sufficient discrete convex 1}. In this case, $t<0$. We begin by assuming $\alpha_N<1$. The extension to include the possibility $\alpha_N=1$ can then be obtained via a limit argument (see the discussion below). It follows that $\beta_\ell < 0$ for all $\ell$. By Lemma~\ref{lem:Minkowski}, the function $\vec{x} \mapsto \big(\sum\nolimits_{i=1}^d x_i^{\alpha_\ell}\big)^{1/\alpha_\ell}$ is concave. Hence, for two vectors $\vec{x}$ and $\vec{z}$ and a coefficient $\lambda \in [0,1]$, we have that
\begin{align}
    \left(\sum\nolimits_{i} (\lambda x_i+(1-\lambda)z_i)^{\alpha_\ell}\right)^\frac{1}{\alpha_\ell} \geq \lambda \left(\sum\nolimits_{i} x_i^{\alpha_\ell}\right)^\frac{1}{\alpha_\ell} + (1-\lambda) \left(\sum\nolimits_{i}  z_i^{\alpha_\ell}\right)^\frac{1}{\alpha_\ell} \,.
\end{align}
Moreover, the exponent $\beta_\ell < 0$, so the corresponding power is antimonotone. Therefore, we have that
\begin{align}
    \left(\sum\nolimits_{i} (\lambda x_i+(1-\lambda)z_i)^{\alpha_\ell}\right)^{\frac{\beta_\ell}{\alpha_\ell}} \leq \left(\lambda \left(\sum\nolimits_{i} x_i^{\alpha_\ell}\right)^\frac{1}{\alpha_\ell} + (1-\lambda) \left(\sum\nolimits_{i}  z_i^{\alpha_\ell}\right)^\frac{1}{\alpha_\ell}\right)^{\beta_\ell} \,.
\end{align}
Finally, we use the joint convexity result from Lemma~\ref{lem:Matrix-majorization}, item~\ref{it: convexity matrix majorization}. Since it holds that $\beta_\ell <0$, we only need to ensure that $0 \leq \bar{\beta}= 1- \sum\nolimits_\ell\beta_\ell $. However, this condition already follows from the fact that $\beta_\ell<0$ and hence we do not impose any additional constraint on the measure.  Under this condition, we have that
\begin{align}
    &\left(\sum\nolimits_{i} \lambda x_i+(1-\lambda )z_i \right)^{\bar{\beta}} \prod\nolimits_{\ell} \left(\sum\nolimits_{i} (\lambda x_i+(1-\lambda ) z_i)^{\alpha_\ell} \right)^{\frac{\beta_\ell}{\alpha_\ell}} \\
    &\qquad \leq \lambda \left(\sum\nolimits_{i} x_i \right)^{\bar{\beta}} \prod\nolimits_{\ell} \left(\sum\nolimits_{i} x_i^{\alpha_\ell} \right)^{\frac{\beta_\ell}{\alpha_\ell}} + (1-\lambda) \left(\sum\nolimits_{i} z_i \right)^{\bar{\beta}} \prod\nolimits_{\ell} \left(\sum\nolimits_{i} z_i^{\alpha_\ell} \right)^{\frac{\beta_\ell}{\alpha_\ell}} \,.
\end{align}
Hence, we proved the second item of the lemma.

 Let us turn to item~\ref{item: sufficient discrete convex 2}. By Lemma~\ref{lem:Minkowski}, the function $\vec{x} \mapsto \big(\sum\nolimits_{i=1}^d x_i^{\alpha_\ell}\big)^{1/\alpha_\ell}$ is concave for $\ell\leq N-1$. Hence, for two vectors $\vec{x}$ and $\vec{z}$ and a coefficient $\lambda \in [0,1]$, we have that, for $\ell\leq N-1$,
\begin{align}
    \left(\sum\nolimits_{i} (\lambda x_i+(1-\lambda)z_i)^{\alpha_\ell}\right)^\frac{1}{\alpha_\ell} \geq \lambda \left(\sum\nolimits_{i} x_i^{\alpha_\ell}\right)^\frac{1}{\alpha_\ell} + (1-\lambda) \left(\sum\nolimits_{i}  z_i^{\alpha_\ell}\right)^\frac{1}{\alpha_\ell} \,.
\end{align}
Additionally, for $\alpha_N > 1$, the function is convex and hence
\begin{align}
    \left(\sum\nolimits_{i} (\lambda x_i+(1-\lambda)z_i)^{\alpha_N}\right)^{\frac{1}{\alpha_N}} \leq \lambda \left(\sum\nolimits_{i} x_i^{\alpha_N}\right)^{\frac{1}{\alpha_N}} + (1-\lambda) \left(\sum\nolimits_{i}  z_i^{\alpha_N}\right)^{\frac{1}{\alpha_N}} \,.
\end{align}
Moreover, the exponent $\beta_\ell < 0$ for $\ell\leq N-1$, so the corresponding power is antimonotone. Therefore, we have that
\begin{align}
    \left(\sum\nolimits_{i} (\lambda x_i+(1-\lambda)z_i)^{\alpha_\ell}\right)^{\frac{\beta_\ell}{\alpha_\ell}} \leq \left(\lambda \left(\sum\nolimits_{i} x_i^{\alpha_\ell}\right)^\frac{1}{\alpha_\ell} + (1-\lambda) \left(\sum\nolimits_{i}  z_i^{\alpha_\ell}\right)^\frac{1}{\alpha_\ell}\right)^{\beta_\ell} \,.
\end{align}
Additionally, for $\alpha_N > 1$, the exponent $\beta_N = \tau(\{\alpha_N\}) \frac{\alpha_N}{1-\alpha_N}t > 0$ and hence the corresponding power is monotone. Therefore, we have that
\begin{align}
   \qquad \;\, \left(\sum\nolimits_{i} (\lambda x_i+(1-\lambda)z_i)^{\alpha_N}\right)^\frac{\beta_N}{\alpha_N} \leq \left(\lambda \left(\sum\nolimits_{i} x_i^{\alpha_N}\right)^\frac{1}{\alpha_N} + (1-\lambda) \left(\sum\nolimits_{i}  z_i^{\alpha_N}\right)^\frac{1}{\alpha_N}\right)^{\beta_N} \,.
\end{align}
Finally, we use the joint convexity result from Lemma~\ref{lem:Matrix-majorization}, item~\ref{it: convexity matrix majorization}. Since it holds that $\beta_\ell <0$ for $\ell\leq N-1$ and $\beta_N >0$, we only need to ensure that $ \bar{\beta}= 1- \sum\nolimits_\ell\beta_\ell \leq 0$.  Under this condition, we have that
\begin{align}
    &\left(\sum\nolimits_{i} \lambda x_i+(1-\lambda )z_i \right)^{\bar{\beta}} \prod\nolimits_{\ell} \left(\sum\nolimits_{i} (\lambda x_i+(1-\lambda ) z_i)^{\alpha_\ell} \right)^{\frac{\beta_\ell}{\alpha_\ell}} \\
    &\qquad \leq \lambda \left(\sum\nolimits_{i} x_i \right)^{\bar{\beta}} \prod\nolimits_{\ell} \left(\sum\nolimits_{i} x_i^{\alpha_\ell} \right)^{\frac{\beta_\ell}{\alpha_\ell}} + (1-\lambda) \left(\sum\nolimits_{i} z_i \right)^{\bar{\beta}} \prod\nolimits_{\ell} \left(\sum\nolimits_{i} z_i^{\alpha_\ell} \right)^{\frac{\beta_\ell}{\alpha_\ell}} \,.
\end{align}
The constraint $ \bar{\beta} =1-\sum\nolimits_\ell\beta_\ell \leq 0$ yields the condition on the measure appearing on the third item of the lemma.

Finally, we include the endpoints by pointwise approximation. An atom at $\alpha=1$, which can occur only in item~\ref{item: sufficient discrete convex 1}, is replaced by an atom of the same mass at $1-1/n$. The resulting functions are convex by the preceding argument and converge pointwise to the desired function since $H_{1-1/n}(p)\to H(p)$. For an atom at $\alpha=0$, replace $\tau({0})\delta_0$ by $\tau({0})\delta_{\varepsilon_n}$, where $\varepsilon_n\downarrow0$. Writing
\begin{equation}
I:=\int_{[0,+\infty]\setminus\{0,1\}}
\frac{\alpha}{1-\alpha}\d\tau(\alpha),
\qquad
\delta_n:=\tau({0})\frac{\varepsilon_n}{1-\varepsilon_n},
\end{equation}
we have $I_n:=I+\delta_n\to I$. In the cases with a moment constraint, set $t_n=t/(1+t\delta_n)$.
Then $t_n\to t$ and $t_n I_n-1=(tI-1)/(1+t\delta_n)$
so the required sign of the moment constraint is preserved for all sufficiently large $n$. In the case $t<0$ with support contained in $[0,1]$, where no moment constraint is required, we simply keep $t_n=t$. Since $H_{\varepsilon_n}(p)\to H_0(p)$, the corresponding functions converge pointwise to the desired one, and pointwise limits preserve convexity and concavity. 
\end{proof}
\subsection{Sufficient conditions for the temperate homomorphisms - general measure}
In this section, we extend the convexity result to the case where $\tau$ is a general probability measure. The idea is to approximate $\tau$ by a sequence of finitely supported probability measures $(\tau_n)_{n \in \mathbb{N}}$, obtained by discretizing the intervals with increasingly fine partitions. Convexity holds for each $\tau_n$ by the result from the previous section. Since $(\tau_n)$ converges to $\tau$ (in a suitable sense), and convexity is preserved under pointwise limits, the result carries over to the continuous case.

\begin{proposition}
\label{prop:sufficient-conditions}
Let $\tau$ be a probability measure on $[0,+\infty]$ and let
$t\in\R\setminus\{0\}$.  The function in Eq.~\eqref{eq:conditionfunct} is
\begin{enumerate}
\item concave if $t>0$,
$\operatorname{supp}(\tau)\subseteq[0,1]$, $\tau(\{1\})=0$, and $\displaystyle\int_{[0,1)}\frac{\alpha}{1-\alpha}\,\d\tau(\alpha)\leq\frac1t$
\item convex if $t<0$ and either
  \begin{enumerate}[label=(\alph*)]
  \item $\operatorname{supp}(\tau)\subseteq[0,1]$, or
  \item there exists $\alpha_*\in(1,+\infty]$ such that
$\operatorname{supp}(\tau)\subseteq[0,1]\cup\{\alpha_*\}$,
$\tau(\{1\})=0$, and
\begin{align}
    \int_{[0,1)}\frac{\alpha}{1-\alpha}\,\d\tau(\alpha)
+\int_{(1,+\infty]}\frac{\alpha}{1-\alpha}\,\d\tau(\alpha)
\leq\frac{1}{t}
\end{align}
with both integrals finite.
  \end{enumerate}
\end{enumerate}
\end{proposition}

\begin{proof}
We begin by considering the first case. The idea is to approximate the R\'enyi parameter directly, thereby obtaining a single finitely supported measure that can be used simultaneously for all vectors appearing in the concavity inequality.

For every $N\geq2$, define the finite-valued Borel map
\begin{equation}
    T_N(\alpha):=\frac{\lfloor N\alpha\rfloor}{N}
    \quad (0\leq\alpha<1),
    \qquad
    T_N(1):=1-\frac{1}{N},
\end{equation}
where $\lfloor r\rfloor$ denotes the greatest integer not exceeding
$r$, and let $\tau_N$
be the push-forward of $\tau$ through $T_N$, namely
$\tau_N(B):=\tau(T_N^{-1}(B))$ for every Borel set
$B\subseteq[0,1]$. Equivalently, for every bounded Borel function $h$,
\begin{equation}
    \int_{[0,1]}h(\beta)\,\d\tau_N(\beta)
    =
    \int_{[0,1]}h(T_N(\alpha))\,\d\tau(\alpha).
\end{equation}
The measure $\tau_N$ is supported on the finite set
$\{0,1/N,\ldots,(N-1)/N\}$. Since $T_N(\alpha)\leq\alpha$ for
$\alpha<1$ and the function $r\mapsto r/(1-r)$ is increasing on
$[0,1)$, we have
\begin{align}
    \int_{[0,1)}\frac{\beta}{1-\beta}\,\d\tau_N(\beta)
    =
    \int_{[0,1]}
    \frac{T_N(\alpha)}{1-T_N(\alpha)}\,\d\tau(\alpha)
    \leq
    \int_{[0,1)}
    \frac{\alpha}{1-\alpha}\,\d\tau(\alpha)
    \leq\frac{1}{t}.
\end{align}
Hence, $\tau_N$ satisfies the condition in
Lemma~\ref{lem:discrete-measure} for every $N$.

The atoms at $0$ and, when present, at $+\infty$ are included
directly by continuity, without using expressions that divide by
$\alpha$ or $1-\alpha$ at the endpoints. Since $T_N(\alpha)\to\alpha$ for $\tau$-almost every $\alpha$, continuity gives $H_{T_N(\alpha)}(p)\to H_\alpha(p)$, while $\left|H_{T_N(\alpha)}(p)\right|\leq \log d$. Since $\tau$ is finite, the dominated convergence theorem gives
\begin{align}
    \int_{[0,1]}H_\beta(p)\,\d\tau_N(\beta)
    =
    \int_{[0,1]}H_{T_N(\alpha)}(p)\,\d\tau(\alpha)
    \longrightarrow
    \int_{[0,1]}H_\alpha(p)\,\d\tau(\alpha).
\end{align}
Consequently, defining
\begin{align}
    \Phi_N(\vec{x})
    :=
    \|\vec{x}\|_1
    \exp\left(
        t\int_{[0,1]}H_\beta(\hat{x})\,\d\tau_N(\beta)
    \right)
    =
    \|\vec{x}\|_1
    \exp\left(
        t\int_{[0,1]}H_{T_N(\alpha)}(\hat{x})\,\d\tau(\alpha)
    \right),
    \label{eq:PhiN}
\end{align}
we obtain $\Phi_N(\vec{x})\to\Phi(\vec{x})$ for every $\vec{x}$.

To prove concavity, let $\lambda\in[0,1]$ and set $\vec{x}=\lambda\vec{w}+(1-\lambda)\vec{z}$.
Since the same measure $\tau_N$ is used at all three vectors,
Lemma~\ref{lem:discrete-measure} gives $\Phi_N(\vec{x})
    \geq
    \lambda\Phi_N(\vec{w})
    +(1-\lambda)\Phi_N(\vec{z})$.
Taking $N\to\infty$ yields 
\begin{equation}
    \Phi(\vec{x})
    \geq
    \lambda\Phi(\vec{w})
    +(1-\lambda)\Phi(\vec{z}),
    \label{eq:Phiconc}
\end{equation}
which proves the first case. Notice also that for $\alpha\in[0,1)$, because
\begin{equation}
    0\leq
    \frac{T_N(\alpha)}{1-T_N(\alpha)}
    \leq
    \frac{\alpha}{1-\alpha},
\end{equation}
dominated convergence gives
\begin{equation}
    \int_{[0,1)}
    \frac{\beta}{1-\beta}\,\d\tau_N(\beta)
    \longrightarrow
    \int_{[0,1)}
    \frac{\alpha}{1-\alpha}\,\d\tau(\alpha).
\end{equation}

For the second case, we use the same definition of $T_N$ on
$[0,1)$ but set $T_N(1):=1$.
Thus, any atom at the Shannon point is kept fixed. The measures
$\tau_N=(T_N)_\#\tau$ are finitely supported on $[0,1]$, so the second
case of Lemma~\ref{lem:discrete-measure} applies without any integral
condition. The same dominated-convergence argument, now using the
continuous value $H$, shows that the corresponding functions
$\Phi_N$ converge pointwise to $\Phi$, and passage to the limit proves
the desired convexity.

Finally, in the case with a unique point $\alpha_*>1$, we keep this
point fixed and discretise only the part below $1$:
\begin{equation}
    T_N(\alpha_*):=\alpha_*,
    \qquad
    T_N(\alpha):=\frac{\lfloor N\alpha\rfloor}{N}
    \quad (0\leq\alpha<1),
    \qquad
    T_N(1):=1-\frac{1}{N}.
\end{equation}
If $\alpha_*=+\infty$, we set
$T_N(+\infty)=+\infty$
and interpret $\alpha/(1-\alpha)$ at $+\infty$ as its continuous
limit $-1$.

The positive contribution from $[0,1)$ can only decrease, whereas the
contribution of the atom at $\alpha_*$ is unchanged. Therefore, with
the signed integral understood only after the two one-sided integrals
are finite,
\begin{equation}
    \int_{[0,+\infty]\setminus\{1\}}
    \frac{\beta}{1-\beta}\,\d\tau_N(\beta)
    \leq
    \int_{[0,+\infty]\setminus\{1\}}
    \frac{\alpha}{1-\alpha}\,\d\tau(\alpha)
    \leq\frac{1}{t}.
\end{equation}
Thus every approximating function satisfies the discrete condition in
Lemma~\ref{lem:discrete-measure}. Since
$T_N(\alpha)\to\alpha$ for $\tau$-almost every $\alpha$, including the
fixed endpoint $\alpha_*=+\infty$ when it occurs, the same
dominated-convergence argument proves pointwise convergence to
$\Phi$. Taking the limit in the corresponding convexity inequalities
proves the final case.
\end{proof}

\subsection{Sufficient conditions for tropical homomorphisms}

We next take the tropical limits.  Monotonicity is preserved under pointwise
limits, so the temperate conditions immediately provide the admissible endpoint
families.

\begin{proposition}
\label{prop:sufficient-tropical}
Let $\tau$ be a probability measure on $[0,+\infty]$.  The function
\begin{equation}
\label{eq:TropicalForm}
P_{XY}\longmapsto
\min_{\substack{y\in\Y\\P_Y(y)>0}}
\int_{[0,+\infty]}H_\alpha(P_{X|Y=y})\,\d\tau(\alpha)
\end{equation}
is monotone under conditionally mixing channels if either
\begin{enumerate}[label=(\alph*)]
\item $\operatorname{supp}(\tau)\subseteq[0,1]$, or
\item there is $\alpha_*\in(1,+\infty]$ such that
$\operatorname{supp}(\tau)\subseteq[0,1]\cup\{\alpha_*\}$,
$\tau(\{1\})=0$ and
\begin{equation}
\int_{[0,1)}\frac{\alpha}{1-\alpha}\,\d\tau(\alpha)
+\int_{(1,+\infty]}\frac{\alpha}{1-\alpha}\,\d\tau(\alpha)\leq0
\end{equation}
with both integrals finite.
\end{enumerate}
\end{proposition}

\begin{proof}
Use
\begin{align}
\label{eq:limit-tropical}
&\lim_{t\to-\infty}\frac1t\log\!\left(
\sum_{y\in\Y}P_Y(y)
\exp\!\left(t\int H_\alpha(P_{X|Y=y})\,\d\tau(\alpha)\right)
\right)=
\min_{\substack{y\in\Y\\P_Y(y)>0}}
\int H_\alpha(P_{X|Y=y})\,\d\tau(\alpha).
\end{align}
If \(\operatorname{supp}(\tau)\subseteq[0,1]\), the result follows from Proposition~\ref{prop:sufficient-conditions}(2a) by letting \(t\to-\infty\). For case~(b), let us consider the sum \(\int_{[0,1)}\frac{\alpha}{1-\alpha}\,\d\tau(\alpha)+\int_{(1,+\infty]}\frac{\alpha}{1-\alpha}\,\d\tau(\alpha)\). If \(\tau\) gives positive weight to the single point \(\alpha_*>1\) and the sum of the integrals is negative, the proposition applies when \(t\) is negative enough. The same limit then proves the result. If the sum of the integrals is zero, set \(\tau_\varepsilon=(1-\varepsilon)\tau+\varepsilon\delta_{\alpha_*}\), where \(\delta_{\alpha_*}\) is the point measure that puts all its weight at \(\alpha_*\). Since \(\alpha_*>1\), this makes the sum of the integrals strictly negative. We can therefore use the previous case and then let \(\varepsilon\to0\).
\end{proof}

\begin{proposition}
\label{prop:sufficient-tropical-positive}
The function
\begin{equation}
P_{XY}\longmapsto
\max_{\substack{y\in\Y\\P_Y(y)>0}}H_0(P_{X|Y=y})
\end{equation}
is monotone under conditionally mixing channels.
\end{proposition}

\begin{proof}
Choose \(\tau=\delta_0\), the point measure that puts all its weight at \(\alpha=0\). This satisfies Proposition~\ref{prop:sufficient-conditions}(1): its support is contained in \([0,1]\), it gives no weight to \(\alpha=1\), and \(\int_{[0,1)}\frac{\alpha}{1-\alpha}\,\d\tau(\alpha)=0\leq 1/t\). Therefore, the corresponding temperate homomorphism is monotone for every \(t>0\). Letting \(t\to+\infty\) gives the desired tropical homomorphism:
\begin{equation}
\lim_{t\to+\infty}\frac1t\log\!\left(
\sum_{y\in\Y}P_Y(y)\exp\!\left(tH_0(P_{X|Y=y})\right)\right)
=\max_{\substack{y\in\Y\\P_Y(y)>0}}H_0(P_{X|Y=y}).
\end{equation}
\end{proof}

\subsection{Sufficient conditions for derivations}

We now turn to the derivations.

\begin{proposition}
\label{prop:sufficient-derivations}
The function
\begin{equation}
P_{XY}\longmapsto
\sum_{y\in\Y}P_Y(y)H_\alpha(P_{X|Y=y})
\end{equation}
is monotone under conditionally mixing channels if $\alpha\in[0,1]$.
\end{proposition}

\begin{proof}
Fix \(\alpha\in[0,1]\) and choose the representing measure
\(\tau=\delta_\alpha\). Then
\(\operatorname{supp}(\tau)\subseteq[0,1]\), so the first \(t<0\) case of
Proposition~\ref{prop:sufficient-conditions} applies. Consequently, the corresponding
temperate homomorphism is monotone for every \(t<0\). The finite log-sum-exp identity
\begin{equation}
\lim_{t\to0^-}\frac{1}{t}\log\!\left(
\sum_{y\in\mathcal Y}P_Y(y)
\exp\!\left(tH_\alpha(P_{X|Y=y})\right)
\right)
=
\sum_{y\in\mathcal Y}P_Y(y)
H_\alpha(P_{X|Y=y}).
\end{equation}
passes monotonicity to the limit.
\end{proof}

\section{Necessary conditions for the parameters}
\label{sec:necessary-conditions-measures}
In this section, we show that the sufficient conditions on the parameters derived in the previous section that ensure monotonicity under conditionally mixing channels are also necessary.

Our approach is to identify an explicit perturbation direction and compute the second derivative of the function along this direction. Since the relevant monotonicity properties are equivalent to the corresponding convexity or concavity properties (see Proposition~\ref{prop:convconc}), a second derivative with the wrong sign provides the desired counterexample.

Finally, we also note that the problem is inherently challenging: obtaining a sharp characterisation of the admissible range requires constructing counterexamples in increasingly high dimensions. Indeed, our numerical investigations indicate that no such violations occur in low-dimensional settings. While the proofs are rather intricate, the underlying idea is simple. We therefore provide a detailed explanation of the proof strategy used to establish the necessary conditions.

 In this section, we occasionally use the natural logarithm \(\ln\) as an auxiliary function when differentiating positive functions. This is only for notational convenience and does not affect the entropy convention or any of the resulting parameter conditions.

\bigskip
\noindent\textbf{Proof idea}. The main idea is to leverage the counterexamples to the required convexity or concavity properties of the function
\begin{equation}
\label{eq: multivariate divergences function}
    (v_0,v_1,\ldots,v_N)
    \longmapsto v_0^{\beta_0}v_1^{\beta_1}\cdots v_N^{\beta_N},
    \qquad
    \beta_0=1-\sum_{i=1}^N\beta_i.
\end{equation}
Such counterexamples can be found in~\cite[Lemma 8]{mu21_renyi} (see also~\cite[Proposition 13]{farooq2024matrix}). In the previous section, we argued that it is sufficient to study the convexity properties of the function
\begin{equation}
    \vec{x}\longmapsto \|\vec{x}\|_1
    \exp\left(t\int_{[0,+\infty]}H_\alpha(\hat{x})\,\d\tau(\alpha)\right).
\end{equation}
Whenever $\tau$ is supported on finitely many points $\alpha_1,\ldots,\alpha_N$ with $0<\alpha_i<+\infty$ and $\alpha_i\neq1$, this function is connected to the one in equation~\eqref{eq: multivariate divergences function} under the identification
\begin{align}
    v_0&\longleftrightarrow \sum_jx_j,
\quad v_i\longleftrightarrow
    \left(\sum_jx_j^{\alpha_i}\right)^{1/\alpha_i},\quad 
    \beta_i\longleftrightarrow
    \tau(\{\alpha_i\})t\frac{\alpha_i}{1-\alpha_i}.
\end{align}

The key difference is that, while the variables $v_i$ can be chosen independently in the first function, the terms $\bigl(\sum_jx_j^{\alpha_i}\bigr)^{1/\alpha_i}$ are not independent: varying $\vec{x}$ simultaneously affects all of them. Nevertheless, one can control the variations of $\vec{x}$ so that they influence only a selected range of values of $\alpha$. This is achieved by constructing instances in increasingly high dimensions. For example, appending an $m$-dimensional tail with entries $1/m$ to a vector $\vec{x}=(x_1,\ldots,x_k)$ gives
\begin{equation}
   \vec{y}=(x_1,\ldots,x_k,
   \underbrace{1/m,\ldots,1/m}_{m\ \mathrm{times}}),
\end{equation}
and
\begin{equation}
    \left(\sum_jy_j^\alpha\right)^{1/\alpha}
    =\left(\sum_{j=1}^kx_j^\alpha+m^{1-\alpha}\right)^{1/\alpha}
    \xrightarrow[m\to+\infty]{}
    \begin{cases}
       +\infty, & 0<\alpha<1,\\[4pt]
       \left(\sum_{j=1}^kx_j^\alpha\right)^{1/\alpha}, & \alpha>1.
    \end{cases}
\end{equation}
The points used below refine this idea by using two or three groups of repeated entries.

\bigskip
\noindent\textbf{Points and perturbation directions.}
All the counterexamples below are obtained by considering a path
\begin{equation}
    \lambda\longmapsto \vec{x}^{(d)}+\lambda\vec{v}^{(d)},
\end{equation}
where $\vec{x}^{(d)}$ is the point, $\vec{v}^{(d)}$ is the perturbation direction, and $\lambda\in\R$ is taken sufficiently close to zero so that every entry of $\vec{x}^{(d)}+\lambda\vec{v}^{(d)}$ remains positive. We first compute the first and second derivatives at $\lambda=0$ and then take the limit $d\to+\infty$. Here, $d$ is a parameter that tends to infinity, not the dimension of the vector. The dimension is the total number of entries appearing below and therefore increases with $d$. 

For the first choice, fix $r>0$ with $r\neq1$ and define
\begin{equation}
    R=r+1,\qquad
    n_{0,d}=\left\lceil d^R\right\rceil,\qquad
    n_{1,d}=\left\lceil d^{R-r}\right\rceil.
\end{equation}
The corresponding point and direction are
\begin{align}
\label{eq: two block point and direction}
\vec{x}^{(d)}_r
=\bigg[
\underbrace{1,\ldots,1}_{n_{0,d}\ \mathrm{times}},
\underbrace{d,\ldots,d}_{n_{1,d}\ \mathrm{times}}
\bigg],\quad 
\vec{v}^{(d)}_r(u_0,u_1)
=\bigg[
\underbrace{u_0,\ldots,u_0}_{n_{0,d}\ \mathrm{times}},
\underbrace{du_1,\ldots,du_1}_{n_{1,d}\ \mathrm{times}}
\bigg],
\end{align}
where $u_0,u_1\in\R$ are constants that will be chosen separately in each proof. The leading powers of $d$ contributed by the two blocks to $\sum_i(x_i+\lambda v_i)^\alpha$ are $d^R$ and $d^{R-r+\alpha}$, respectively. The ceiling signs do not change these leading powers because $\lceil d^s\rceil/d^s\to1$. Hence, the first block gives the largest contribution when $\alpha<r$, whereas the second block gives the largest contribution when $\alpha>r$. The value $r$ separates these two ranges.

For the second choice, fix
\begin{equation}
    0<a<b<+\infty,\qquad a,b\neq1,
\end{equation}
and define
\begin{equation}
    R=a+b+1,\qquad
    n_{0,d}=\left\lceil d^R\right\rceil,\qquad
    n_{1,d}=\left\lceil d^{R-a}\right\rceil,\qquad
    n_{2,d}=\left\lceil d^{R-a-b}\right\rceil.
\end{equation}
The corresponding point and direction are
\begin{align}
\label{eq: three block point and direction}
\vec{x}^{(d)}_{a,b}
=\bigg[
\underbrace{1,\ldots,1}_{n_{0,d}\ \mathrm{times}},
\underbrace{d,\ldots,d}_{n_{1,d}\ \mathrm{times}},
\underbrace{d^2,\ldots,d^2}_{n_{2,d}\ \mathrm{times}}
\bigg],\,
\vec{v}^{(d)}_{a,b}(u_0,u_1,u_2)
=\bigg[
\underbrace{u_0,\ldots,u_0}_{n_{0,d}\ \mathrm{times}},
\underbrace{du_1,\ldots,du_1}_{n_{1,d}\ \mathrm{times}},
\underbrace{d^2u_2,\ldots,d^2u_2}_{n_{2,d}\ \mathrm{times}}
\bigg],
\end{align}
where $u_0,u_1,u_2\in\R$ are constants that will be fixed in the corresponding proof. The leading powers contributed by the three blocks are
\begin{equation}
    d^R,\qquad d^{R-a+\alpha},\qquad d^{R-a-b+2\alpha},
\end{equation}
respectively.  Therefore, the first block gives the largest contribution for $\alpha<a$, the second one for $a<\alpha<b$, and the third one for $\alpha>b$. The values $a$ and $b$ separate the three ranges of $\alpha$.

The third choice separates the contribution at $\alpha=1$ from a contribution above $1$. Fix $c>1$ and set
\begin{equation}
    p=q=\frac12,\qquad
    R=c+2,\qquad
    n_{0,d}=\left\lceil d^R\right\rceil,\qquad
    n_{1,d}=\left\lceil d^{R-1}\right\rceil,\qquad
    n_{2,d}=\left\lceil d^{R-1-c}\right\rceil.
\end{equation}
We define the Shannon point and direction by
\begin{align}
\label{eq: Shannon point and direction}
\vec{x}^{(d)}_{\mathrm{Sh}}
&=\bigg[
\underbrace{q,\ldots,q}_{n_{0,d}\ \mathrm{times}},
\underbrace{dp,\ldots,dp}_{n_{1,d}\ \mathrm{times}},
\underbrace{d^2,\ldots,d^2}_{n_{2,d}\ \mathrm{times}}
\bigg],\\
\vec{v}^{(d)}_{\mathrm{Sh}}(h,v)
&=\bigg[
\underbrace{\dfrac{h}{\ln d},\ldots,\dfrac{h}{\ln d}}_{n_{0,d}\ \mathrm{times}},
\underbrace{-\dfrac{dh}{\ln d},\ldots,-\dfrac{dh}{\ln d}}_{n_{1,d}\ \mathrm{times}},
\underbrace{d^2v,\ldots,d^2v}_{n_{2,d}\ \mathrm{times}}
\bigg],
\end{align}
where $h,v\in\R$ will be specified in each application. Here, the scalar $v$ is a constant and should not be confused with the direction vector $\vec{v}^{(d)}_{\mathrm{Sh}}(h,v)$. The first and second blocks exchange their roles at $\alpha=1$, while the second and third blocks exchange their roles at $\alpha=c$. Thus, the first block controls the range $\alpha<1$, the second controls $1<\alpha<c$, and the third controls $\alpha>c$. The factors $1/\ln d$ in the first two blocks ensure that their contribution at $\alpha=1$ has a finite limit as $d\to+\infty$. If
\begin{equation}
    f_d(\lambda)=f\bigl(\vec{x}^{(d)}+\lambda\vec{v}^{(d)}\bigr),
    \qquad h_d(\lambda)=\ln f_d(\lambda),
\end{equation}
where $f$ denotes the positive function considered in the corresponding proposition, then
\begin{equation}
\label{eq: exponential second derivative necessary}
    \frac{f_d''(0)}{f_d(0)}=\bigl(h_d'(0)\bigr)^2+h_d''(0).
\end{equation}
Thus, it is enough to compute the first and second derivatives of the natural logarithm.

Finally, we stress that the derivative calculations below naturally produce the coefficient
$\alpha/(1-\alpha)$ from the prefactor in the R\'enyi entropies, leading to an apparent pole at the Shannon point $\alpha=1$.
This causes no difficulty: we differentiate at fixed finite dimension, where the singularity is removable, and use truncations away from $1$ before taking the high-dimensional limit; see the companion note~\cite{rubboli2026parametersFull} for details.

\subsection{Necessary conditions for the temperate case}
The temperate case corresponds to the bulk entropies $H_{t,\tau}$.
As discussed at the beginning of Section~\ref{sec:sufficient-conditions-measure}, it is sufficient to study the convexity properties of the function
\begin{equation}
\label{eq: temperate one column necessary}
    \vec{x}\longmapsto
    \|\vec{x}\|_1\exp\left(t\int_{[0,+\infty]}H_\alpha(\hat{x})\,\d\tau(\alpha)\right)
\end{equation}
where $\hat{x}:=\vec{x}/\|\vec{x}\|_1$, $\tau$ is a probability measure, and $t\in\R\setminus\{0\}$. We treat the cases $t>0$ and $t<0$ separately.

\begin{proposition}
\label{prop:necessary-conditions}
Let $\tau$ be a probability measure on $[0,+\infty]$ and let $t>0$. Then the function
\begin{equation}
    \vec{x}\longmapsto \|\vec{x}\|_1
    \exp\left(t\int_{[0,+\infty]}H_\alpha(\hat{x})\,\d\tau(\alpha)\right)
\end{equation}
is concave only if
\begin{equation}
    \operatorname{supp}(\tau)\subseteq[0,1],
    \qquad \tau(\{1\})=0,
    \qquad
    \int_{[0,1)}
    \frac{\alpha}{1-\alpha}\,\d\tau(\alpha)\leq\frac1t;
\end{equation}
\end{proposition}
\begin{proof}[Proof sketch]
The idea of the proof is to construct a point and a perturbation direction along which the second derivative of the function becomes strictly positive whenever one of the conditions in the proposition is violated. This contradicts concavity, since for a concave function the second derivative must be non-positive in every direction. We treat the atom at $1$, the support above $1$, and the integral below $1$ separately.

First, suppose that $\tau(\{1\})>0$. Choose $c>1$ such that $\tau(\{c\})=0$ and use the point and direction
\begin{equation}
    \big(\vec{x}^{(d)}_{\mathrm{Sh}},
    \vec{v}^{(d)}_{\mathrm{Sh}}(1,0)\big)
\end{equation}
defined in~\eqref{eq: Shannon point and direction}. The first two blocks change in opposite directions on the scale $1/\ln d$. Their contributions away from $\alpha=1$ disappear in the limit, while the contribution of $\alpha=1$ remains. Consequently,
\begin{equation}
    h_d'(0)\longrightarrow \tau(\{1\})t,
    \qquad h_d''(0)\longrightarrow0.
\end{equation}
By~\eqref{eq: exponential second derivative necessary},
\begin{equation}
    \frac{f_d''(0)}{f_d(0)}\longrightarrow t^2\tau(\{1\})^2>0.
\end{equation}
For all sufficiently large $d$, this is a convex direction and contradicts concavity. Hence $\tau(\{1\})=0$.

Next, suppose that the support contains a point larger than $1$. We can then choose $r>1$ such that
\begin{equation}
    \tau(\{r\})=0,
    \qquad \tau((r,+\infty])>0.
\end{equation}
Define
\begin{equation}
    A_r:=t\int_{(r,+\infty]}
    \frac{\alpha}{1-\alpha}\,\d\tau(\alpha)<0.
\end{equation}
The inequality is strict because $t>0$, $\tau((r,+\infty])>0$, and $\alpha/(1-\alpha)<0$ for $\alpha>1$. We use the two-block point and direction
\begin{equation}
    \left(\vec{x}^{(d)}_r,
    \vec{v}^{(d)}_r(0,1)\right).
\end{equation}
For $\alpha<r$, the first block is the largest one and does not move. For $\alpha>r$, the second block is the largest one and has relative change $1+\lambda$. It follows that
\begin{equation}
    h_d'(0)\longrightarrow A_r,
    \qquad h_d''(0)\longrightarrow-A_r.
\end{equation}
Therefore,
\begin{equation}
    \frac{f_d''(0)}{f_d(0)}
    \longrightarrow A_r^2-A_r>0,
\end{equation}
which again contradicts concavity. Thus, no point larger than $1$ can belong to the support.

It remains to obtain the bound below $1$. Choose $0<r<1$ with $\tau(\{r\})=0$ and define
\begin{equation}
    B_r:=t\int_{[0,r)}
    \frac{\alpha}{1-\alpha}\,\d\tau(\alpha)\geq0.
\end{equation}
This time we use
\begin{equation}
    \left(\vec{x}^{(d)}_r,
    \vec{v}^{(d)}_r(1,0)\right).
\end{equation}
Only the first block moves, and it is the largest block precisely for $\alpha<r$. Hence
\begin{equation}
    h_d'(0)\longrightarrow B_r,
    \qquad h_d''(0)\longrightarrow-B_r,
\end{equation}
and concavity requires
\begin{equation}
    B_r^2-B_r=B_r(B_r-1)\leq0.
\end{equation}
Since $B_r\geq0$, we obtain $B_r\leq1$. We now choose values $r$ with $\tau(\{r\})=0$ that increase to $1$. The integrand is nonnegative below $1$, so the corresponding integrals increase to the integral over $[0,1)$. Thus,
\begin{equation}
    \int_{[0,1)}\frac{\alpha}{1-\alpha}\,\d\tau(\alpha)
    =\frac1t\lim_{r\uparrow1}B_r\leq\frac1t.
\end{equation}
\end{proof}

We now establish the two intermediate results used in the $t<0$ classification. First, an atom at $1$ cannot coexist with support above $1$.

\begin{proposition}
\label{prop: no alpha=1 necessary}
Let $\tau$ be a probability measure on $[0,+\infty]$ and let $t<0$. Suppose that $\tau(\{1\})>0$. Then the function
\begin{equation}
    \vec{x}\longmapsto \|\vec{x}\|_1
    \exp\left(t\int_{[0,+\infty]}H_\alpha(\hat{x})\,\d\tau(\alpha)\right)
\end{equation}
is convex only if $\operatorname{supp}(\tau)\subseteq[0,1]$.
\end{proposition}
\begin{proof}[Proof sketch]
We need to show that nonzero weight at $\alpha=1$ cannot occur together with support above $1$. Suppose that both are present and define
\begin{equation}
    m:=\tau(\{1\})t<0.
\end{equation}
Choose $c>1$ such that
\begin{equation}
    \tau(\{c\})=0,
    \qquad \tau((c,+\infty])>0,
\end{equation}
and define
\begin{equation}
    A_c:=t\int_{(c,+\infty]}
    \frac{\alpha}{1-\alpha}\,\d\tau(\alpha)>0.
\end{equation}
The last inequality follows because $t<0$, the integrand is negative above $1$, and $\tau((c,+\infty])>0$.

We use the Shannon point and the direction
\begin{equation}
    \left(\vec{x}^{(d)}_{\mathrm{Sh}},
    \vec{v}^{(d)}_{\mathrm{Sh}}\left(-\frac{A_c}{m},1\right)\right).
\end{equation}
The limiting contribution to the first derivative from $\alpha=1$ is $m(-A_c/m)=-A_c$, while the contribution from $(c,+\infty]$ is $A_c$. These two terms cancel. The first derivative therefore tends to zero. The second derivative contributed by the third block tends to $-A_c$, whereas the contribution at $\alpha=1$ tends to zero. Hence
\begin{equation}
    h_d'(0)\longrightarrow0,
    \qquad h_d''(0)\longrightarrow-A_c<0.
\end{equation}
It follows that
\begin{equation}
    \frac{f_d''(0)}{f_d(0)}\longrightarrow-A_c<0.
\end{equation}
For sufficiently large $d$, this is a concave direction and contradicts convexity. Therefore, support above $1$ is impossible when $\tau(\{1\})>0$.
\end{proof}

The next result is a truncated estimate.

\begin{lemma}
\label{lem: truncated negative measure}
Let $\tau$ be a probability measure on $[0,+\infty]$ and let $t<0$. Suppose that $\tau(\{1\})=0$ and that the function
\begin{equation}
    \vec{x}\longmapsto \|\vec{x}\|_1
    \exp\left(t\int_{[0,+\infty]}H_\alpha(\hat{x})\,\d\tau(\alpha)\right)
\end{equation}
is convex. Let $0<a<1<b<+\infty$ satisfy $\tau(\{a\})=\tau(\{b\})=0$ and $\tau((b,+\infty])>0$. Then
\begin{align}
    \int_{[0,a)}\frac{\alpha}{1-\alpha}\,\d\tau(\alpha)
    +\int_{(b,+\infty]}\frac{\alpha}{1-\alpha}\,\d\tau(\alpha)
    \leq \frac1t.
    \label{eq: truncated moment necessary}
\end{align}
\end{lemma}
\begin{proof}[Proof sketch]
Similarly to the previous case, we identify a direction along which the second derivative is negative whenever~\eqref{eq: truncated moment necessary} is violated. This gives a concave direction and contradicts convexity.

Define
\begin{equation}
\label{eq: truncated coefficients necessary}
    B_a:=t\int_{[0,a)}\frac{\alpha}{1-\alpha}\,\d\tau(\alpha),
    \qquad
    A_b:=t\int_{(b,+\infty]}\frac{\alpha}{1-\alpha}\,\d\tau(\alpha).
\end{equation}
The assumptions $t<0$ and $\tau((b,+\infty])>0$ imply that $A_b>0$.

If~\eqref{eq: truncated moment necessary} fails, then multiplication by the negative number $t$ gives
\begin{equation}
    A_b+B_a<1,
\end{equation}
and we define
\begin{equation}
    C:=1-A_b-B_a>0.
\end{equation}
The number $C$ is the total coefficient of the middle block. The term $1$ comes from the outer factor $\|\vec{x}\|_1$. The two other terms remove the coefficients already contained in the ranges $[0,a)$ and $(b,+\infty]$. This definition avoids writing an integral through the singular expression at $\alpha=1$.

We use the three-block point and direction
\begin{equation}
    \left(
    \vec{x}^{(d)}_{a,b},
    \vec{v}^{(d)}_{a,b}
    \left(0,\frac1C,-\frac1{A_b}\right)
    \right).
\end{equation}
The first block controls $\alpha<a$ and does not move. The middle block controls $a<\alpha<b$ and has relative change $1+\lambda/C$. The third block controls $\alpha>b$ and has relative change $1-\lambda/A_b$. Thus, the limiting first derivative is
\begin{equation}
    C\frac1C+A_b\left(-\frac1{A_b}\right)=1-1=0,
\end{equation}
whereas the limiting second derivative is
\begin{equation}
    -C\left(\frac1C\right)^2
    -A_b\left(\frac1{A_b}\right)^2
    =-\frac1C-\frac1{A_b}<0.
\end{equation}
Consequently,
\begin{equation}
    h_d'(0)\longrightarrow0,
    \qquad
    h_d''(0)\longrightarrow-\frac1C-\frac1{A_b}<0,
\end{equation}
and~\eqref{eq: exponential second derivative necessary} gives a negative second derivative for all sufficiently large $d$. This contradicts convexity and proves~\eqref{eq: truncated moment necessary}.
\end{proof}

The next result is stated in the absence of an atom at $1$.

\begin{proposition}
\label{prop: complete negative measure}
Let $\tau$ be a probability measure on $[0,+\infty]$ and let $t<0$. Suppose that $\tau(\{1\})=0$ and that the function
\begin{equation}
    \vec{x}\longmapsto \|\vec{x}\|_1
    \exp\left(t\int_{[0,+\infty]}H_\alpha(\hat{x})\,\d\tau(\alpha)\right)
\end{equation}
is convex. Then either
\begin{enumerate}[label=(\alph*)]
\item $\operatorname{supp}(\tau)\subseteq[0,1]$; or
\item the support above $1$ consists of a single point $\alpha_*\in(1,+\infty]$, so that
\begin{equation}
    \operatorname{supp}(\tau)\subseteq[0,1]\cup\{\alpha_*\},
\end{equation}
and
\begin{equation}
    \int_{[0,1)}\frac{\alpha}{1-\alpha}\,\d\tau(\alpha)
    +\int_{(1,+\infty]}\frac{\alpha}{1-\alpha}\,\d\tau(\alpha)
    \leq\frac1t.
\end{equation}
with both integrals finite.
\end{enumerate}
\end{proposition}
\begin{proof}[Proof sketch]
We first show that the measure can have support on at most one point $\alpha>1$. Suppose, to the contrary, that
\begin{equation}
    1<\alpha_1<\alpha_2\leq+\infty
\end{equation}
are two points in the support. Choose finite values $a,b$ such that
\begin{equation}
    1<a<\alpha_1<b<\alpha_2,
    \qquad \tau(\{a\})=\tau(\{b\})=0.
\end{equation}
When $\alpha_2=+\infty$, the value $b$ is any finite value larger than $\alpha_1$ that separates two parts of the support. Define
\begin{equation}
    A_1:=t\int_{(a,b)}\frac{\alpha}{1-\alpha}\,\d\tau(\alpha)>0,
    \qquad
    A_2:=t\int_{(b,+\infty]}\frac{\alpha}{1-\alpha}\,\d\tau(\alpha)>0.
\end{equation}
The quantities are strictly positive because each interval contains support, $t<0$, and $\alpha/(1-\alpha)<0$ above $1$.

We use the three-block point and direction
\begin{equation}
    \left(
    \vec{x}^{(d)}_{a,b},
    \vec{v}^{(d)}_{a,b}
    \left(0,\frac1{A_1},-\frac1{A_2}\right)
    \right).
\end{equation}
The first block does not move. The second and third blocks separate the two intervals carrying $A_1$ and $A_2$. Their limiting contributions to the first derivative cancel:
\begin{equation}
    A_1\frac1{A_1}+A_2\left(-\frac1{A_2}\right)=0.
\end{equation}
The limiting second derivative is
\begin{equation}
    -A_1\left(\frac1{A_1}\right)^2
    -A_2\left(\frac1{A_2}\right)^2
    =-\frac1{A_1}-\frac1{A_2}<0.
\end{equation}
Hence the second derivative of the original function is negative for all sufficiently large $d$, which contradicts convexity. The support above $1$ therefore contains at most one point. If there is no support above $1$, alternative~(a) holds.

Suppose now that this unique point $\alpha_*$ exists. The upper integral is finite because it is the contribution of one point of a finite measure; when $\alpha_*=+\infty$, the value of $\alpha/(1-\alpha)$ is $-1$. Choose $b$ with $1<b<\alpha_*$ and $\tau(\{b\})=0$. Then $A_b$ in~\eqref{eq: truncated coefficients necessary} is $t$ times the full upper integral and is finite. Lemma~\ref{lem: truncated negative measure} gives
\begin{equation}
    B_a\geq1-A_b
\end{equation}
for every $a$ with $0<a<1$ and $\tau(\{a\})=0$. As $a$ increases to $1$, the lower integrals with respect to $\tau$ increase, so the numbers $B_a$, which contain the negative factor $t$, decrease. The displayed bound prevents them from tending to $-\infty$. Their limit is therefore finite and equals $t$ times the lower one-sided integral. Taking the limit in $A_b+B_a\geq1$ gives $t$ times the sum of the two one-sided integrals greater than or equal to $1$. Since $t<0$, dividing by $t$ reverses the inequality and gives the bound in alternative~(b). This also proves that both integrals in that alternative are finite.
\end{proof}

\subsection{Necessary conditions for the tropical homomorphisms}
We next turn to the necessary conditions for the tropical homomorphisms. The negative tropical function is
\begin{equation}
\label{eq: negative tropical necessary}
    P_{XY}\longmapsto
    \min_{\substack{y\in\Y\\P_Y(y)>0}}
    \int_{[0,+\infty]}H_\alpha(P_{X|Y=y})\,\d\tau(\alpha),
\end{equation}
whereas the positive tropical function is obtained by replacing the minimum by a maximum.

The same points and directions defined at the beginning of the section will be used. There is one additional point to explain. For the negative tropical function, monotonicity requires the corresponding one-column function to be quasiconvex. This means that its value at any point of the line segment between two points cannot be larger than both endpoint values. More precisely, a function $g$ is quasiconvex if
\begin{equation}
    g\bigl(\lambda\vec{x}+(1-\lambda)\vec{y}\bigr)
    \leq\max\{g(\vec{x}),g(\vec{y})\},
    \qquad 0<\lambda<1.
\end{equation}
In particular, this inequality applies at the midpoint $\lambda=1/2$. A negative second derivative gives a contradiction to the midpoint inequality when the first derivative is exactly zero. We will therefore correct each direction by a quantity that tends to zero so that its first derivative vanishes in every sufficiently large finite dimension.

\begin{proposition}
\label{prop: necessary negative tropical}
Let $\tau$ be a probability measure on $[0,+\infty]$. If the function in~\eqref{eq: negative tropical necessary} is monotone under conditionally mixing channels, then either
\begin{enumerate}[label=(\alph*)]
    \item $\operatorname{supp}(\tau)\subseteq[0,1]$; or
    \item the support above $1$ consists of a single point $\alpha_*\in(1,+\infty]$, so that
    \begin{equation}
        \operatorname{supp}(\tau)\subseteq[0,1]\cup\{\alpha_*\},
        \qquad \tau(\{1\})=0,
    \end{equation}
    and
    \begin{equation}
    \label{eq: tropical signed moment necessary}
        \int_{[0,1)}\frac{\alpha}{1-\alpha}\,\d\tau(\alpha)
        +\int_{(1,+\infty]}\frac{\alpha}{1-\alpha}\,\d\tau(\alpha)
        \leq0.
    \end{equation}
    with both integrals finite.
\end{enumerate}
\end{proposition}
\begin{proof}[Proof sketch]
The corresponding one-column homomorphism is
\begin{equation}
    \vec{x}\longmapsto
    \exp\left(-\int_{[0,+\infty]}H_\alpha(\hat{x})\,\d\tau(\alpha)\right).
\end{equation}
Since \(\exp\) is increasing, this function is quasiconvex if and only if its natural logarithm is quasiconvex. The positive conversion factor \(1/\log e\) does not affect quasiconvexity, and we compute derivatives of this natural logarithm below.

First, suppose that $\tau(\{1\})>0$ and that support above $1$ is present. Define
\begin{equation}
    m:=\tau(\{1\})>0.
\end{equation}
Choose $c>1$ with $\tau(\{c\})=0$ and $\tau((c,+\infty])>0$, and define
\begin{equation}
    A_c:=\int_{(c,+\infty]}
    \frac{-\alpha}{1-\alpha}\,\d\tau(\alpha)>0.
\end{equation}
At the Shannon point, use the two directions
\begin{equation}
    \vec{v}^{(d)}_{\mathrm{Sh}}\left(\frac{A_c}{m},1\right),
    \qquad
    \vec{v}^{(d)}_{\mathrm{Sh}}(0,1).
\end{equation}
Along the first direction, the first derivative tends to zero and the second derivative tends to $-A_c<0$. Along the second direction, the first derivative tends to $A_c\neq0$.

Second, suppose that two different points $1<\alpha_1<\alpha_2\leq+\infty$ belong to the support. Choose $1<a<\alpha_1<b<\alpha_2$ with $\tau(\{a\})=\tau(\{b\})=0$ and define
\begin{equation}
    A_1:=\int_{(a,b)}\frac{-\alpha}{1-\alpha}\,\d\tau(\alpha)>0,
    \qquad
    A_2:=\int_{(b,+\infty]}\frac{-\alpha}{1-\alpha}\,\d\tau(\alpha)>0.
\end{equation}
At the three-block point, use the two directions
\begin{equation}
    \vec{v}^{(d)}_{a,b}\left(0,\frac1{A_1},-\frac1{A_2}\right),
    \qquad
    \vec{v}^{(d)}_{a,b}(0,1,0).
\end{equation}
Along the first direction, the first derivative tends to zero and the second derivative tends to
\begin{equation}
    -\frac1{A_1}-\frac1{A_2}<0.
\end{equation}
Along the second direction, the first derivative tends to $A_1\neq0$.

Third, suppose that $\tau(\{1\})=0$ and that the support above $1$ consists of a unique point $\alpha_*>1$. Choose $0<a<1<b<\alpha_*$ with $\tau(\{a\})=\tau(\{b\})=0$ and define
\begin{equation}
    B_a:=\int_{[0,a)}\frac{\alpha}{1-\alpha}\,\d\tau(\alpha)\geq0,
    \qquad
    A_b:=\int_{(b,+\infty]}\frac{-\alpha}{1-\alpha}\,\d\tau(\alpha)>0.
\end{equation}
If $B_a>A_b$, define
\begin{equation}
    C:=B_a-A_b>0.
\end{equation}
There is no added $1$ in this definition because, unlike the temperate function, the tropical one-column function has no outer factor $\|\vec{x}\|_1$. At the three-block point, use
\begin{equation}
    \vec{v}^{(d)}_{a,b}\left(0,\frac1C,-\frac1{A_b}\right),
    \qquad
    \vec{v}^{(d)}_{a,b}(0,0,1).
\end{equation}
Along the first direction, the first derivative tends to zero and the second derivative tends to
\begin{equation}
    -\frac1C-\frac1{A_b}<0,
\end{equation}
while along the second direction the first derivative tends to $A_b\neq0$.

The limiting calculation does not by itself give a contradiction,
because for finite $d$ the first derivative may be small but not
exactly zero. We therefore add a suitable multiple of the second
direction to make the first derivative exactly zero. Since the first
derivative along the second direction has a nonzero limit, the required
correction tends to zero as $d\to+\infty$. It therefore does not change
the strictly negative limiting second derivative. For all sufficiently
large $d$, the corrected direction consequently has zero first
derivative and negative second derivative, giving a strict local
maximum and contradicting the midpoint inequality required by
quasiconvexity.

If there is no support above $1$, alternative~(a) holds. Suppose instead
that support above $1$ is present. The first counterexample excludes an
atom at $1$, while the second shows that the support above $1$ contains
only one point. It remains to prove the finiteness and signed-moment
requirements in alternative~(b). Thus, suppose that $\tau(\{1\})=0$ and
that the support above $1$ consists of exactly one point
$\alpha_*\in(1,+\infty]$. The third counterexample shows that, whenever
$0<a<1<b<\alpha_*$ and $\tau(\{a\})=\tau(\{b\})=0$,
\begin{equation}
    \int_{[0,a)}
    \frac{\alpha}{1-\alpha}\,\d\tau(\alpha)
    \leq
    \int_{(b,+\infty]}
    \frac{-\alpha}{1-\alpha}\,\d\tau(\alpha).
\end{equation}
Fix such a value of $b$. Since $\alpha_*$ is the only point in the
support above $1$, the right-hand side is the full upper integral and
is finite. It therefore gives a finite upper bound on the truncated
lower integrals as $a\uparrow1$. These integrals increase to the full
lower integral, which must consequently also be finite. Taking the
limit $a\uparrow1$ gives
\begin{equation}
    \int_{[0,1)}\frac{\alpha}{1-\alpha}\,\d\tau(\alpha)
    +\int_{(1,+\infty]}\frac{\alpha}{1-\alpha}\,\d\tau(\alpha)
    \leq0.
\end{equation}
The preceding argument proves that both integrals are finite, so
alternative~(b) follows.
\end{proof}

\subsection{Necessary conditions for the derivations}
\label{sec: necessary derivations}
Next, we turn to the necessary conditions for the derivations. Derivations are represented by positive integrals of the extremal quantities
\begin{equation}
    H_{0,\alpha}(P_{XY})
    =\sum_{y\in\Y}P_Y(y)H_\alpha(P_{X|Y=y}).
\end{equation}
Hence, for a finite positive measure $\tau$, it is sufficient to consider the concavity properties of the function
\begin{equation}
    \vec{x}\longmapsto
    \|\vec{x}\|_1
    \int_{[0,+\infty]}H_\alpha(\hat{x})\,\d\tau(\alpha).
\end{equation}
By Lemma~\ref{lem: conv derivations}, monotonicity of derivations requires concavity. We now show that concavity excludes every point $\alpha>1$ from the support of the representing measure. In particular, for an extremal derivation, it must be that $\alpha\in[0,1]$.

\begin{proposition}
\label{prop: general measure deriv}
Let $\tau$ be a finite positive measure on $[0,+\infty]$. The function
\begin{equation}
    \vec{x}\longmapsto
    \|\vec{x}\|_1
    \int_{[0,+\infty]}H_\alpha(\hat{x})\,\d\tau(\alpha)
\end{equation}
is concave only if $\operatorname{supp}(\tau)\subseteq[0,1]$.
\end{proposition}
\begin{proof}[Proof sketch]
We must show that the measure cannot have support in $(1,+\infty]$. Suppose that such support is present. Choose $r>1$ such that
\begin{equation}
    \tau(\{r\})=0,
    \qquad \tau((r,+\infty])>0,
\end{equation}
and define
\begin{equation}
    A_r:=\int_{(r,+\infty]}
    \frac{\alpha}{1-\alpha}\,\d\tau(\alpha)<0.
\end{equation}
We use the two-block point and direction
\begin{equation}
    \left(\vec{x}^{(d)}_r,
    \vec{v}^{(d)}_r(0,1)\right).
\end{equation}
For $\alpha<r$, the block that does not move gives the largest contribution. For $\alpha>r$, the moving block gives the largest contribution. Set
\begin{equation}
    s_d:=\|\vec{x}^{(d)}_r\|_1>0.
\end{equation}
Substituting the point and direction into the derivation, differentiating twice at $\lambda=0$, and dividing by this positive factor gives
\begin{equation}
    \frac1{s_d}\left.\frac{\d^2}{\d\lambda^2}
    \left[
    \|\vec{x}^{(d)}_r+\lambda\vec{v}^{(d)}_r(0,1)\|_1
    \int_{[0,+\infty]}
    H_\alpha\!\left(
    \widehat{\vec{x}^{(d)}_r+\lambda\vec{v}^{(d)}_r(0,1)}
    \right)\,\d\tau(\alpha)
    \right]\right|_{\lambda=0}
    \longrightarrow -(\log e)A_r>0.
\end{equation}
The factor \(\log e\) comes from the chosen logarithmic base of the R\'enyi entropies and is positive, so it does not affect the sign. The factor $s_d$ is independent of $\lambda$, and division by a positive number does not change the sign of the second derivative. Thus, for sufficiently large $d$, the function has a convex direction, contradicting concavity. The same calculation includes a point at $\alpha=+\infty$, where $\alpha/(1-\alpha)=-1$. Hence $\operatorname{supp}(\tau)\subseteq[0,1]$.
\end{proof}

\section{Large-sample conditional majorization}
\label{sec: large sample}
In this section, we present the large-sample result that provides the ordering
input for the general form of a conditional entropy in
Theorem~\ref{th:barycenter conditional entropies}.
It follows from the general preordered-semiring theorem,
Theorem~\ref{thm:Fritz2022}, applied to conditional majorization, together with
the characterisation of monotone homomorphisms and derivations in
Section~\ref{sec: semiring of conditional entropies}.
As shown in the following section (see Theorem~\ref{th:barycenter conditional entropies}), these homomorphisms and derivations coincide, up to a constant and a logarithmic factor, with the extremal conditional entropies that generate the entire family of conditional entropies through convex combinations.
Since conditional entropies, according to Definition~\ref{def: conditional entropy}, require proper normalization and additivity under tensor products, we reformulate the monotone homomorphisms from Proposition~\ref{prop: monotone homomorphisms} and the derivations from Proposition~\ref{prop: monotone derivations} as follows.

\begin{definition}
\label{def: extremal entropies}
Let $P_{XY}$ be a probability distribution, $t\in\R\setminus\{0\}$, and $\tau:\mc B\big([0,+\infty]\big)\to\mb R$ be a probability measure. Then, we define
\begin{align}
    &H_{t,\tau}(X|Y)_P = \frac{1}{t}\log \left(\sum_{y \in \Y} P_Y(y)\exp{\left(t\int_{[0,+\infty]}H_\alpha(P_{X|Y=y})\d\tau(\alpha)\right)}\right) \,, \label{eq:TempEnt}\\
    & H_{-\infty,\tau}(X|Y)_P =  \min \limits_{\substack{y \in \mathcal{Y} \\ P_Y( y) > 0}}\int_{[0,+\infty]}H_\alpha(P_{X|Y=y})\d\tau(\alpha) \,, \label{eq:TropEnt}\\
    & H_{0,\alpha}(X|Y)_P =  \sum_{y \in \Y} P_Y(y)H_\alpha(P_{X|Y=y})\,,\label{eq:DerivEnt} \\
    & H_{+\infty,0}(X|Y)_P =  \max \limits_{\substack{y \in \mathcal{Y} \\ P_Y( y) > 0}} H_0(P_{X|Y=y}) \,. \label{eq:TropEntno}
\end{align}
\end{definition}
The fact that \(H_{t,\tau}\), \(H_{-\infty,\tau}\), and \(H_{+\infty,0}\) are also extremal, in addition to the \(H_{0,\alpha}\) considered above---that is, that they cannot be represented as nontrivial combinations of the other quantities---is a consequence of the general argument presented in~\cite{haapasalo_inprep}. Throughout this manuscript, we refer to \(H_{t,\tau}\) as the bulk conditional entropies, as they are defined for finite values of the parameter \(t\). This distinguishes them from the entropies \(H_{-\infty,\tau}\), \(H_{+\infty,0}\), and \(H_{0,\alpha}\), which arise as limiting cases of $H_{t,\tau}$ corresponding to \(t \to -\infty\), \(t \to +\infty\), and \(t \to 0\), respectively.

As we prove later in Theorem~\ref{th:barycenter conditional entropies}, the above quantities encompass all extremal conditional entropies for appropriate values of $t$ and $\tau$. Not every parameter choice is admissible: for some choices the resulting quantity violates the axioms of Definition~\ref{def: conditional entropy}, in particular monotonicity under conditionally mixing channels.  Consequently, we introduce the following sets of parameters and measures.

\begin{definition}
\label{def: D}
Let \(\D_{\rm{bulk}}\) be the set of pairs $(t,\tau)$ consisting of $t\in\R\setminus\{0\}$ and probability measures $\tau:\mc B\big([0,+\infty]\big)\to\mb R$ satisfying one of the following conditions
\begin{enumerate}[itemsep=0pt,topsep=5pt]
    \item $t > 0, \;\textup{supp}(\tau) \subseteq [0,1], \tau(\{1\})=0$ and
    \begin{align}
    \label{eq: main integral in D}
        \int_{[0,+\infty]\setminus \{1\}}\frac{\alpha}{1-\alpha}\,\mathrm{d}\tau(\alpha)\leq\frac{1}{t}\,.
    \end{align}
    \item $t < 0$ and
    \begin{itemize}
    \item[(a)] $\;\textup{supp}(\tau) \subseteq [0,1]$ or
    \item[(b)] $\;\operatorname{supp}(\tau) \subseteq [0,1] \cup \{\alpha_*\}$, $\tau(\{1\})=0$ for some $\alpha_* \in (1, +\infty]$ such that~\eqref{eq: main integral in D} is satisfied.
    \end{itemize}
\end{enumerate}
\end{definition}

\begin{definition}
\label{def: D_infty}
 Let \(\D_{-\infty}\) be the set of probability measures $\tau:\mc B\big([0,+\infty]\big)\to\mb R_+$ satisfying one of the following conditions
    \begin{enumerate}[itemsep=0pt,topsep=5pt]
    \item  $\;\textup{supp}(\tau) \subseteq [0,1]$.
    \item $\;\operatorname{supp}(\tau) \subseteq [0,1] \cup \{\alpha_*\}$, $\tau(\{1\})=0$ for some $\alpha_* \in (1, +\infty]$ such that
    \begin{equation}
    \int_{[0,+\infty]\setminus\{1\}}\frac{\alpha}{1-\alpha}\,\mathrm{d}\tau(\alpha) \leq 0.
    \end{equation}
\end{enumerate}
\end{definition}
We note that, whenever it is imposed, the condition $\tau(\{1\})=0$
excludes an atom at $\alpha=1$, while still allowing the measure to
accumulate at $1$. This differs from requiring
$\text{supp}(\tau)\subseteq[0,1)$, which would also exclude such accumulation,
since $1$ would then belong to the topological support.

We now state the resulting criterion for eventual conversion of tensor powers
under the operations that define conditional majorization.

\begin{theorem}
\label{th:large-sample-conditional-majorization}
    Let $P_{XY}$ and $Q_{X'Y'}$ be joint probability distributions. If
    \begin{align}
    & H_{t,\tau}(X|Y)_P < H_{t,\tau}(X'|Y')_Q, \qquad\qquad\qquad   \forall (t,\tau) \in \D_{\rm{bulk}},\label{eq:TempCond} \\
    &H_{-\infty,\tau}(X|Y)_P < H_{-\infty,\tau}(X'|Y')_Q, \; \quad\qquad \;\;\, \forall \tau\in \D_{-\infty},\label{eq:TropCond} \\
    &H_{0,\alpha}(X|Y)_P<H_{0,\alpha}(X'|Y')_Q,  \quad\,\quad\qquad\quad\;\; \forall \alpha\in[0,1]\,,\label{eq:DerivCond} \\
    &H_{+\infty,0}(X|Y)_P < H_{+\infty,0}(X'|Y')_Q\,,
    \end{align}
    then $P_{XY}^{\otimes n} \succeq Q_{X'Y'}^{\otimes n}$ for every
    sufficiently large $n\in\mathbb N$.
\end{theorem}
\begin{proof}
The result follows directly from an application of the general large-sample theorem for preordered semirings stated in Theorem~\ref{thm:Fritz2022} along with the specific characterisation of monotone homomorphisms and derivations for the conditional majorization semiring in Proposition~\ref{prop: monotone homomorphisms} and~\ref{prop: monotone derivations}. Indeed, the entropies appearing in the assumptions coincide with the homomorphisms and derivations up to additive constants and logarithmic factors, which preserve the monotonicity properties. Especially for the derivations, note that, according to Proposition \ref{prop: monotone derivations} and~\ref{prop: general measure deriv}, all derivations at $\|\cdot\|$ are spanned by $H_{0,\alpha}$, $\alpha\in[0,1]$, so having strict inequalities over all these derivations is equivalent to having the strict inequalities of \eqref{eq:DerivCond}. Hence, by appropriately accounting for the direction of the inequalities, the assumptions of the theorem coincide with those required in Theorem~\ref{thm:Fritz2022} on the ordering of homomorphisms and derivations.

In addition, Theorem~\ref{thm:Fritz2022} requires the input distribution $Q_{X'Y'}$ to be power universal. We note that here \(Q_{X'Y'}\) is the input distribution, rather than \(P_{XY}\), since the preorder of the semiring \(\rgeq\) is reversed relative to the conditional majorization \(\succeq\) according to the definition of the preorder of the conditional majorization semiring in item~\ref{item: preorder} of Section~\ref{sec: semiring of conditional entropies}.
The fact that $Q_{X'Y'}$ is power universal is already guaranteed by the above assumptions on the strict ordering of entropies. Indeed, if $Q_{X'Y'}$ were not power universal, according to Proposition~\ref{prop: power universals}, there would exist a column of $Q_{X'Y'}$ with exactly one nonzero entry.
In this case, the Hayashi conditional entropy, which corresponds to the case $\alpha=\infty$ for $t=-\infty$---a well-known fact in the literature (see, e.g., \cite{gour2024inevitability})--- would be zero for $Q_{X'Y'}$. Indeed,
\begin{equation}
    H^{\rm{H}}_{\infty}(X'|Y')_Q = \min_{y \in \mathcal{Y}'}  H_{\infty}\bigl(Q_{X'|Y'=y'}\bigr) \, = 0.
\end{equation}
Since $ H^{\rm{H}}_{\infty}(X|Y)_P\geq0$, the strict inequalities assumed in the theorem imply that $Q_{X'Y'}$ must be power universal.
\end{proof}

 \section{The set of conditional entropies}
\label{sec: set of conditional entropies}
In this section, we present the main result of our work, namely the general characterisation of conditional entropies in Theorem \ref{th:barycenter conditional entropies}. Specifically, we show that any conditional entropy defined as a function satisfying the axioms in Definition~\ref{def: conditional entropy} can be expressed as a (generalized) convex combination of the entropies $H_{t,\tau}$ and its limiting forms $H_{0,\alpha}$, $H_{-\infty,\tau}$ and $H_{+\infty,0}$ defined in~\ref{def: extremal entropies}. In this sense, the conditional entropies $H_{t,\tau},H_{0,\alpha},H_{-\infty,\tau}$ and $H_{+\infty,0}$ constitute the extremal points of the set of the conditional entropies; for the extremality, see Theorem 14 of \cite{haapasalo_inprep}.


\subsection{The general form of conditional entropies}

Let us note that the set of conditional entropies characterized in Definition \ref{def: conditional entropy} is convex. This means that, whenever $\mb H_1$ and $\mb H_2$ are conditional entropies (i.e., satisfy the postulates of Definition \ref{def: conditional entropy}) and $\lambda\in[0,1]$, then also $\mb H=\lambda\mb H_1+(1-\lambda)\mb H_2$,
\begin{equation}
\mb H(X|Y)_P=\lambda\mb H_1(X|Y)_P+(1-\lambda)\mb H_2(X|Y)_P\qquad\forall\, P=P_{XY}\,,
\end{equation}
is a conditional entropy. Especially all convex combinations of the conditional entropies $H_{t,\tau}$, $(t,\tau)\in\mc D_{\rm bulk}$, and $H_{-\infty,\tau}$, $\tau\in\mc D_{-\infty}$, and their pointwise limits are also conditional entropies. This fact also extends to the generalized convex combinations of these quantities, i.e.,\ barycentres over the set of these specific conditional entropies with respect to any (inner and outer regular) probability measure. For the latter statement, we need the fact that the set of these (extremal) conditional entropies is a compact Hausdorff space with respect to the topology of pointwise convergence \cite[Proposition 8.5]{fritz2023abstractII}. The following theorem characterizing the total set of conditional entropies as described in Definition \ref{def: conditional entropy} tells us that any conditional entropy can be expressed as a generalized convex combination of the special conditional entropies characterized in this work.

We recall that $H_{t,\tau}$, $H_{-\infty,\tau}$, $H_{0,\alpha}$, and $H_{+\infty,0}$ are introduced in Definition~\ref{def: extremal entropies}, while the sets $\mathcal{D}_{\mathrm{bulk}}$ and $\mathcal{D}_{-\infty}$ are defined in Definitions~\ref{def: D} and~\ref{def: D_infty}, respectively.
\begin{theorem}[General form of conditional entropies]
\label{th:barycenter conditional entropies}
Let $\H$ be a conditional entropy in the sense of Definition~\ref{def: conditional entropy}.
Then $\H$ can be represented as a generalized convex combination of the conditional entropies $H_{t,\tau}, H_{-\infty,\tau}, H_{0,\alpha}$ and $H_{+\infty,0}$, with parameters $(t,\tau)\in\mathcal{D}_{\mathrm{bulk}}$, $\tau\in\mathcal{D}_{-\infty}$,  and $\alpha\in[0,1]$, respectively, for the first three families.
\end{theorem}
\begin{proof}
We do not provide all details of the proof here, as this result is proven analogously to Proposition \ref{proposition:barycenter entropies} and it is a special case of Theorem 7 of \cite{haapasalo_inprep}. However, we give an overview of the main steps for completeness. Let us denote by $\mc H$ the set of conditional entropies $H_{t,\tau}$, $H_{-\infty,\tau}$, $H_{0,\alpha}$, and $H_{+\infty,0}$ as defined in \eqref{eq:TempEnt}, \eqref{eq:TropEnt}, \eqref{eq:DerivEnt} and~\eqref{eq:TropEntno} of Definition~\ref{def: extremal entropies}, where the parameters and measures satisfy the required monotonicity conditions. According to Proposition 8.5 of \cite{fritz2023abstractII}, this set is a compact Hausdorff space when equipped with the topology of pointwise convergence. Just as in the proof of Proposition \ref{proposition:barycenter entropies}, now using Theorem \ref{th:large-sample-conditional-majorization}, we may easily prove that
\begin{equation}
\Gamma(X|Y)_P \leq \Gamma(X'|Y')_Q,
\, \forall \Gamma \in \mc H \, \implies \, \H(X|Y)_P \leq \H(X'|Y')_Q, \, \forall \H \;\text{satisfying Definition~\ref{def: conditional entropy}} \,.
\end{equation}
Using this as a stepping stone, we may define a positive linear functional $G$ on the vector subspace of the set $C(\mc H)$ of continuous real functions on the compact Hausdorff $\mc H$ generated by the evaluation functions ${\rm ev}_{P}$, ${\rm ev}_P(\Gamma)=\Gamma(X|Y)_P$, through $G({\rm ev}_P)=\mb H(X|Y)_P$. Using Kantorovich's theorem we may extend $G$ into a positive linear functional $I:C(\mc H)\to\R$ and the Riesz-Markov-Kakutani theorem implies the existence of a finite positive measure $\nu:\mc B(\mc H)\to\R_+$ (Borel $\sigma$-algebra of the compact Hausdorff space $\mc H$) which is inner and outer regular such that
\begin{equation}
\mb H(X|Y)_P=\int_{\mc H}\Gamma(X|Y)_P\,\mathrm{d}\nu(\Gamma)\,.
\end{equation}
This measure is easily seen to be a probability measure.
\end{proof}

\subsection{Conditional entropies in the literature}
\label{sec: literature}
In this section, we review in detail several conditional entropies that have been introduced in the literature and show how each of these quantities arises as a special case of the general form established in Theorem~\ref{th:barycenter conditional entropies}.

The conditional Shannon entropy $H(X|Y)$,
\begin{equation}
    H(X|Y)_P=-\sum_{x \in \X,y\in\Y} P_{XY}(x,y)\log{\frac{P_{XY}(x,y)}{P_Y(y)}} \,,
\end{equation}
is equal to \(H_{0,\alpha}\) for \(\alpha = 1\).
The Hayashi conditional entropy introduced by Hayashi and Skori\'c et al.~\cite{hayashi2011exponential,vskoric2011sharp}
\begin{equation}
H^\text{H}_{\alpha}(X|Y)_P=\frac{1}{1-\alpha}\log{\Big(\sum_{x\in\X,y\in\Y} P_Y(y) P_{X|Y}(x|y)^\alpha\Big)} \,,
\end{equation}
arises as a special case of $H_{t,\tau}$ when $\tau$ is the point at $\alpha$ and $t = 1 - \alpha$. The Arimoto R\'enyi conditional entropy~\cite{arimoto1977information} 
\begin{align}
H^\text{A}_{\alpha}(X|Y)_P=\frac{\alpha}{1-\alpha}\log{\bigg(\sum_{y\in \Y} P_Y(y)\Big(\sum_{x \in \X} P_{X|Y}(x|y)^\alpha\Big)^\frac{1}{\alpha}\bigg)} \,,
\end{align}
is recovered similarly as a special case of $H_{t,\tau}$ when $\tau$ is the point measure at $\alpha$ and $t = \frac{1 - \alpha}{\alpha}$. A two-parameter family that interpolates between the Arimoto and Hayashi conditional entropies was introduced in~\cite{hayashi2016equivocations} and later investigated in detail in~\cite{rubboli2024quantum}
\begin{align}
\label{eq: H alpha beta}
H_{\alpha,\beta}(X|Y)_P=\frac{\alpha}{1-\alpha}\frac{1}{\beta}\log{\bigg(\sum_{y \in \Y} P_Y(y)\Big(\sum_{x \in \X} P_{X|Y}(x|y)^\alpha\Big)^\frac{\beta} {\alpha}\bigg)} \,.
\end{align}
This quantity is a special case of $H_{t,\tau}$ when $\tau$ is the point measure at $\alpha$ and $t = \frac{1 - \alpha}{\alpha} \beta$. More recently, this quantity has been studied in~\cite{li2025two}.

More exotic quantities have appeared in the literature. In~\cite{tan2018conditional}, the authors defined the conditional entropy
\begin{equation}
    H_{1+a|1+b}(X|Y)_P=-\frac{1+b}{a}\log{\bigg(\sum_{y \in \Y} P_Y(y)\Big(\sum_{x \in \X} P_{X|Y}(x|y)^{1+a}\Big)\Big(\sum_{x \in \X} P_{X|Y}(x|y)^{1+b}\Big)^{-\frac{a}{1+b}}\bigg)} \,.
\end{equation}
This corresponds to $H_{t,\tau}$ for the two-point measure $\tau=(1+b)\delta_{1+a}-b\delta_{1+b}$ and $t=-a/(1+b)$.
An alternative conditional entropy was introduced in~\cite{hayashi2016uniform}\footnote{Note that the definition in~\cite[Equation (14)]{hayashi2016uniform} contains a typo. The correct exponent for the term $\sum_{x \in \X} P_{X|Y}(x|y)^{1+\theta'}$ is $-\theta/(1+\theta')$. This corrected formulation, which includes the minus sign, is used consistently throughout the remainder of their paper (see, e.g.,~\cite[Lemma 3, Property (11)]{hayashi2016uniform}).}
\begin{align}
    H_{1+\theta,1+\theta'}(X|Y)_P=&-\frac{1}{\theta}\log{\bigg(\sum_{y \in \Y} P_Y(y)\Big(\sum_{x \in \X} P_{X|Y}(x|y)^{1+\theta}\Big)\Big(\sum_{x \in \X} P_{X|Y}(x|y)^{1+\theta'}\Big)^{-\frac{\theta}{1+\theta'}}\bigg)} \\
    &\qquad + \frac{\theta'}{1+\theta'}H^{\text{A}}_{1+\theta'}(P_{XY}) \\
    =&\frac{1}{1+\theta'}H_{t_1,\tau_1}(X|Y)_P+\frac{\theta'}{1+\theta'}H_{t_2,\tau_2}(X|Y)_P\,,
\end{align}
Here, $\tau_1=(1+\theta')\delta_{1+\theta}-\theta'\delta_{1+\theta'}$, $t_1=-\theta/(1+\theta')$, $\tau_2$ is the point measure at $1+\theta'$, and $t_2=-\theta'/(1+\theta')$.

The conditional entropy
 \begin{align}
H_{0,\alpha}(P)=\sum_{y\in\mc Y} P_Y(y)H_\alpha(P_{X|Y=y})\,
 \end{align}
 was introduced by Cachin~\cite{cachin1997entropy} and later studied in~\cite{teixeira2012conditional}.

In~\cite{renner2005simple},
as a part of a study of conditional smooth R\'enyi entropies, the authors introduced the quantity
\begin{align}
    H^{\text{R}}_\alpha(X|Y) = \min _{y} H_\alpha(X|Y=y)
\end{align}
for $\alpha > 1$ and with $\min$ replaced by $\max$ for $\alpha \in (0,1)$. In the range $\alpha>1$, it coincides with $H_{-\infty,\tau}$ when $\tau$ is a point measure at $\alpha$.
For $\alpha \in (0,1)$, it coincides with the limit $t\rightarrow +\infty$ of $H_{t,\tau}$ in the case where $\tau$ is a point measure at $\alpha$. However, although this limit is true for all $\alpha \in (0,1)$, Theorem~\ref{th:barycenter conditional entropies} implies that these quantities satisfy all our axioms required for a conditional entropy  (Definition~\ref{def: conditional entropy}) only in the case $\alpha = 0$.
In this case, it recovers \(H_{+\infty,0}\).  

Finally, the quantities $ H_{+\infty,0}$ and $H_{-\infty,\tau}$ for single-point measures have been studied as limiting cases of
$H_{\alpha,\beta}$~\eqref{eq: H alpha beta} in~\cite{li2025two} and~\cite{rubboli2024quantum}.

\section*{Acknowledgements}
We thank Frits Verhagen for helpful discussions.

The results presented in the first version of this manuscript,
posted on arXiv in January 2026, were obtained entirely without
AI assistance. During preparation of the second version, repeated
interactions with ChatGPT 5.6 Sol (OpenAI) helped us develop the
arguments in Section~V establishing the necessary parameter
conditions without assuming the additional integrability conditions
imposed on the parameter measure in Assumption~1 of the first
version. These arguments resolve Conjecture~VII.5 of that version.
ChatGPT 5.6 Sol was also used to assist with the Lean~4 formalization
of the results in Sections~IV and~V
\cite{rubboli2026parametersFull,rubboli2026parametersLean,rubboli2026lean}.
The authors take full responsibility for the mathematical arguments,
the formalization, and the content of this manuscript.

\bibliographystyle{IEEEtran}
\bibliography{bibliography}

@article{fehr14conditional,
  author={Fehr, Serge and Berens, Stefan},
  journal={IEEE Transactions on Information Theory},
  title={On the Conditional {R{\'e}nyi} Entropy},
  year={2014},
  volume={60},
  number={11},
  pages={6801-6810},
  doi={10.1109/TIT.2014.2357799}}

@inproceedings{renyi61,
   author = {R{\'e}nyi, Alfr{\'e}d},
   city = {Berkeley, California, USA},
   booktitle = {Proc. 4th Berkeley Symposium on Mathematical Statistics and Probability},
   pages = {547--561},
   publisher = {University of California Press},
   title = {On Measures of Entropy and Information},
   volume = {1},
   year = {1961},
}

@article{gour21_axiomatic,
   author = {Gilad Gour and Marco Tomamichel},
   doi = {10.1109/TIT.2021.3078337},
   issn = {0018-9448},
   issue = {10},
   journal = {IEEE Transactions on Information Theory},
   month = {10},
   pages = {6313-6327},
   title = {Entropy and Relative Entropy From Information-Theoretic Principles},
   volume = {67},
   year = {2021},
}

@article{mu21_renyi,
   author = {Xiaosheng Mu and Luciano Pomatto and Philipp Strack and Omer Tamuz},
   doi = {10.3982/ECTA17548},
   issn = {0012-9682},
   issue = {1},
   journal = {Econometrica},
   pages = {475-506},
   title = {From {Blackwell} Dominance in Large Samples to {R{\'e}nyi} Divergences and Back Again},
   volume = {89},
   year = {2021},
}

@article{vskoric2011sharp,
  title={Sharp lower bounds on the extractable randomness from non-uniform sources},
  author={{\v{S}}kori{\'c}, Boris and Obi, Chibuzo and Verbitskiy, Evgeny and Schoenmakers, Berry},
  journal={Information and Computation},
  volume={209},
  number={8},
  pages={1184--1196},
  year={2011},
  publisher={Elsevier},
  doi={10.1016/j.ic.2011.06.001}
}

@article{arimoto1977information,
  title={Information measures and capacity of order $\alpha$ for discrete memoryless channels},
  author={Arimoto, Suguru},
  journal={Topics in information theory},
  year={1977},
  publisher={The Netherlands}
}

@article{hayashi2011exponential,
  title={Exponential decreasing rate of leaked information in universal random privacy amplification},
  author={Hayashi, Masahito},
  journal={IEEE Transactions on Information Theory},
  volume={57},
  number={6},
  pages={3989--4001},
  year={2011},
  publisher={IEEE},
  doi={10.1109/TIT.2011.2110950}
}

@article{gour2024inevitability,
  title={Inevitability of knowing less than nothing},
  author={Gour, Gilad and Wilde, Mark M and Brandsen, Sarah and Geng, Isabelle Jianing},
  journal={Quantum},
  volume={8},
  pages={1529},
  year={2024},
  publisher={Verein zur F{\"o}rderung des Open Access Publizierens in den Quantenwissenschaften},
  doi={10.22331/q-2024-11-20-1529}
}

@article{gour2018conditional,
  title={Conditional uncertainty principle},
  author={Gour, Gilad and Grudka, Andrzej and Horodecki, Micha{\l} and K{\l}obus, Waldemar and {\L}odyga, Justyna and Narasimhachar, Varun},
  journal={Physical Review A},
  volume={97},
  number={4},
  pages={042130},
  year={2018},
  publisher={APS},
  doi={10.1103/PhysRevA.97.042130}
}

@book{bullen2013handbook,
  title={Handbook of means and their inequalities},
  author={Bullen, Peter S},
  volume={560},
  year={2013},
  publisher={Springer Science \& Business Media},
  doi={10.1007/978-94-017-0399-4}
}

@ARTICLE{Jensen_Kjaerulf_2019,
  author={Jensen, Asger Kjærulff},
  journal={IEEE Transactions on Information Theory},
  title={Asymptotic Majorization of Finite Probability Distributions},
  year={2019},
  volume={65},
  number={12},
  pages={8131-8139},
  doi={10.1109/TIT.2019.2922627}}

@article{Kantorovich,
    author = {Kantorovich, L.V.},
    title = {On the moment problem for a finite interval},
    journal = {Doklady Akademii Nauk SSSR},
    year = {1937},
    volume = {14},
    pages = {531-537}
}

@article{fritz2023abstract,
  title={Abstract Vergleichsstellens{\"a}tze for preordered semifields and semirings I},
  author={Fritz, Tobias},
  journal={SIAM Journal on Applied Algebra and Geometry},
  volume={7},
  number={2},
  pages={505--547},
  year={2023},
  publisher={SIAM}
}

@article{fritz2023abstractII,
  title={Abstract Vergleichsstellens{\"a}tze for preordered semifields and semirings II},
  author={Fritz, Tobias},
  journal={arXiv:2112.05949},
  year={2022}
}

@misc{haapasalo_inprep,
    author={Haapasalo, Erkka},
    title={Barycentric decompositions for extensive monotone divergences},
    journal={arXiv:2509.18725},
    year={2025},
    url={https://arxiv.org/abs/2509.18725}
}

@article{birkhoff1946tres,
  title={Tres observaciones sobre el algebra lineal},
  author={Birkhoff, Garrett},
  journal={Univ. Nac. Tucuman, Ser. A},
  volume={5},
  pages={147--154},
  year={1946}
}

@book{marshall1979inequalities,
  author={Marshall, Albert W. and Olkin, Ingram and Arnold, Barry C.},
  title={Inequalities: Theory of Majorization and Its Applications},
  edition={2},
  series={Springer Series in Statistics},
  publisher={Springer},
  address={New York},
  year={2011},
  isbn={978-0-387-40087-7},
  doi={10.1007/978-0-387-68276-1}
}

@article{rubboli2024quantum,
  title={Quantum Conditional Entropies},
  author={Rubboli, Roberto and Goodarzi, Milad M and Tomamichel, Marco},
  journal={arXiv preprint:2410.21976},
  year={2024}
}

@article{hayashi2016uniform,
  title={Uniform random number generation from {Markov} chains: Non-asymptotic and asymptotic analyses},
  author={Hayashi, Masahito and Watanabe, Shun},
  journal={IEEE Transactions on Information Theory},
  volume={62},
  number={4},
  pages={1795--1822},
  year={2016},
  publisher={IEEE},
  doi={10.1109/TIT.2016.2530084}
}

@article{shannon1948mathematical,title={A mathematical theory of communication},
  author={Shannon, Claude E},
  journal={The Bell system technical journal},
  volume={27},
  number={3},
  pages={379--423},
  year={1948},
  publisher={Nokia Bell Labs},
  doi={10.1002/j.1538-7305.1948.tb01338.x}
}

@article{farooq2024matrix,
  title={Matrix majorization in large samples},
  author={Farooq, Muhammad Usman and Fritz, Tobias and Haapasalo, Erkka and Tomamichel, Marco},
  journal={IEEE Transactions on Information Theory},
  volume={70},
  number={5},
  pages={3118--3144},
  year={2024},
  publisher={IEEE},
  doi={10.1109/TIT.2024.3352088}
}

@article{hayashi2016equivocations,
  title={Equivocations, exponents, and second-order coding rates under various {R{\'e}nyi} information measures},
  author={Hayashi, Masahito and Tan, Vincent YF},
  journal={IEEE Transactions on Information Theory},
  volume={63},
  number={2},
  pages={975--1005},
  year={2016},
  publisher={IEEE}
}

@phdthesis{cachin1997entropy,
  author={Cachin, Christian},
  title={Entropy measures and unconditional security in cryptography},
  school={ETH Zurich},
  year={1997},
  address= {},
  arxivid = {}
}

@article{teixeira2012conditional,
  title={Conditional {R{\'e}nyi} entropies},
  author={Teixeira, Andreia and Matos, Armando and Antunes, Luis},
  journal={IEEE Transactions on Information Theory},
  volume={58},
  number={7},
  pages={4273--4277},
  year={2012},
  publisher={IEEE}
}

@inproceedings{renner2005simple,
  title={Simple and tight bounds for information reconciliation and privacy amplification},
  author={Renner, Renato and Wolf, Stefan},
  booktitle={International conference on the theory and application of cryptology and information security},
  pages={199--216},
  year={2005},
  organization={Springer}
}

@article{li2025two,
  title={Two-Parameter {R{\'e}nyi} Information Quantities with Applications to Privacy Amplification and Soft Covering},
  author={Li, Shi-Bing and Li, Ke and Yu, Lei},
  journal={ arXiv:2511.02297},
  year={2025}
}

@article{hartley28,
author = {Hartley, R. V. L.},
journal = {Bell System Technical Journal},
month = {jul},
number = {3},
pages = {535--563},
title = {{Transmission of Information}},
volume = {7},
year = {1928}
}

@book{aczel75,
author = {Acz{\'{e}}l, J. and Dar{\'{o}}czy, Z.},
isbn = {9780080956244},
publisher = {Academic Press},
series = {Mathematics in Science and Engineering},
title = {{On Measures of Information and their Characterizations}},
volume = {115},
year = {1975}
}

@book{ebanks98,
author = {Ebanks, Bruce and Sahoo, Prasanna and Sander, Wolfgang},
doi = {10.1142/3354},
isbn = {978-981-02-3006-7},
month = {apr},
publisher = {World Scientific},
title = {{Characterizations of Information Measures}},
url = {https://www.worldscientific.com/worldscibooks/10.1142/3354},
year = {1998}
}

@article{csiszar08,
author = {Csisz{\'{a}}r, Imre},
doi = {10.3390/e10030261},
issn = {1099-4300},
journal = {Entropy},
month = {sep},
number = {3},
pages = {261--273},
title = {{Axiomatic Characterizations of Information Measures}},
url = {http://www.mdpi.com/1099-4300/10/3/261},
volume = {10},
year = {2008}
}

@article{tan2018conditional,
   author = {Vincent Y. F. Tan and Masahito Hayashi},
   doi = {10.1109/TIT.2018.2792495},
   issn = {0018-9448},
   issue = {5},
   journal = {IEEE Transactions on Information Theory},
   month = {5},
   pages = {3734-3755},
   title = {Analysis of Remaining Uncertainties and Exponents Under Various Conditional {R{\'e}nyi} Entropies},
   volume = {64},
   year = {2018}
}

@misc{rubboli2026parametersFull,
  author       = {Rubboli, Roberto and Haapasalo, Erkka and Tomamichel, Marco},
  title        = {Sufficient and Necessary Conditions for the Parameters of Conditional Entropies: Complete Proofs},
  howpublished = {Auxiliary document on GitHub},
  year         = {2026},
  url          = {https://github.com/robertorubboli/A-complete-characterization-of-conditional-entropies/blob/main/paper/auxiliary-files/pdf/sections-4-5-full-details.pdf},
  note         = {Complete proofs for Sections 4 and 5 of the canonical manuscript}
}

@misc{rubboli2026parametersLean,
  author       = {Rubboli, Roberto and Haapasalo, Erkka and Tomamichel, Marco},
  title        = {Lean Formalization of the Parameter Classification for Conditional Entropies},
  howpublished = {Auxiliary document on GitHub},
  year         = {2026},
  url          = {https://github.com/robertorubboli/A-complete-characterization-of-conditional-entropies/blob/main/paper/auxiliary-files/pdf/lean-formalization-note.pdf},
  note         = {{Lean} proof map, trust audit, and statement correspondence for Sections 4 and 5}
}

@misc{rubboli2026lean,
  author = {Rubboli, Roberto and Haapasalo, Erkka
            and Tomamichel, Marco},
  title = {{Lean} formalization accompanying
           {`A complete characterisation of conditional entropies'}},
  year = {2026},
  howpublished = {Zenodo},
  note = {Version 1.0.0. GitHub repository archived on Zenodo},
  doi = {10.5281/zenodo.22661629},
  url = {https://doi.org/10.5281/zenodo.22661629}
}

@book{khintchine1957mathematical,
  author     = {Khinchin, Aleksandr Ya.},
  title      = {Mathematical Foundations of Information Theory},
  translator = {Silverman, Richard A. and Friedman, M. D.},
  publisher  = {Dover Publications},
  address    = {New York},
  year       = {1957},
  isbn       = {978-0-486-60434-3}
}

@article{faddeev56,
  author  = {Faddeev, D. K.},
  title   = {On the Concept of Entropy of a Finite Probabilistic Scheme},
  journal = {Uspekhi Matematicheskikh Nauk},
  volume  = {11},
  number  = {1(67)},
  pages   = {227--231},
  year    = {1956},
  note    = {In Russian},
  url     = {https://www.mathnet.ru/eng/rm7756}
}

@article{daroczy1964,
  author  = {Dar{\'o}czy, Zolt{\'a}n},
  title   = {{\"U}ber Mittelwerte und Entropien vollst{\"a}ndiger Wahrscheinlichkeitsverteilungen},
  journal = {Acta Mathematica Academiae Scientiarum Hungaricae},
  volume  = {15},
  number  = {1--2},
  pages   = {203--210},
  year    = {1964},
  doi     = {10.1007/BF01897038}
}

@book{csiszar11,
  author    = {Csisz{\'a}r, Imre and K{\"o}rner, J{\'a}nos},
  title     = {Information Theory: Coding Theorems for Discrete Memoryless Systems},
  edition   = {2},
  publisher = {Cambridge University Press},
  address   = {Cambridge, U.K.},
  year      = {2011},
  isbn      = {978-0-521-19681-9},
  doi       = {10.1017/CBO9780511921889}
}

@article{blackwell1953,
  author  = {Blackwell, David},
  title   = {Equivalent Comparisons of Experiments},
  journal = {The Annals of Mathematical Statistics},
  volume  = {24},
  number  = {2},
  pages   = {265--272},
  year    = {1953},
  doi     = {10.1214/aoms/1177729032}
}

@article{shannon1958partial,
  author  = {Shannon, Claude E.},
  title   = {A Note on a Partial Ordering for Communication Channels},
  journal = {Information and Control},
  volume  = {1},
  number  = {4},
  pages   = {390--397},
  year    = {1958},
  doi     = {10.1016/S0019-9958(58)90239-0}
}

@article{nasser2018channelorderings,
  author  = {Nasser, Rajai},
  title   = {Characterizations of Two Channel Orderings: Input-Degradedness and the {Shannon} Ordering},
  journal = {IEEE Transactions on Information Theory},
  volume  = {64},
  number  = {10},
  pages   = {6759--6770},
  month   = oct,
  year    = {2018},
  doi     = {10.1109/TIT.2018.2859252}
}

@book{torgersen1991comparison,
  author    = {Torgersen, Erik},
  title     = {Comparison of Statistical Experiments},
  series    = {Encyclopedia of Mathematics and its Applications},
  volume    = {36},
  publisher = {Cambridge University Press},
  address   = {Cambridge, U.K.},
  year      = {1991},
  isbn      = {978-0-521-25030-6},
  doi       = {10.1017/CBO9780511666353}
}

\newpage
\appendix

\section{Convexity and monotonicity}
\label{app: convexity and monotonicity}
In this appendix, we present several known equivalences between convexity and monotonicity under channels. These will be useful in the main text to characterize the parameter regimes for which the conditional entropy is monotone under conditionally mixing channels, and thus qualifies as a valid conditional entropy according to Definition~\ref{def: conditional entropy}. In particular, these equivalences allow us to reduce the question of monotonicity to that of convexity, which is often easier to analyze.

In Definition~\ref{def: conditional entropy}, we introduced conditional entropies as quantities satisfying a collection of operational axioms, motivated by the requirement that they quantify conditional uncertainty. A central axiom is monotonicity under conditionally mixing channels.
Within the mathematical framework of preordered semirings, monotone homomorphisms and derivations are functions designed to satisfy closely analogous structural properties (see Section~\ref{sec: background semiring}). In particular, in the semiring associated with conditional majorization, these functions correspond---up to an additive constant and a logarithmic factor---to (extremal) conditional entropies.
Consequently, understanding their monotonicity properties is essential. In the following, we show that determining whether these functions are monotone under conditionally mixing channels is equivalent to studying their convexity properties.

We begin by recalling several notions of convexity for functions defined on
(unnormalized) joint distributions~$\mathcal V$, which---up to embedding in a larger space (padding with zeros)---
can be identified with the elements of the semiring of conditional majorization
introduced in Section~\ref{sec: semiring of conditional entropies}.
 We say that a map $F:\V\to\mb R$ is convex if, for any $P,Q\in\V$ of the same size and $t\in[0,1]$, we have
\begin{equation}
F\big(tP+(1-t)Q\big)\leq tF(P)+(1-t)F(Q).
\end{equation}
We say that $F$ is concave if the above inequality is reversed. We say that $F$ is quasi-convex if, for $P,Q\in\V$ of the same size and $t\in(0,1)$,
\begin{equation}
F\big(tP+(1-t)Q\big)\leq\max\{F(P),F(Q)\}.
\end{equation}
We say that $F$ is strongly quasi-concave if the inequality above is reversed, i.e.,
\begin{equation}
\label{eq: def strongly quasi-conc}
    F\big(tP+(1-t)Q\big)\geq\max\{F(P),F(Q)\}.
\end{equation}

 In the following, we show that the required convexity properties depend on the output space of the homomorphism. As used in Sections~\ref{sec:sufficient-conditions-measure} and~\ref{sec:necessary-conditions-measures} and detailed in the companion note~\cite{rubboli2026parametersFull}, the key monotonicity is monotonicity under maps acting on the conditioning subsystem of a joint distribution, which form the key subclass of conditionally mixing channels. In matrix form, for a positive joint distribution $P_{XY}$, such an action is represented as $P_{XY}\rightarrow P_{XY}D$, where $D$ is row stochastic.

\begin{proposition}\label{prop:convconc}
Let $S$ be the semiring of conditional majorization of Section~\ref{sec: semiring of conditional entropies}. Let
$\Phi: \cS \to \mathbb{K}$ be a candidate monotone homomorphism, and hence of the form in Proposition~\ref{prop: monotone homomorphisms}. Here, $\mathbb{K} \in \{\mathbb{R}_+,  \mathbb{R}_+^{\mathrm{op}},\mathbb{TR}_+,\mathbb{TR}_+^{\mathrm{op}}\}$. Then, monotonicity with respect to maps acting on the conditioning system is equivalent to the following convexity properties, depending on the set $\mathbb{K}$:
\begin{itemize}[itemsep=0pt, topsep=5pt]
\item If $\mathbb{K} =\mathbb{R}_+$, $\Phi$ is monotone if and only if it is concave.
\item If $ \mathbb{K}= \mathbb{R}_+^{\mathrm{op}}$, $\Phi$ is monotone if and only if it is convex.
\item If $ \mathbb{K}=\mathbb{TR}_+$, $\Phi$ is monotone if and only if it is strongly quasi-concave.
\item If $ \mathbb{K}=\mathbb{TR}_+^{\mathrm{op}}$, $\Phi$ is monotone if and only if it is quasi-convex.
\end{itemize}
\end{proposition}

\begin{proof}
Let us first prove the `only if' part of the claim. The proof relies on associating the distributions $tP$ and $(1 - t)Q$ with distinct outcomes of the conditioning variable $Y$. Monotonicity under discarding outcomes of $Y$ then yields the desired inequality. Since monotone homomorphisms are invariant under embedding (see the discussion in
Section~\ref{sec: invariance embedding and shifting}), we may, in the following,
identify their evaluation on equivalence classes with their evaluation on
unnormalized joint distributions.

Explicitly, suppose that $P,Q\in\V$ and $t\in(0,1)$. We now have
\begin{equation}\label{eq:affinechain}
\left(\begin{array}{cc}
tP&0\\
0&(1-t)Q
\end{array}\right)\succeq\left(\begin{array}{cc}
tP&(1-t)Q\\
0&0
\end{array}\right)\succeq\left(\begin{array}{cc}
tP+(1-t)Q&0\\
0&0
\end{array}\right).
\end{equation}
The first implication follows by applying a permutation on the $X$ register conditioned on the second $Y$ outcome, while the second follows by merging the two $Y$-outcomes and introducing a dummy outcome with zero probability.

In case $\mathbb{R}_+$, it now follows that
\begin{align}
\Phi\big(tP+(1-t)Q\big)&\geq\Phi\left[\left(\begin{array}{cc}
tP&0\\
0&(1-t)Q
\end{array}\right)\right]\\
&=\Phi(tP)+\Phi\big((1-t)Q\big)\\
&=\Phi(t)\Phi(P)+\Phi(1-t)\Phi(Q)\\
&=t\Phi(P)+(1-t)\Phi(Q).
\end{align}
Similarly, in case $\mathbb{R}_+^{\mathrm{op}}$, we get
\begin{align}
\Phi\big(tP+(1-t)Q\big)&\leq t\Phi(P)+(1-t)\Phi(Q) \,.
\end{align}
In case $\mathbb{TR}_+$, we get
\begin{align}
\Phi\big(tP+(1-t)Q\big)&\geq\max\{\Phi(P),\Phi(Q)\} \,.
\end{align}
In case $\mathbb{TR}^{\rm{op}}_+$, we have
\begin{align}
\Phi\big(tP+(1-t)Q\big)&\leq\Phi\left[\left(\begin{array}{cc}
tP&0\\
0&(1-t)Q
\end{array}\right)\right]\\
&=\max\{\Phi(tP),\Phi\big((1-t)Q\big)\}\\
&=\max\{\Phi(t)\Phi(P),\Phi(1-t)\Phi(Q)\}\\
&=\max\{\Phi(P),\Phi(Q)\}.
\end{align}

Next, we prove the ``if'' direction and assume that $\Phi$ is concave. Let $P \in \mathcal{V}$ and let $D$ be a row-stochastic matrix. We aim to show that $\Phi(PD) \geq \Phi(P)$. Let us first assume that $P=R=R_{X:YZ}$ (colon separating the primary $X$ from the conditioning $YZ$) and that $D$ is the map ignoring $Z$. 
We have
\begin{align}
\Phi(R)&=\sum_{z\in\mc Z}R_Z(z)\Phi(R_{XY|Z=z})\leq\Phi\left(\sum_{z\in\mc Z}R_Z(z)R_{XY|Z=z}\right)=\Phi(R_{XY}),
\end{align}
where we have used the form of \eqref{eq:FinalForm} in the first equality and the assumed concavity of $\Phi$ in the inequality. This proves the first special case. Assume then the case $P_{XY}\succeq P_{XY}D=:Q_{XY'}$ where $D$ is a row-stochastic matrix which has only entries 0 or 1. This means that the matrix entries of $D$ are given by $D(y'|y)=\delta_{f(y),y'}$ for some function $f:Y\to Y'$. Let us define $R=R_{X:YY'}$ where $R(x,y,y')=P(x,y)$ whenever $f(y)=y'$ and $R(x,y,y')=0$ otherwise. It immediately follows that $P\approx R$ and $R_{XY'}=Q=PD$. Thus, using the above result, we have
\begin{align}
\Phi(PD)=\Phi(R_{XY'})\geq\Phi(R)=\Phi(P).
\end{align}
Finally, we prove the case of a general row-stochastic matrix $D$. Since $D$ can be expressed as a convex combination of row-stochastic matrices $D_k$ whose entries are either 0 or 1, there exist weights $t_k > 0$ with $\sum_k t_k = 1$
such that
\begin{equation}
    D = \sum_k t_k D_k.
\end{equation}
Our observations thus far, together with the concavity of $\Phi$ imply
\begin{align}\label{eq:apuconv}
\Phi(PD)&=\Phi\left(\sum_k t_kPD_k\right)\geq\sum_k t_k\Phi(PD_k)\geq\sum_k t_k\Phi(P)=\Phi(P).
\end{align}
The other cases are similar.
\end{proof}

A similar argument applies to derivations, which must be concave.

\begin{lemma}
\label{lem: conv derivations}
Let $S$ be the semiring of conditional majorization of Section~\ref{sec: semiring of conditional entropies}. Let $\Delta: \mathcal{S} \to \mathbb{R}$ be a candidate extremal monotone derivation at $\|\cdot\|$, and hence of the form in Proposition~\ref{prop: monotone derivations}.
Then, $\Delta$ is a monotone derivation if and only if it is concave.
\end{lemma}

\begin{proof}
By assumption, $\Delta:\cS\to\mb R$ is a derivation at $\|\cdot\|$ of the form in Proposition~\ref{prop: monotone derivations}. In particular, it satisfies $\Delta(x)=0$ for all $x\in\mb R_+$. Let $P,Q\in\V$ have the same size. The endpoint cases are immediate, so let $t\in(0,1)$. Relation~\eqref{eq:affinechain}, followed by monotonicity and additivity, gives the first inequality below. Apply next the Leibniz rule to the two scaled terms, use that $\Delta$ vanishes on positive scalars, and use the scalar normalization of $\|\cdot\|$. These steps give
\begin{equation}
\Delta\big(tP+(1-t)Q\big)
\geq\Delta(tP)+\Delta\big((1-t)Q\big)
=t\Delta(P)+(1-t)\Delta(Q).
\end{equation}

For the reverse implication, the monotonicity of $\Delta$ can be shown in exactly the same way as in the second part of the proof of Proposition~\ref{prop:convconc}.
\end{proof}

The next lemma shows that, for candidate homomorphisms, monotonicity under arbitrary maps acting on the conditioning system already implies monotonicity under the full class of conditionally mixing channels. This follows from the additional facts that the candidate homomorphisms are monotone under mixing channels and additive under direct sums. Together, these properties suffice to establish monotonicity under general conditionally mixing channels.
\begin{lemma}
\label{lem: extension to general channels}
Let $S$ be the semiring of conditional majorization of Section~\ref{sec: semiring of conditional entropies}.
Let $\Phi: \cS \to \mathbb{K}$ be a candidate monotone homomorphism, and hence of the form in Proposition~\ref{prop: monotone homomorphisms}. If $\Phi$ is monotone under row stochastic maps acting on the conditioning system, then it is monotone under conditionally mixing operations.
\end{lemma}

\begin{proof}
We present the proof for the case $\mathbb{K} = \mathbb{R}_+$, as the other cases follow analogously.

We first note that
monotonicity under mixing channels is guaranteed by the explicit form of
the homomorphisms derived in Proposition~\ref{prop: monotone homomorphisms}, since
these are given by exponential expectations of R\'enyi entropies, which are
monotone under such maps.

To establish monotonicity under conditionally mixing channels, we introduce an auxiliary register that separates the individual terms appearing in the sum defining the channel, and then exploit the direct-sum property of the homomorphisms. The proof then follows by applying monotonicity under the action of doubly stochastic matrices and row-stochastic ones to each term separately.
Explicitly, let $P_{XY}\in\V$, $S^{(i)}$ be doubly stochastic matrices, and $D^{(i)}$ have non-negative entries ($i=1,\ldots,n$) with $D^{(1)}+\cdots+D^{(n)}$ row-stochastic. Let us denote the output of the conditionally mixing channel $Q_{X'Y'}=\sum_i S^{(i)}P_{XY}D^{(i)}$. Let us define new conditioning systems $\widetilde Y$ and
$\widetilde Y'$, both with alphabet $Y'\times\{1,\ldots,n\}$.
Let $\pi_{Y'|\widetilde{Y}'}$ be the row-stochastic matrix which ignores the part $\{1,\ldots,n\}$ in $\widetilde{Y}'$. Let us also define the matrices
\begin{align}
P_{X\widetilde{Y}}&:=\big(P_{XY}D^{(1)}\,\cdots\,P_{XY}D^{(n)}\big)\\
Q_{X'\widetilde{Y}'}&:=\big(S^{(1)}P_{XY}D^{(1)}\,\cdots\,S^{(n)}P_{XY}D^{(n)}\big).
\end{align}
Because of the form of \eqref{eq:FinalForm}, 
we have
\begin{equation}
\Phi(P_{X\widetilde{Y}})=\Phi\left(\bigoplus_{i=1}^n P_{XY}D^{(i)}\right),\quad\Phi(Q_{X'\widetilde{Y}'})=\Phi\left(\bigoplus_{i=1}^n S^{(i)}P_{XY}D^{(i)}\right).
\end{equation}
Putting all this together and noting that $(D^{(1)}\,\cdots D^{(n)})$ is row-stochastic and that $\bigoplus_{i=1}^n S^{(i)}$ is doubly stochastic, we have
\begin{align}
\Phi(P_{XY})&\leq\Phi\big[P_{XY}\big(D^{(1)}\,\cdots\,D^{(n)}\big)\big]\\
&=\Phi(P_{X\widetilde{Y}})=\Phi\left(\bigoplus_{i=1}^n P_{XY}D^{(i)}\right)\\
&\leq\Phi\left[\left(\bigoplus_{i=1}^n S^{(i)}\right)\left(\bigoplus_{i=1}^n P_{XY}D^{(i)}\right)\right]\\
&=\Phi\left(\bigoplus_{i=1}^n S^{(i)}P_{XY}D^{(i)}\right)\\
&=\Phi(Q_{X'\widetilde{Y}'})\\
&\leq\Phi(Q_{X'\widetilde{Y}'}\pi_{Y'|\widetilde{Y}'})\\
&=\Phi(Q_{X'Y'}),
\end{align}
showing the monotonicity in case $\mathbb{R}_+$. The other cases are essentially the same; the sums turn into maxima in the tropical cases $\mathbb{TR}_+$ and $\mathbb{TR}^{\mathrm{op}}_+$, and the inequalities are reversed in cases $\mathbb{R}^{\mathrm{op}}_+$ and $\mathbb{TR}^{\mathrm{op}}_+$.
\end{proof}
\begin{lemma}
\label{lem:extension-to-general-channels-derivations}
An extremal derivation at $\|\cdot\|$ that is monotone under row-stochastic
maps on the conditioning system is monotone under conditionally mixing
operations.
\end{lemma}
\begin{proof}
Repeat the proof of Lemma~\ref{lem: extension to general channels}, using
additivity of the derivation in place of additivity of the homomorphism.
\end{proof}

\end{document}